\documentclass[11pt]{article}
\usepackage{mystyle} 
\graphicspath{{figures/}}

\title{Approximate Majority With Catalytic Inputs}
\date{}
\author{%
  Talley Amir 
        \\ Yale University \\  {\normalsize talley.amir@yale.edu}
  \and James Aspnes\thanks{Supported in part by NSF grant CCF-1650596}
        \\Yale University \\ {\normalsize james.aspnes@gmail.com}
  \and John Lazarsfeld\\Yale University \\ {\normalsize john.lazarsfeld@yale.edu}
}

\usepackage{hyperref}
\hypersetup{
    colorlinks=true,
    citecolor=green,
    urlcolor=cyan,
}
\usepackage[top=1in, bottom=1in, left=1in, right=1in]{geometry}

\newtheorem{theorem}{Theorem}[section]
\newtheorem{corollary}{Corollary}[section]
\newtheorem{lemma}{Lemma}[section]

\begin{document}

\maketitle

\begin{abstract}

Population protocols \cite{AngluinADFP2006} are a class of algorithms for modeling distributed computation in networks of finite-state agents communicating through pairwise interactions. Their suitability for analyzing numerous chemical processes has motivated the adaptation of the original population protocol framework to better model these chemical systems. In this paper, we further the study of two such adaptations in the context of solving approximate majority: persistent-state agents (or \textit{catalysts}) and spontaneous state changes (or \textit{leaks}).

Based on models considered in recent protocols for populations with persistent-state agents \cite{stoc/DudekK18, alistarh2017robust, alistarh2020robust}, we assume a population with $n$ catalytic input agents and $m$ worker agents, and the goal of the worker agents is to compute some predicate over the states of the catalytic inputs. We call this model the Catalytic Input (CI) model. 
For $m = \Theta(n)$, we show that computing the \textit{exact} majority 
of the input population with high probability requires at least $\Omega(n^2)$ total interactions, demonstrating a strong separation between the CI model and the standard population protocol model. 
On the other hand, we show that the simple third-state dynamics \cite{angluin2008simple, perron2009consensus} for \textit{approximate} majority in the standard model can be naturally adapted to the CI model: we present such a constant-state protocol for the CI model that solves approximate majority in $O(n \log n)$ total steps with high probability when the input margin is $\Omega(\sqrt{n \log n})$.

We then show the robustness of third-state dynamics protocols to the transient leaks events introduced by \cite{alistarh2017robust, alistarh2020robust}. In both the original and CI models, these protocols successfully compute approximate majority with high probability in the presence of leaks occurring at each step with probability $\beta \leq O\left(\sqrt{n \log n}/n\right)$. The resilience of these dynamics to leaks exhibits similarities to previous work involving Byzantine agents, and we define and prove a notion of equivalence between the two.

\end{abstract}

\newpage

\section{Introduction}
\label{sec:intro}

The population protocol model \cite{AngluinADFP2006} is
a theoretical framework for analyzing distributed computation in
ad hoc networks of anonymous, mobile agents: at each step,
a random pair of agents is chosen to interact, and
their local states are updated according to
a global transition function.
Population protocols can
solve numerous problems in distributed computing, 
including majority (which is also referred to as consensus)
\cite{angluin2008simple, condon2019, alistarh2020robust}, source detection 
\cite{alistarh2017robust, stoc/DudekK18}, and leader election 
\cite{Alistarh:2015, GS18, GLS+19}.

Population protocols are a special case of chemical reaction networks (CRNs), which are systems of transition rules describing how a set of chemical reactants stochastically transform into a set of products. In particular, population protocols are chemical reaction networks with exactly two reactants which form two products, where each transition rule for a pair of reactants is weighted with probability 1. 
Given their suitability for modeling chemical processes, population protocols have been used
to study computation not only by chemical reaction networks \cite{CDS12}, but also DNA strand displacement \cite{CDS13, TWS2015} and biochemical networks \cite{CC12}. These applications of population protocols in chemistry have inspired various adaptations of the model. In this paper, we focus on two such variations on population protocols in the context of solving \concept{majority}, the problem of determining which of two states is initially more prevalent in a population.


The first modification to the model we consider, which was introduced to the literature in previous works studying source detection and bit-broadcast \cite{stoc/DudekK18, alistarh2017robust} and later studied in the context of the majority problem \cite{alistarh2020robust, dAmore20}, is the presence of \concept{persistent-state agents}, or agents whose state never changes. While some works use persistent-state agents to model authoritative sources of information \cite{stoc/DudekK18} or ``stubborn'' nodes that are unwilling to change state \cite{dAmore20}, others describe these entities as an embodiment of chemical catalysts because they induce a state transition in another agent without themselves changing state \cite{alistarh2017robust, alistarh2020robust}. 
Using the latter perspective, we refer to
these persistent-state agents as 
\concept{catalysts}.

In this work, we call the class of population protocols with catalysts the \concept{catalytic input (CI) model}. We formally define the model to 
consist of $n$ \concept{catalytic input agents},
which in accordance with their name do not ever change state, and $m$
\concept{worker agents} that can change state and wish to compute some function on the states of the 
catalysts. While the CI model is similar to the standard population protocol model, we show that there exists a strong separation between the two in terms of their computational power. 

The next variation on the model we consider is the introduction of 
\textbf{transient leak events}, studied previously in the contexts of solving source
detection and comparison \cite{alistarh2017robust, alistarh2020robust}. In brief, a 
``leak'' simulates the low-probability event that a molecule undergoes a reaction that 
would typically take place in the presence of a catalyst. In population protocols, this is 
modeled by a spontaneous change of state at a single agent, and note that catalytic agents 
in the CI model are not susceptible to leaks because they never change state. A leak 
replaces an interaction between two agents at any given step with some fixed probability, 
known as the \concept{leak rate} \cite{alistarh2017robust}.
Although leaks have typically been studied in the presence of catalysts, we consider
leaks to more generally model unpredictable or adversarial behavior which may occur in the absence of catalysts as well.

We explore the impact of leaks on third-state dynamics \cite{angluin2008simple, perron2009consensus} solving majority. Our work demonstrates that third-state dynamics can solve \concept{approximate majority}, or majority with a lower-bounded initial difference between the counts of the two input states, with upper-bounded leak rate both in the standard and CI population models.

\subsection{Related Work}

The third-state dynamics protocol in the original population model
(sometimes called \textit{undecided-state dynamics}) was introduced
by Anlguin et al. \cite{angluin2008simple} and independently by
Perron et al. \cite{perron2009consensus}.
An agent is either in a state $X$ or $Y$, or in a
\textit{blank} state $B$ (sometimes called an \textit{undecided} state).
The transition rules are shown in Figure~\ref{fig:dbam},
and we refer to this protocol as $\dbam$\footnote{
  $\dbam$ stands for \textit{double-B approximate majority}
  where \textit{double-B} captures the fact that following an
  $X+Y$ interaction, both agents transition to the $B$ state. This protocol is the two-way variant of the original protocol from \cite{angluin2008simple}, which uses one-way communication and where only one agent updates its state per pairwise interaction.
}.
Assuming an initial $X$ majority,
a simplified analysis from Condon et al. \cite{condon2019}
showed that all $n$ agents in the population transition to the
$X$ state within $O(n \log n)$ total interactions with high
probability, so long as the 
input margin $|X| - |Y|$ at the start of the protocol
is at least $\Omega(\sqrt{n \log n})$.
The $\dbam$ protocol is also robust to a small subset of faulty
\textit{Byzantine} agents \cite{angluin2008simple, condon2019},
meaning that all but a $O(\sqrt{n \log n}/n)$ fraction of 
the population still reaches the $X$ state within $O(n \log n)$
interactions with high probability, despite the
presence of these dishonest agents.

The $\dbam$ protocol and similar variants of third-state dynamics
have been shown to more generally compute consensus
(where all agents converge
to either $X$ or $Y$, but where this need not be the initial
majority value), both in the original population protocols
model \cite{angluin2008simple, condon2019} and
in other similar distributed models
\cite{becchetti2020consensus, dAmore20}.
In particular, the closely related results of 
d'Amore et al. \cite{dAmore20} analyzed an analogous version of
the $\dbam$ protocol in the synchronous PULL model.
The authors considered systems with
\textit{stubborn} agents (as in \cite{yildiz2013sutbborn})  which are 
similar to the persistent-state catalytic agents
we consider in the present work. 
However, the parallel synchronous scheduling model considered in \cite{dAmore20}
is fundamentally distinct from the sequential pairwise scheduling used
in population protocols.

The notion of a persistent \textit{source} state in
population protocols originated from \cite{stoc/DudekK18}, where sources are used to solve \concept{detection} (the detection of a source in the population) and \concept{bit broadcast} (the broadcast of a 0 or 1 message from a set of source agents). An accompanying work \cite{alistarh2017robust}  introduces the concept of leaks, or spontaneous state changes, and investigates the detection problem in their presence.
Generally, leaks can be dealt with using error-correcting codes \cite{WangE12182};
however, for certain problems there are more efficient specialized solutions.
For example, Alistarh et al. \cite{alistarh2017robust} demonstrate that detection in the presence of leaks (up to rate $\beta = O(1/n)$) can be solved with high probability using $\log \frac{n}{k} + O(\log \log n)$ states, where $k \leq n$ is the number of sources in the population.

More recently, \cite{alistarh2020robust} examines leaks in the context of the \concept{comparison} problem. Comparison is a generalization of the majority problem, where some possibly small subset of the population is in input state $X_0$ or $Y_0$ and the task of the population is to determine which of the two states is more prevalent. Alistarh et al. \cite{alistarh2020robust} solve comparison in $O(n \log n)$ interactions with high probability using $O(\log n)$ states per agent, assuming $\lvert X_0 \rvert \geq C \lvert Y_0 \rvert$ 
for some constant $C$, and $X_0, Y_0 \ge \Omega(\log n)$. The protocol is self-stabilizing, meaning that it dynamically responds to changes in the counts of input states.

\subsection{Our Contribution}

\begin{figure}[t!]
\vspace*{-0.4cm}
  \centering
  \begin{minipage}[b]{0.45\textwidth}
    \small
    \centering
    \begin{align*}
      X + B &\longrightarrow X + X \\
      Y + B &\longrightarrow Y + Y \\
      X + Y &\longrightarrow B + B
    \end{align*}
    \caption{Transition rules for the $\dbam$ protocol \cite{angluin2008simple}
      in the original population model.}
    \label{fig:dbam}
  \end{minipage}
  \hfill
  \begin{minipage}[b]{0.45\textwidth}
    \small
    \begin{align*}
      X + B &\longrightarrow X + X \;\;
              &I_X &+ B \longrightarrow I_X + X \\ 
      Y + B &\longrightarrow Y + Y \;\;
              &I_Y &+ B \longrightarrow I_Y + Y \\
      X + Y &\longrightarrow B + B
    \end{align*}
    \caption{Transition rules for our $\dbamc$ protocol in the CI model.}
    \label{fig:dbamc}
  \end{minipage}
  \vspace*{-0.4cm}
\end{figure}


In this work, motivated by the recent interest in population models
with catalytic agents and with transient leaks,
we study 
the well-known third-state dynamics protocols
\cite{angluin2008simple, perron2009consensus} for solving approximate majority in the presence of each of these variants separately as well as together.
To begin, 
we formalize the CI model
consisting of $n$ catalysts and $m$
workers, where $N = n + m$. While conceptually similar to other models
considering these types of catalytic agents 
\cite{alistarh2017robust, alistarh2020robust, dAmore20}, introducing 
the distinction between the two (possibly unrelated) population sizes 
provides a new level of generality for designing and analyzing
protocols in this setting, both with and without leaks.

Although the CI and original population models are almost identical,
we show a strong separation between the computational power
of the two. When $m = \Theta(n)$, we prove a lower bound 
showing that \textit{exact} majority (i.e., the majority problem 
on instances with input margin equal to one)
cannot be computed in fewer than $\Omega(n^2)$ interactions
with high probability\footnote{We define ``high probability'' to
mean with probability at least $1-n^{-c}$ where $n$ is the total 
number of agents and $c \ge 1$.} 
in the CI model. On the other hand, exact majority is known to be 
computable in the standard model within $O(n \polylog n)$ total steps 
with high probability as long as each agent has $\Omega(\log n)$ states
\cite{AlistarhAEGR17, AlistarhAG2018, BenNunKKP20}.
While some problems have strictly different lower bounds on running time in these two models, others do not and can in fact be solved using nearly identical techniques. In particular, we show that the \textit{approximate} majority problem can be solved in the CI model by naturally extending the \dbam\ protocol. 

In the approximate majority problem in the CI model, each catalytic input agent holds a persistent value of $I_X$ or $I_Y$ and each worker agent holds either an undecided, or blank value $B$, or an $X$ or $Y$ value corresponding to a belief in an $I_X$ or $I_Y$ input majority, respectively. 
The worker agents all start in state $B$ and 
seek to correctly determine the larger of $|I_X|$ and $|I_Y|$ 
so long as the input margin $||I_X| - |I_Y||$ is sufficiently large.
By adapting the third-state dynamics
process \cite{angluin2008simple} of the original model,
we present a constant-state protocol for approximate majority with
catalytic inputs called \dbamc\ (see Figure \ref{fig:dbamc}).
The protocol converges with high probability in $O(N \log N)$ total
steps when the
initial input margin is $\Omega(\sqrt{N \log N})$ and $m=\Theta(n)$.
We then show that this
input margin is optimal in the CI model up to a $O(\sqrt{\log N})$ factor
when $m = \Theta(n)$. 
Moreover, in the presence of transient leak events,
we show that both the third-state dynamics protocol in the original model
and our adapted protocol in the CI model exhibit a strong robustness
to leaks. When the probability of a leak event 
is bounded, we show that with high probability both protocols still quickly reach a configuration 
where nearly all agents share the correct input majority value.

Notice that the approximate majority problem in the CI model is equivalent to the comparison problem considered by \cite{alistarh2020robust}, so we demonstrate how our protocol compares to the results of this work. We show that our $\dbamc$ protocol converges correctly within the same time complexity of $O(n \log n)$ total steps, while only using  \textit{constant} state space (compared to the logarithmic state used by the protocols in their work). Moreover, in populations where $m = \Theta(n)$, our protocol tolerates a less restrictive bound on the input margin compared to \cite{alistarh2020robust} ($\Omega(\sqrt{n \log n})$ compared to $\Omega(n)$). In the presence of transient leaks, our protocol also shows robustness to a higher leak rate of $\beta \le O(\sqrt{n \log n}/n)$. However, unlike \cite{alistarh2020robust}, our protocol is not self-stabilizing and requires that the number of inputs be at least a constant fraction of the total population for our main results. In order to achieve these results, we leverage the random walk analysis techniques and analysis structure introduced by \cite{condon2019}.

Finally, we compare the impact of leaks on population protocols with that of faulty Byzantine processes. 
While the fast robust approximate majority protocol of \cite{angluin2008simple} is proven to be robust to a number of Byzantine agents that is bounded by the input margin \cite{angluin2008simple, condon2019}, we show that \dbam\ is robust to a similarly bounded leak rate and has sampling error matching the result from \cite{condon2019}.

\ifabs{}{\subsection{Overview}}
The structure of the remainder of the paper is as follows:
in Section~\ref{sec:prelims}
we introduce notation and definitions central to our results. 
Section~\ref{sec:ci-model-lb} presents our lower bounds
over the CI model, which demonstrates the separation between the
CI and original population models. 
In Section~\ref{sec:dbamc}, we analyze the correctness and 
efficiency of the $\dbamc$ protocol for approximate majority in the
CI model, and in Section~\ref{sec:leaks} we demonstrate the 
robustness of both the $\dbamc$ and original $\dbam$ protocols
to transient leak events. 
Then in Section~\ref{sec:byzantine-leaks-equiv}, 
we compare the notion of transient leaks with the adversarial 
Byzantine model, demonstrating parallels between previous results 
examining Byzantine behavior and our work.


\nocite{kosowski2018}
\section{Preliminaries}
\label{sec:prelims}

We begin with some definitions. Denote by $N$ the number of agents in the population. 

\paragraph*{Population Protocols}
\label{subsec:pop-protocols}
Population protocols are a class of algorithms which model interactions between mobile agents with limited communication range. Agents only interact with one another if they are within close enough proximity of each other. In order to model this type of system in an asynchronous setting, interactions between pairs of agents are executed in sequence. The interaction pattern of these agents is dictated by a \concept{scheduler}, which may be random or adversarial. In this work we will assume that the scheduler is uniformly random, meaning that an ordered pair of agents is chosen to interact at each time step independently and uniformly at random from all $N(N-1)$ ordered pairs of agents in the system.

As defined by \cite{AngluinADFP2006} which first introduced the model, a population protocol $\mathcal{P}$ consists of a \concept{state set} $\mathcal{S}=\{s_1, s_2, ..., s_{k}\}$, a \concept{rule set} $\mathcal{R} : \mathcal{S}^2 \mapsto \mathcal{S}^2$, an \concept{output alphabet} $\mathcal{O}$, and an \concept{output function} $f : \mathcal{S} \mapsto \mathcal{O}$. The output function computes the evaluation of some function on the population locally at each agent.
The \concept{configuration} of the population is denoted as a vector $\mathbf{c} = \langle c_1, c_2, ..., c_k \rangle$ such that each $c_i \ge 0$ is equal to the number of agents in the population in state $s_i$, from which it follows that $\sum c_i = N$. For convenience, we denote by $|s_i|$ the number of agents in the population in state $s_i$.

At each point in time, the scheduler chooses an ordered pair of agents $(a_i, a_j)$, where $a_i$ is the \concept{initiator} and $a_j$ is the \concept{responder} \cite{AngluinADFP2006}. The agents interact and update their state according to the corresponding rule in $\mathcal{R}$. In general, a rule in $\mathcal{R}$ is written as $A + B \longrightarrow C + D$ to convey that two agents, an initiator in state $A$ and a responder in state $B$, interact and update their states to be $C$ and $D$, respectively. 
By convention, $N/2$ interactions make one  unit of \concept{parallel time} \cite{angluin2008simple}. This convention is equivalent to assuming every agent interacts once per time unit on average.

An \concept{execution} is the sequence of configurations of a run of the protocol, which \concept{converges} when the population arrives at a configuration $\mathbf{d}$ such that all configurations chronologically after $\mathbf{d}$ have the same output at each agent as those in $\mathbf{d}$ \cite{AngluinADFP2006}. 
In order to determine the success or failure of an execution of $\mathcal{P}$, we will consider a sample of the population to signify the outcome of the protocol \cite{alistarh2017robust}. After the expected time to converge, one agent is selected at random and its state is observed. The output associated with the agent's state is considered the output of the protocol. The probability of sampling an agent whose state does not reflect the desired output of the protocol is called the \concept{sample error rate}. Multiple samples can be aggregated to improve the rate of success. 

\paragraph*{Catalysts and Leaks}
Following~\cite{alistarh2017robust}, in an interaction of the form $A+B \longrightarrow A+D$, we say $A$ \textit{catalyzes} the transformation of the agent in state $B$ to be in state $D$. If $A$ catalyzes every interaction it participates in, $A$ is referred to as a \concept{catalyst}.

In chemistry, a reaction that occurs in the presence of a catalyst also occurs at a lower rate in the absence of that catalyst.
For this reason, recent work in DNA strand displacement, chemical reactions networks, and population protocols \cite{TWS2015,alistarh2017robust, alistarh2020robust} have studied the notion of \textit{leakage}: When a catalytic reaction $A+B \longrightarrow A+D$ is possible, then there is some probability that a transition $B \longrightarrow D$ can occur without interacting with $A$ at all. This type of event, called a \concept{leak}, was introduced in \cite{TWS2015}.

The probability with which the non-catalyzed variation of a reaction takes place is the \concept{leak rate}, which we denote by $\beta$. 
We simulate a leak as follows: At each step with probability $1-\beta$, the scheduler samples an ordered pair of agents to interact with one another as described in the 
beginning of the section; the rest of the time (i.e. with probability $\beta$) one agent is chosen uniformly at random from all possible agents and the \concept{leak function} $\ell: \mathcal{S} \rightarrow \mathcal{S}$ is applied to update this agent's state.  
Note that we only consider \textit{non-catalytic} agents to 
be susceptible to these events. 

\paragraph*{Catalytic Input Model} 
In this work, we formalize a \concept{catalytic input} (CI) model consisting of $n$ catalytic agents that supply the input and $m$ worker agents that perform the computation and produce output. We define $N=m+n$ to be the total number of agents in the population. At each time step, the scheduler samples any two agents in the population to interact with one another. If two catalysts are chosen to interact, then the interaction is considered to be \concept{null} as no nontrivial state transition occurs. When $n=o(m)$, the probability that two catalysts are chosen to interact is upper bounded by a constant, and so the total running time of the protocol is asymptotically equivalent to the number of non-null interactions needed to reach convergence. In the CI model, convergence is a term that refers to the states of the worker agents only, as the catalytic agents never change state. Namely, for the approximate majority problem, successful convergence equates to all \textit{worker agents} being in the majority-accepting state. 
In general, we wish to obtain results that hold with high probability
with respect to the \textit{total} 
number of agents $N$.

\section{Catalytic Input Model Lower Bounds}
\label{sec:ci-model-lb}

In this section, we characterize the
computational power of the CI population protocol model.
Using information-theoretic arguments, we prove the following
two lower bounds over the catalytic input model when the number of input agents
is a constant fraction of the total population:


\begin{restatable}{theorem}{thmemlb}
  \label{thm:em-lb}
  In the catalytic input model with $n$ input agents and $m = \Theta(n)$
  worker agents, any protocol that computes the exact majority of the
  inputs with probability at least $1-N^{-\gamma}$ requires
  at least $\Omega(N^2)$ total steps for any $\gamma \ge 1$.
\end{restatable}

\begin{restatable}{theorem}{thminputmarginlb}
  \label{thm:input-margin-lb}
  In the catalytic input model with $n$ input agents and $m = \Theta(n)$
  worker agents, any protocol that computes the \majority\ of the inputs
  within $O(N \log N)$ total steps requires an input
  margin of at least $\Omega(\sqrt N)$ to be correct with
  probability at least $1 - N^{-\gamma}$ for any $\gamma \ge 1$. 
\end{restatable}

The first result can be viewed as a separation between the CI and
original population models: as mentioned earlier, several works 
\cite{AlistarhAEGR17, AlistarhAG2018, Berenbrink2018APP}
have shown that \textit{exact} majority 
can be computed in the original model within $O(\polylog n)$ 
parallel time with high probability. Thus, our result 
indicates that in the CI model, when the input size is 
a constant fraction of the entire population, 
not all efficiently-computable functions
in the standard model can be solved in sub-linear parallel
time with high probability.
On the other hand, the second result indicates the existence of
a predicate --- approximate majority --- that does not
require a large increase in convergence time to be computed
with high probability in this new model.

\paragraph*{Sampling Catalytic Inputs}
One key characteristic of a CI population is the inability for worker agents to distinguish which
inputs have previously interacted with a worker. Instead, every worker-input interaction acts like a random sample with replacement from the input population. 
For proving lower bounds in this model, this characteristic
of a CI population leads to the following natural argument:
consider a population of $n$ catalytic input agents and a worker population
consisting of a single \concept{super-agent}. Here, we assume the
super-agent has unbounded state and computational power, and it is thus
able to simulate the entire worker population of any
protocol with more workers.
In this simulation, any interaction between a worker and an input agent is equivalent to the super-agent interacting with an input chosen uniformly at random: 
in other words, as a sample with replacement from the
input population. Thus we view the super-agent as running a central
randomized algorithm to simulate the random interactions that 
occur in population protocols. 
If the super-agent needs $S$ samples to compute
some predicate over the inputs with high probability, then so does
any multi-worker protocol in the CI model.
We denote this information-theoretic model as the \concept{Super CI model},
and restate the above argument more formally in the following lemma.

\begin{lemma}
  \label{lemma:super-ci}
  Consider a population with $n$ catalytic input agents
  and a worker population consisting of a single super-agent $W$.
  Let $P$ be a predicate over the input population that requires $S$ total
  interactions between $W$ and the input population in order for $W$ to correctly compute $P$ with probability $\epsilon$.
  Then for a CI population with $n$ catalytic inputs and $m$
  worker agents, computing $P$ correctly with probability $\epsilon$
  requires at least $S$ total interactions.
\end{lemma}

\subsection{Proof of Theorem~\ref{thm:em-lb}}
\label{sec:em-lb-overview}

In a CI model population with $n$ input agents and
$m$ worker agents where $m = \Theta(n)$, Theorem~\ref{thm:em-lb} shows
that computing the exact majority of the inputs requires at least 
$\Omega(n^2) = \Omega(N^2)$ 
total interactions to be correct with high probability.
We prove this by showing that in the Super CI model described in
the previous subsection, a computationally unbounded super-agent $W$
requires at least $\Omega(n^2)$ samples of the input population
to correctly compute exact majority with high probability.
Applying Lemma~\ref{lemma:super-ci} then gives
Theorem~\ref{thm:em-lb}.

\paragraph*{Optimality of the Sample Majority Map}

Recall that in the Super CI model, a predicate over the input population
$C$ is computed by a single super agent worker $W$ with unbounded
computational power. Thus, the output of $W$ can be viewed as a
mapping between a string of input values obtained from interactions
with between $W$ and the input population and the output set $\{0, 1\}$.
We refer to interactions between $W$ and the input population
as \concept{samples} of the input, and for a fixed number of samples $S$,
we refer to $W$'s output as its \concept{strategy}.

First, we show that for some fixed distribution over the input
values of $C$, the strategy that maximizes $W$'s probability
of correctly outputting the majority value of $C$ is simply to
output the majority value of its samples.
Let $I \in \{0, 1\}^S$ be the \concept{sample string} representing the
$S$ independent samples with replacement taken by $W$, 
and let $\mathbb{S}$ denote the set of all $2^S$ possible sample strings. 
We model the population of input agents as being generated by 
an adversary. Specifically, let $M$ denote the majority value (0 or 1) 
of the input population, where we treat $M$ as a a random variable whose 
distribution is unknown. In any realization of $M$, we assume
a fixed fraction $p > 1/2$ of the inputs hold the majority value. 

Given an input population, the objective of the worker agent is 
to correctly determine the value of $M$ through its input sample
string $I$.
By Yao's principle \cite{yao77}, the error of any randomized 
algorithm (i.e., the randomized simulation run by the super-agent) 
on the worst case value of $M$ is no smaller than the
error of the best deterministic algorithm on some fixed
distribution over $M$. So our strategy is to pick a distribution 
over $M$, and to use the the error of the best deterministic strategy 
with respect to this distribution as a lower bound on the worst-case
error of any randomized algorithm used by the super-agent.

Thus, assuming $M$ is chosen according to some fixed distribution,
we model the worker's strategy
as a fixed map $f: \{0, 1\}^S \to \{0, 1\}$. Letting $\mathcal{F}_S$ 
denote the set of all such maps, $W$ then faces the 
following optimization problem: 
$\max_{f \in \mathcal{F}_S} \Pr[f(I) = M]$.
For a given $f \in \mathcal{F}_S$, let $p_f = \Pr[f(I) = M]$, 
and let $\Phi \in \mathcal{F}_S$ denote the map that outputs the
majority value of the input sample string $I$. 
In the following lemma, we show that when the distribution 
over $M$ is uniform, setting $f := \Phi$ \  maximizes $p_f$. 
In other words, to maximize the probability of correctly guessing
the input population majority value, the worker's optimal strategy
is to simply guess the majority value of its $S$ independent samples. 
The proof of the lemma simply uses the definitions 
of conditional probability and the Law of Total Probability to 
obtain the result.

\begin{restatable}{lemma}{majopt}
\label{lemma:ciamlb-maj-opt}
Let $I = \{0, 1\}^S$ be a sample string of size $S$ drawn from an 
input population with majority value $M$ and majority ratio $p$,
and assume $\Pr[M = 1] = \Pr[M = 0] = 1/2$. 
Then 
$
\Pr[\Phi(I) = M] \geq \Pr[f(I) = M]
$
for all maps $f \in \mathcal{F}_S$, where $\Phi$ is the map 
that outputs the majority value of the sample string $I$.
\end{restatable}

\begin{proof}
Recall that we model the majority value of the input population
$M$ as a 0-1 random variable. Assume here that the 
distribution of $M$ is fixed, and that $\Pr[M = 0] = \Pr[M = 1] = 1/2$.


For any map $f \in \mathcal{F}_S$, we can compute $p_f = \Pr[f(I) = M]$ by
\begin{align}
p_f &=
\Pr[f(I) = M]\\
&=
\Pr[f(I) = M, M = 0] + \Pr[f(I) = M, M=1] \\
&= 
\Pr[M = 0] \cdot \Pr[f(I) = M | M = 0] 
+ 
\Pr[M = 1] \cdot \Pr[f(I) = M | M = 1] \\
&= 
\frac{1}{2} \cdot \Pr[f(I) = M | M = 0] 
+ 
\frac{1}{2} \cdot \Pr[f(I) = M | M = 1],
\end{align}
where the last inequality follows from assuming 
$\Pr[M=0] = \Pr[M=1] = 1/2$. 
Recall that $I = \{0, 1\}^S$ is the input string of 
$S$ independent samples from the input population, and
$\mathbb{S}$ is the set of all possible values of $I$.
Thus for any $f \in \mathcal{F}_S$, the law of total probability gives
\begin{align}
\Pr[f(I) = M | M = 0]
&= \sum_{s \in \mathbb{S}} \Pr[f(I) = M, I = s | M = 0] \\
&= \sum_{s \in \mathbb{S}} \Pr[f(s) = M, I = s | M = 0] \\
&= 
\sum_{s \in \mathbb{S}} 
\Pr[f(s) = M | M = 0] \cdot \Pr[I = s | f(s) = M, M = 0]
\label{pr-fi-1}
\end{align}
Since the events $I = s$ and $f(s) = M$ are independent, 
$\Pr[f(s) = M | I = s, M = 0] = \Pr[f(s) = M | M = 0]$ 
for every $s \in \mathbb{S}$. Additionally, given that
every $f \in \mathcal{F}_s$ is a deterministic map, we can rewrite
$$
\Pr[f(s) = M | M = 0] = \Pr[f(s) = 0] = \mathbf{1}_{\{f(s) = 0\}},
$$
where $\mathbf{1}_{\{f(s) = 0\}}$ is the indicator random variable
of the event $f(s) = 0$. 
Thus for any $f \in \mathcal{F}_S$ and every $s \in \mathbb{S}$ we have 
\begin{align}
\Pr[f(s) = M | M = 0]
= 
\sum_{s \in \mathbb{S}} 
\Pr[I = s | M = 0] \cdot \mathbf{1}_{\{f(I) = 0\}}.
\end{align}
It can be similarly shown that
\begin{align}
\Pr[f(s) = M | M = 1]
= 
\sum_{s \in \mathbb{S}} 
\Pr[I = s | M = 1] \cdot \mathbf{1}_{\{f(s) = 1\}}
\end{align}
for every $s \in \mathbb{S}$ and a fixed $f \in \mathcal{F}_S$.
Thus substituting back into \eqref{pr-fi-1} gives
\begin{align}
\Pr[f(I) = M] 
&=
\frac{1}{2}
\sum_{s \in \mathbb{S}}
\left(\;
\Pr[I = s | M = 0] \cdot \mathbf{1}_{\{f(I) = 0\}}
+
\Pr[I = s | M = 1] \cdot \mathbf{1}_{\{f(I) = 1\}}
\;\right).
\label{pr-fi-2}
\end{align}
Now, let $\mathbb{S}_0 \subset \mathbb{S}$ denote the set of 
of input sample strings $\{0, 1\}^S$ with a 0-majority, and let
$\mathbb{S}_1 \subset \mathbb{S}$ denote the set
of sample strings with a 1-majority. Without loss of generality,
assume $\mathbb{S}_0$ and $\mathbb{S}_1$ are disjoint and that
$\mathbb{S}_0 \cup \mathbb{S}_1 = \mathbb{S}$. 
Additionally, for a fixed $f \in \mathcal{F}_S$ and any $s \in \mathbb{S}$ 
define $\alpha(f, s)$ by
$$
\alpha(f, s) 
= 
\Pr(I = s | M = 0) \cdot \mathbf{1}_{\{f(s) = 0\}}
+ 
\Pr(I = s | M = 1) \cdot \mathbf{1}_{\{f(s) = 1\}}.
$$
Thus for a fixed $f \in \mathcal{F}_S$ we can again rewrite
\begin{align}
\Pr[f(I) = M] 
= 
\frac{1}{2}
\left(
\sum_{s \in \mathbb{S}_0} \alpha(f, s) 
+ 
\sum_{s \in \mathbb{S}_1} \alpha(f, s)
\right).
\label{pr-fi-3}
\end{align}

Recall that $\Phi \in \mathcal{F}_S$ is the map that outputs
the majority value of the input sample string $I \in \{0, 1\}^S$. 
Fix any other map $f \neq \Phi \in \mathcal{F}_S$. Since $f \neq \Phi$, 
there exists at least one string $s \in \mathbb{S}$ such that
$f(s) \neq \Phi(s)$, and assume without loss of generality that
$s \in \mathbb{S}_0$. By definition, this means $\Phi(s) = 0$ and
$f(s) = 1$. Using the definition of $\alpha(f, s)$,
and recalling that $\Pr[s_i = 0 | M = 1] = 1-p < 0.5$ and 
$\Pr[s_i = 1 | M = 1] = p > 0.5$ are the probabilities that a single 
sample of $s$ is 0 or 1 respectively, we have
\begin{align}
\alpha(s, f) 
&=
\Pr[I = s | M = 0] \cdot \mathbf{1}_{\{f(s) = 0 \}} +
\Pr[I = s | M = 1] \cdot \mathbf{1}_{\{f(s) = 1 \}} \\
&= 
\Pr[I = s | M = 0] \cdot 0 + 
\Pr[I = s | M = 1] \cdot \mathbf{1}_{\{f(s) = 1 \}} \\
&= 
\Pr[s_i = 0 | M = 1]^k \cdot \Pr[s_i = 1 | M = 1]^{S-k} \\
&=
(1 - p)^k \cdot p^{S-k}
\end{align}
where $S/2 \le k < S$ since $s \in \mathbb{S}_0$. 
Meanwhile, for the same $s \in \mathbb{S}_0$, using the
majority sample map $\Phi$ gives
\begin{align}
\alpha(s, \Phi)
&=
\Pr[I = s | M = 0] \cdot \mathbf{1}_{\{\Phi(s) = 0 \}} + 
\Pr[I = s | M = 1] \cdot \mathbf{1}_{\{\Phi(s) = 1 \}} \\
&=
\Pr[I = s | M = 0] \cdot \mathbf{1}_{\{\Phi(s) = 0 \}} + 
\Pr[I = s | M = 1] \cdot 0 \\
&= 
\Pr[s_i = 0 | M = 0]^k \cdot \Pr[s_i = 1 | m = 0]^{S-k} \\
&=
p^k \cdot (1-p)^{S-k}
\end{align}
where again $S/2 \le k < S$ since $s \in \mathbb{S}_0$. 
Since by definition $p > 1/2$, it follows that
$\alpha(\Phi, s) > \alpha(f, s)$ for any
$f \in \mathbb{F}_S$ where $f \neq \Phi$, and for 
any $s \in \mathbb{S}_0$ where $f(s) \neq \Phi(s)$. It can
similarly be shown that 
$\alpha(\Phi, s) > \alpha(f, s)$ for any $s \in \mathbb{S}_1$
with $f(s) \neq \Phi(s)$. 
By the definition of $\Pr[f(I) = S]$ from \eqref{pr-fi-3}, it follows that
$\Pr[f(I) = M] < \Pr[\Phi(I) = M]$ for any $f \neq \Phi \in \mathcal{F}_S$,
thus proving the claim.
\end{proof}


\paragraph*{Sample Lower Bound for Majority With Input Margin 1}
We have established by Lemma~\ref{lemma:ciamlb-maj-opt} that to correctly
output the input population majority, the super worker agent's
error-minimizing strategy is to output the majority of its $S$ samples
when the distribution over $M$ is uniform.
Now the following lemma shows that when the input margin
of the population is 1, this strategy requires at least $\Omega(n^2)$
samples in order to output the input majority with
probability at least $1 -n^{-c}$ for some constant $c \ge 1$. 
The proof uses a tail bound on 
the Binomial distribution to show the desired trade off between
the error of probability and the requisite number of samples 
needed to achieve this error. 

\begin{restatable}{lemma}{majsamplelb}
  \label{lemma:maj-sample-lb}
  Let $C$ be a Super CI population of $n$ agents with majority value $M$ and
  input margin 1, and consider an input sample string $I = \{0, 1\}^S$
  obtained by a super worker agent $W$. 
  Then for any $c \ge 1$, letting $\Phi(I)$ denote the sample majority of $I$, 
  $
  \Pr[\Phi(I) \neq M] \le n^{-c}
  $
  only holds when $S \ge \Omega(n^2)$. 
\end{restatable}



\begin{proof}
  We will assume $M = 1$ without loss of generality, meaning that
  $p = 1/2 + 1/2n$. Since $\Pr[\Phi(I) \neq M] = \Pr[X_S \le S/2]$,
  we will prove that $S \ge \Omega(n^2)$ is a necessary constraint to
  satisfy $\Pr[X_S \le S/2] \le n^{-c}$.

  Here $\Pr[X_S \le S/2]$ is just the lower tail of the CDF of a
  binomial distribution with parameter $p$. Thus when $p = 1/2 + \delta$
  for $\delta > 0$, we have the following lower bound on $\Pr[X_S \le S/2]$
  (see \cite{ash1990-infotheory}):
  \begin{align}
    \Pr[X_S \le S/2]
    \ge
    \frac{1}{\sqrt{2S}}
    \cdot
    \exp\left(-S \cdot D_{KL}( \tfrac{1}{2} \;||\; p ) \right).
    \label{eq:binom-lb-1}
  \end{align}
  Here, $D_{KL}$ denotes the Kullback-Leibler (KL) divergence between
  a fair coin and a Bernoulli random variable with bias $p$.
  This can be rewritten as
  \begin{align}
    D_{KL}( \tfrac{1}{2} \;||\; p )
    &=
    \frac{1}{2}\cdot\log\frac{1/2}{p} + \frac{1}{2}\cdot\log\frac{1/2}{1-p}\\
    &= 
    \frac{1}{2} \cdot \log \frac{1}{4p\cdot (1-p)} \\
    &\le
      4 \delta^2,
      \label{eq:kl-delta}
  \end{align}
  where the last inequality holds for $0 < \delta \le 1/3$.

  Substituting \eqref{eq:kl-delta} into \eqref{eq:binom-lb-1} then gives
  \begin{align}
    \Pr[X_S \le S/2]
    \ge
    \frac{1}{\sqrt{2S}}
    \cdot
    \exp\left(-4S \delta^2  \right),
    \label{eq:binom-lb-2}
  \end{align}
  and since we are assuming $p = 1/2 + 1/2n$, we have
  \begin{align}
    \Pr[X_S \le S/2]
    \ge
    \frac{1}{\sqrt{2S}}
    \cdot
    \exp\left(-\frac{S}{n^2} \right),
  \end{align}
  where $\delta = 1/2n \le 1/3$ for all $n \ge 2$. 
  Thus to ensure $\Pr[X_S \le S/2] \le n^{-c}$, it is necessary to
  have $(1/\sqrt{2S}) \cdot \exp\left(-S/n^2\right)\le n^{-c}$.
  
  Taking natural logarithms then yields the following constraint on
  $S$:
  \begin{align}
    \frac{S}{n^2}+ 0.5 \log S + 0.5 \ge c \cdot \log n
    \label{eq:s-constr-parity}.
  \end{align}

  We now want to show $S \ge \Omega(n^2)$ is needed to satisfy
  \eqref{eq:s-constr-parity}. To do this, consider any $S' = o(n^2)$.
  Observe then that $S'/n^2 = o(1)$ and $0.5 \log S' < \log n$.
  It follows that
  \begin{align*}
    \frac{S'}{n^2} + 0.5 \log S' + 0.5
    &<
      o(1) + \log n \\
    &\le
      c \log n,
  \end{align*}
  where the final inequality necessarily holds for all $c \ge 1$
  for large enough $n$. 

  Thus no value $S = o(n^2)$ can satisfy the necessary condition of
  \eqref{eq:s-constr-parity}, which means that we must have
  $S \ge \Omega(n^2)$ in order to ensure
  $\Pr[X_S \le S/2] = \Pr[\Phi(X_S) \neq M] \le n^{-c}$ holds
  for any $c \ge 1$. 
\end{proof}



The proof of Theorem~\ref{thm:em-lb}
(which is restated for convenience) follows from Lemmas~\ref{lemma:super-ci},~\ref{lemma:ciamlb-maj-opt}, 
and~\ref{lemma:maj-sample-lb} by invoking Yao's principle. 

\thmemlb*

\begin{proof}
  By Lemmas~\ref{lemma:ciamlb-maj-opt} and~\ref{lemma:maj-sample-lb}
  and using Yao's principle,
  in the Super CI model with an input population $C$ of $n$
  agents and input margin 1,
  a single super-agent worker $W$ can only compute
  the majority value of $C$ with high probability by taking
  at least $S = \Omega(n^2)$ input samples in the worst case.   
  By Lemma~\ref{lemma:super-ci}, this means that in the
  regular CI model with an input population of $n$ agents,
  any protocol for majority with input margin 1 requires
  at least $\Omega(n^2)$ total steps to be computed correctly
  with probability at least $1 - n^{-c}$. 
  When the size of the worker population is $m = \Theta(n)$, this
  means that $N = m + n = \Theta(n)$. Thus for an appropriate 
  choice of $c$, computing exact majority on such populations
  requires at least $\Omega(N^2)$ samples to be correct
  with probability at least $1 - N^{-\gamma}$ for any
  $\gamma \ge 1$. 
\end{proof}


\subsection{Proof of Theorem~\ref{thm:input-margin-lb}}

As mentioned, Theorem~\ref{thm:em-lb} 
implies a strong separation between the CI model and original population model, 
as prior works have
shown that exact majority is computable with
high probability within $O(n \polylog n)$ total steps
in the original model \cite{AlistarhAEGR17, AlistarhAG2018, Berenbrink2018APP}.
Thus, the persistent-state nature of input agents in the CI model
may seem to pose greater challenges than in the original model
for computing predicates quickly with high probability.
However, using the same sampling-based lower bound techniques
developed in the preceding section, Theorem~\ref{thm:input-margin-lb}
shows that when $m = \Theta(n)$, and 
when restricted only to $S = O(n \log n)$ total steps, 
any protocol computing majority in the CI model
requires an input margin of at least $\Omega(\sqrt{n}) = \Omega(\sqrt N)$ 
to be correct with high probability in $N$. 

Moreover, in Section~\ref{sec:dbamc}
we present a protocol for approximate majority 
in the CI model that converges correctly 
with high probability within $O(N \log N)$ total steps, 
so long as the initial input margin is $\Omega(\sqrt{N \log N})$. 
Thus, the existence of such a protocol indicates that the 
$\Omega(\sqrt{N})$ lower bound on the input margin is nearly tight 
(up to $\sqrt{\log N}$ factors) for protocols limited 
to $O(N \log N)$ total steps when $m = \Theta(n)$. 

\paragraph*{Input Margin Lower Bound for Majority}
We return to the Super CI model and use the same notation developed
in Section~\ref{sec:em-lb-overview}. 
We want to show that when $S = O(n \log n)$, we must have
$p \ge 1/2 + \Omega(1/\sqrt{n})$ in order to ensure
$\Pr[\Phi(I) \neq M] = \Pr[X_S \le S/2] \le n^{-c}$.
This lower bound on $p$ (the proportion of 1-agents, wlog, in the
input population) corresponds to an input margin
lower bound of $\Omega(\sqrt{n})$.

\begin{lemma}
  \label{lemma:maj-input-lb}
  Assume a 0-1 population of $n$ agents with majority value $M$ and
  majority proportion $p = 1/2 + \delta$, 
  and consider an input sample string $I = \{0, 1\}^S$ where
  $S = O(n \log n)$. 
  Then for any $c \ge 1$,
  $ \Pr[\Phi(I) \neq M] \le n^{-c}$
  only holds when $\delta \ge \Omega(1/\sqrt{n})$, where $\Phi$ is
  the map that outputs the majority value of $I$. 
\end{lemma}

\begin{proof}
  Again wlog assume $M = 1$. Since $\Pr[\Phi(I) \neq M] = \Pr[X_S \le S/2]$,
  we will show that $\delta \ge \Omega(1/\sqrt{n})$ is a necessary condition
  to have $\Pr[X \le S/2] \le n^{-c}$ when $S = O(n \log n)$.

  Recall the lower bound on $\Pr[X_S \le S/2]$ from Lemma~\ref{lemma:maj-sample-lb}:
  \begin{align}
    \Pr[X_S \le S/2]
    \ge
    \frac{1}{\sqrt{2S}} \cdot \exp\left(- 4S\delta^2 \right),
    \label{eq:binom-lb-3}
  \end{align}
  which holds for $0 < \delta \le 1/3$. (Note that when $\delta = 1/\sqrt{n}$,
  $\delta \le 1/3$ for all $n \ge 9$). 
  
  Setting $S = an \log n$ for some $a > 0$ lets us write
  \eqref{eq:binom-lb-3} as
  \begin{align}
    \Pr[X_S \le S/2]
    \ge
    \frac{1}{\sqrt{2an \log n}}
    \cdot \exp\left(- 4an\log n \cdot \delta^2 \right),
  \end{align}
  which means it is necessary to have
  $(1/\sqrt{2an\log n})\cdot \exp(-4an \log n \cdot \delta^2) \le n^{-c}$ to
  ensure that $\Pr[X_S \le S/2] \le n^{-c}$. 

  Again by taking natural logarithms, we find that we require
  \begin{align}
    4an \log n \cdot \delta^2 + 0.5 \log n + 0.5 \log \log n + 0.5\log a + 0.5
    \ge c \log n.
    \label{eq:s-constr-input}
  \end{align}
  To show that $\delta \ge \Omega(1/\sqrt{n})$ is needed to satisfy
  \eqref{eq:s-constr-input}, we use a similar strategy as in
  Lemma~\ref{lemma:maj-sample-lb} and consider any $\delta' = o(1/\sqrt{n})$.
  This would imply $4an \log n\cdot (\delta')^2 = o(4a\log n) = o(\log n)$,
  and since $0.5 \log \log n = o(\log n)$ and $a > 0$ is a constant, we have
  \begin{align}
    4an\log n \cdot (\delta')^2 + 0.5 \log n + 0.5 \log \log n + 0.5\log a + 0.5
    &=
      0.5\log n + o(\log n)  \\
    &<
      c\log n,
  \end{align}
  where the final equality will hold for all $c \ge 1$ and large enough $n$. 

  Thus if $\delta = o(1/\sqrt{n})$, then the necessary condition
  \eqref{eq:s-constr-input} will be violated, meaning that 
  that we must have $\delta = \Omega(1/\sqrt{n})$ to ensure
  $\Pr[X_S \le S/2] \le n^{-c}$ holds for any $c \ge 1$ when $S = O(n \log n)$.
  Since we defined $p = 1/2 + \delta$, this corresponds to requiring an
  input margin of at least $\Omega(\sqrt{n})$ when $S = O(n \log n)$. 
\end{proof}

We now formally prove Theorem~\ref{thm:input-margin-lb}, which is 
restated for convenience.

\thminputmarginlb*

\begin{proof}
  By Lemmas~\ref{lemma:ciamlb-maj-opt} and ~\ref{lemma:maj-input-lb},
  in the Super CI model with an input population $C$ of size $n$,
  computing the majority of inputs correctly in $S = O(n \log n)$ samples 
  with probability at least $1 - n^{-c}$ requires an input margin of at 
  least $\Omega(\sqrt{n})$.
  By an argument similar to Lemma~\ref{lemma:super-ci} and Theorem~\ref{thm:em-lb},
  note that this implies that any protocol that computes $\majority$
  in the regular CI model within $O(n \log n)$ total steps also requires
  an input margin of $\Omega(\sqrt{n})$ to be correct with probability
  at least $1 - n^{-c}$. 
  Now consider that the size of the worker population is $m = \Theta(n)$,
  which means that $N = m + n = \Theta(n)$.
  This implies that for an appropriate choice of constant $c$ and taking
  only $O(N \log N)$ total steps, the input margin must be at least 
  $\Omega(\sqrt{N})$ in order for $\majority$ to be computed correctly
  with probability at least $1 - N^{-\gamma}$ for any $\gamma \ge 1$.
\end{proof}


 
\section{Approximate Majority with Catalytic Inputs} 
\label{sec:dbamc}

We now present and analyze the $\dbamc$ protocol for computing approximate
majority in the CI model. 
The protocol is a natural adaptation of the third-state dynamics from the 
original model, where we now account for the behavior of $n$
catalytic input agents and $m$ worker agents. 
Using the CI model notation introduced in Section~\ref{sec:prelims}, 
we consider a population with $N = n + m$ total agents.
Each input agent begins (and remains) in 
state $I_X$ or $I_Y$, and we assume each worker agent begins in 
a \textit{blank} state $B$, but may transition to states $X$ or $Y$
according to the transition rules found in Figure~\ref{fig:dbam}. 
Letting $i_X$ and $i_Y$ (and similarly $x, y$ and $b$) be random variables 
denoting the number of agents in states $I_X$ and $I_Y$ (and 
respectively $X, Y$, and $B$), we denote the \textit{input margin}
of the population by $\epsilon = |i_X - i_Y|$.
Throughout the section, we assume without loss of generality
that $i_X > i_Y$. 

Intuitively, an undecided (blank) worker agent
adopts the state of a decided agent (either an input or worker), 
but decided workers only revert back to a blank state upon 
interactions with other workers of the opposite opinion. 
Thus the protocol shares the opinion-spreading 
behavior of the original $\dbam$ protocol, but note that the 
inability for decided worker agents to revert back to 
the blank state upon subsequent interactions with an input
allows the protocol to converge to a configuration where 
all workers share the same $X$ or $Y$ opinion.

\paragraph*{Main Result}
The main result of the section characterizes the convergence 
behavior of the $\dbamc$ protocol when the input margin $\epsilon$ 
is sufficiently large. 
Recall that we say the protocol \textit{correctly computes} 
the $\majority$ of the inputs if we reach a configuration where
$x = m$. The following theorem shows that, subject to
mild constraints on the population sizes, when the
input margin is $\Omega(\sqrt{N \log N})$, the protocol 
correctly computes the majority value of the inputs
in roughly logarithmic parallel time with high probability. \\

\begin{restatable}{theorem}{thmdbamcoverall}
\label{thm:dbamc-overall}
There exists some constant $\alpha \ge 1$ such that, for
a population of $n$ inputs, $m$ workers, and initial input margin 
$\epsilon \ge \alpha \sqrt{N \log N}$, the 
$\dbamc$ protocol correctly computes the majority
value of the inputs within 
$O\left(\frac{N^4}{m^3} \log N\right)$ total interactions with 
probability at least $1- N^{-c}$ for any $c \ge 1$ 
when $m \ge n/10$ and $N$ is sufficiently large.
\end{restatable}

Because the CI model allows for distinct (and possibly
unrelated) input and worker population sizes, we aim to 
characterize all error and success probabilities with respect
to the \textit{total} population size $N$.
The analysis in the proof of Theorem~\ref{thm:dbamc-overall}
characterizes the convergence behavior of the protocol 
in terms of both population sizes $m$ and $n$, and thus
the convergence time of $O((N^4/m^3) \log N)$ is not 
always equivalent to $O(N \log N)$.
On the other hand, in the case
when $m = \Theta(n)$ --- which is an assumption 
used to provide lower bounds over the CI model from 
Section~\ref{sec:ci-model-lb} --- we have as
a corollary (stated further below) that the protocol
correctly computes the majority of the inputs within
$O(N \log N)$ total steps with probability
at least $1- N^{-\alpha}$. 

\subsection{Analysis Overview}
The proof of the main result leverages and applies 
the random walk tools from \cite{condon2019} (in 
their analysis of the original $\dbam$ protocol) 
to the $\dbamc$ protocol. 
Given the uniformly-random behavior of the interaction
scheduler, the random variables $x, y$ and $b$ (which
represent the count of $X$, $Y$, and $B$ worker agents
in the population) each behave according to some
one-dimensional random walk, where the biases in the
walks change dynamically as the values of these
random variables fluctuate.
Based on the coupling principle that an upper bound on 
the number of steps for a random walk with success 
probability $p$ to reach a certain position is an 
\textit{upper bound} on the step requirement for a second random 
walk with probability $\hat p \ge p$ to reach 
the same position, we make use of several
\textit{progress measures} that give the behavior 
of the protocol a natural structure. 
As used in the analysis of Condon et al. \cite{condon2019},
we define $\hat x = x + b/2$, $\hat y = y + b/2$, and 
$P = \epsilon + \hat x - \hat y$. 
It can be easily seen that $\hat x + \hat y = m$
will hold throughout the protocol. On the other hand,
the progress measure $P$ captures the collective 
gap between the majority and non-majority opinions in 
the population.
Observe that the protocol has 
correctly computed the input majority value when 
$P = \epsilon + m$ and $\hat y = 0$. 

%
%
Now, similar again to the analysis of \cite{condon2019},
we define the following \textbf{Phases} and \textbf{Stages}
of the $\dbamc$ protocol. 
Intuitively, every correctly-completed stage of 
Phase 1 results in the progress measure $P$ doubling,
and every correctly-completed stage of Phase 2 results in 
the progress measure $\hat y$ decreasing by a factor of two. 
Described more formally:

\begin{enumerate}
\item
\textbf{Phase 1} of the protocol starts with $P = \epsilon$
and completes \textit{correctly} once $P \ge \epsilon + 7m/8$. 
Each \textbf{stage of Phase 1} begins with 
$P = 2^{t} \cdot \epsilon$ and completes correctly once
$P = 2^{t+1} \cdot \epsilon$ or when $P \ge \epsilon + 7m/8$, 
where $t \in \{0, \dots, O(\log N)\}$. 
\item
\textbf{Phase 2} of the protocol starts with 
$\hat y = 7m/8$ (equivalent to $P \ge \epsilon + 7m/8$)
and completes correctly once $\hat y \le \alpha \log m$.
Each \textbf{stage of Phase 2} begins with 
$\hat y = 2^{-s} \cdot m/16$ and completes correctly
once $\hat y \le 2^{-(s+1)} \cdot m/16$ or when
$\hat y \le \alpha \log m$, where
$s \in \{0, \dots, O(\log N)\}$.  
\item
\textbf{Phase 3} of the protocol starts with 
$\hat y = \alpha \log m$ and completes correctly
once $\hat y = 0$. 
\end{enumerate}

Note that among the protocol's non-null transitions 
(see Figure~\ref{fig:dbamc}), only the
interactions $I_X + B$, $I_Y + B$, $X + B$, and $Y + B$ 
change the value of either progress measure. For this
reason, we refer to the set of non-null transitions 
(which includes $X+Y$ interactions) as \concept{productive}
steps, and the subset of interactions that change our
progress measures as the set of \concept{blank-consuming}
productive steps (sometimes referred to as productive-b steps).
The analysis strategy for every phase and stage is to employ
a combination of standard Chernoff bounds and
martingale techniques (in general, see \cite{gs-book} 
and \cite{feller-vol-1}) to
obtain with-high-probability estimates of 
(1) the number of productive steps needed to complete each 
phase/stage correctly, 
and (2) the number of total steps needed to obtain the
productive step requirements. 
We call these two steps \textit{correctness} and 
\textit{efficiency}, respectively. 
Given an input margin that is sufficiently large,
and also assuming a
population where the number of worker agents is at least
a small constant fraction of the input size, we can
then sum over the error probabilities of each phase/stage
and apply a union bound to yield the final 
result of Theorem~\ref{thm:dbamc-overall}.

While the $\dbamc$ protocol is conceptually
similar to the original $\dbam$ protocol, the presence
of persistent-state catalysts whose opinions 
never change requires a careful analysis of the convergence 
behavior. Moreover, simulation results presented
in Section~\ref{sec:leaks} show interesting differences
in the evolution of the protocol for varying population sizes.

%
%
\subsection{Proof of Theorem~\ref{thm:dbamc-overall}}
In this section we develop the 
tools used to prove
Theorem~\ref{thm:dbamc-overall}. 
To begin, we state the following standard 
probabilistic tools used throughout the analysis: 
absorption probabilities for one-dimensional random walks, 
and standard upper and lower Chernoff bounds. 

\begin{lemma}
\label{lemma:one-dim-rw} \cite{feller-vol-1}
If we run an arbitrarily long sequence
of independent trials, each with success probability at least $p>1/2$, then the probability that the number of failures ever
exceeds the number of successes by $b$ is at most $\left(\frac{1-p}{p}\right)^b$.
\end{lemma}

\begin{lemma}
\label{lemma:chernoff-tail-bounds} \cite{chernoff1952}
If we run $N$ independent Bernoulli trials, each with success probability $p$, 
then the number of successes $S_N$ has expected value $\mu=Np$,
and for $0 < \delta < 1$, $\mathbb{P}[S_N \leq (1-\delta) \mu] \leq \exp(- \frac{\delta^2 \mu}{2})$, and $\mathbb{P}[S_N \geq (1+\delta) \mu] \leq \exp(- \frac{\delta^2 \mu}{3})$.
\end{lemma}
The following subsections proceed to prove the correctness
and efficiency of the stages and phases of the protocol.

\subsubsection{Phase 1: Blank-Consuming Step Bounds}

For a population with an initial input margin $\epsilon \ge \alpha \sqrt{N \log N}$,
the following lemma gives an upper bound on the number of productive-b (blank-consuming)
steps needed to complete each stage of Phase 1 correctly. Recall that each stage of Phase
1 completes correctly when the progress measure $P = \epsilon + \hat x - \hat y$ 
doubles from its initial value. 

%
\begin{lemma}
\label{lemma:doubling-difference}

During Phase 1 of the $\dbamc$ protocol on a population with 
input margin $\epsilon \ge \alpha \sqrt{N \log N}$ for some $\alpha \ge 1$, 
starting at $P = \Delta$, within $100N$ productive-b steps 
$P$ will increase to $\min\{2P \;,\; 7m/8 + \epsilon\}$  with
probability at least $1 - N^{-(\alpha/6)} - N^{-(\alpha/2)}$.
\end{lemma}

\begin{proof}
Observe that from the time at which $P=\Delta$ until the point 
(if ever) $P$ decreases below $\frac{1}{2}\Delta$, $P$ is strictly greater 
than $\frac{1}{2}\Delta$. Therefore, until $P$ reaches $\frac{1}{2}\Delta$, 
the probability of a successful interaction (conditioned on having
a blank-consuming interaction) is at least 
\begin{align}
    p = \frac{bi_X+bx}{bi_X+bx+bi_Y+by} 
    = \frac{i_X+x}{i_X+x+i_Y+y}
    &= \frac{P+i_Y+y}{P+2(i_Y+y)} \\
    &=\frac{1}{2}+\frac{P}{4(i_Y+y)+2P}
    > \frac{1}{2} + \frac{\Delta}{10N}
\end{align}
where the final inequality holds because $P > \frac{1}{2}\Delta$, $i_Y+y < N$, 
and $\Delta < N$.
The change in $P$ can thus be viewed as a biased random walk starting at 
$\Delta$ with success probability $p$. 

Now, by Lemma \ref{lemma:one-dim-rw}, starting at $P=\Delta$, 
the probability of ever having an excess margin of $\frac{1}{2}\Delta$ 
between  $B+Y$ or $B+I_Y$ steps to $B+X$ or $B+I_X$ steps is at most
\begin{align}
    q=
    \left(\frac{1-p}{p}\right)^{\Delta/2} 
    = \left(\frac{5N-\Delta}{5N+\Delta}\right)^{\Delta/2} 
    &= \left(1-\frac{2\Delta}{5N+\Delta}\right)^{\Delta/2} \\
    &\le \exp(-\Delta^2/(5N+\Delta)) \\
    &\le N^{-(\alpha/6)}
\end{align}
where the final inequality holds given the assumption that 
$\epsilon \ge \alpha\sqrt{N\log N}$. 
Thus with high probability, $P$ will never drop below $\frac{1}{2}\Delta$
when starting initially from $\Delta$. 

Now in a sequence of $100N$ productive-b steps, 
in order for $P$ to reach $\min\{2P, 7m/8 + \epsilon\}$, 
it is sufficient to ensure that the number of $B+X$ or $B+I_X$ steps
within the sequence (which we denote by $S_x$) exceeds the number 
of $B+Y$ or $B+I_Y$ steps within the sequence (which we denote by $S_y$) 
by at least $\Delta$. 

Assuming that $P \ge \Delta/2$ holds, we can see that in expectation over
the sequence of $100N$ productive-b steps that $\E[S_X] = 50N + 10\Delta$.
As long as $S_X \ge 50 N + 0.5\Delta$ it follows that $S_X - S_Y \ge \Delta$,
and thus applying an upper Chernoff bound shows that 
\begin{align}
    \Pr[S_X < (50N + 0.5\Delta) ] 
    &= 
    \Pr\left[\; S_X < (1 - (9.5\Delta/(50N + 10\Delta)))\cdot \E[S_X]\; \right] \\
    &\le
    \exp\left(
        - \frac{1}{2} \cdot \frac{9.5^2 \Delta^2}{50N + 10\Delta} 
    \right)\\
    &\le 
    \exp\left(- \frac{1}{2} \cdot \alpha \log N\right)
    \le N^{-(\alpha/2)},
\end{align}
where again the penultimate inequality is due to the 
assumption that $\Delta \ge \epsilon \ge \alpha \sqrt{N \log N}$. 

Now, summing over all error probabilities and taking a union bound
shows that $P$ will increase to $\min\{2P, 7m/8 + \epsilon\}$ 
within $100N$ productive-b steps with probability 
at least $1 - N^{-(\alpha/6)} - N^{-(\alpha/2)}$.
\end{proof}

\subsubsection{Phases 2 and 3: Blank-Consuming Step Bounds}
The following lemma gives analogous bounds on the number of 
productive-b steps needed to complete stages of Phase 2, and Phase 3,
correctly with high probability.
Recall that each stage of Phase 2 of the protocol begins with 
$\hat y = m/k$, where $16 \le k \le m/(\alpha \log m)$, and ends
correctly when $\hat y$ decreases by a factor of 2 from its original
value (similarly, Phase 3 starts with $\hat y = \alpha \log m$ 
and ends correctly once $\hat y$ reaches 0).
Note that the following lemma proves a slightly stronger result by
showing the number of productive-b steps needed to bring $\hat y$ to $0$,
which will always be an upper bound on the number of steps needed
to complete a stage of Phase 2, or Phase 3 correctly.

\begin{lemma}
\label{lemma:dbamc23-corr-mk}
Say $\hat y = \frac{m}{k}$ for $16 \le k \le m/(\alpha \log m)$
during Phase 2 of the $\dbamc$ protocol on a population with 
input margin $\epsilon \ge \alpha \sqrt{N \log N}$ for some $\alpha \ge 1$.
Assuming that $\hat y$ remains below $2m/k$, then after at most
$480 \tfrac{N^2}{m^2}(\alpha \log N + \tfrac{m}{k})$
productive-b steps, $\hat y$ goes to 0 with probability at least 
$1 - N^{-\alpha}$ for $\alpha \ge 1$. 
\end{lemma}

\begin{proof}
Recall from Lemma~\ref{lemma:doubling-difference} that throughout the
execution of $\dbamc$ on a population with 
input margin $\epsilon \ge \alpha \sqrt{N \log N}$, 
the probability of a $B+X$ or $B+I_X$ step 
(conditioned on having a blank-consuming step) is bounded
from below by $\frac{1}{2} + \frac{\Delta}{10N}$, where
$\Delta = \epsilon + \hat x - \hat y$.  
This lower bound holds as long as 
$\Delta$ never drops below $\epsilon / 2$, which is ensured
given the starting conditions of Phase 2 and by the 
assumption that $\hat y \le 2m/k$. 

Now, given that $x + y + b = m$ is invariant 
throughout the execution, $\Delta$ can be rewritten as
$\Delta = m - 2\yhat + \epsilon$. 
Also, since we assume $\yhat$ starts
at $\frac{m}{k}$ and never exceeds $\frac{2m}{k}$ with high probability, 
it follows that
\begin{align}
\Delta = m - 2\yhat + \epsilon
\ge m - \tfrac{4m}{k} + \epsilon
\tfrac{m(k-4)}{k} + \epsilon.
\end{align}
Denoting $B+I_X$ and $B+X$ steps as \textit{succeeding}
and $B+I_Y$ and $B+Y$ steps as \textit{failing}, 
then the probability of a succeeding, productive-b interaction
(conditioned on a blank-consuming step) can be rewritten as
\begin{align}
p = \frac{1}{2} + \frac{\Delta}{10N}
\ge \frac{1}{2} + \frac{(k-4)}{10k} \cdot \frac{m}{N}
\end{align}
Now, consider a sequence of $\lambda$ productive-b interactions. 
Note that succeeding productive-b interactions 
each increase $\yhat$ by 1/2, whereas the remaining failing
productive-b interactions decrease $\yhat$ by 1/2. 
Letting $S_x$ and $S_y$ denote the number of 
succeeding and failing steps among the $\lambda$ total prod-b 
interactions respectively, observe that $S_x$ must exceed
$S_y$ by $(2m)/k$ in order for $\yhat$ to decrease to 0
within the sequence.
This means that having $S_x \ge  (\lambda k + 2m)/(2k)$ is sufficient
to ensure that $\yhat$ decreases to $0$. 

In expectation, we have $\E[S_x] = \lambda p$, and thus $\yhat$ will reach 0 
as long as $S_x$ is no more than $\lambda p - \frac{\lambda k + 2m}{2k}$ 
smaller than its expected value. Using a Chernoff lower tail 
bound and setting $\delta = \frac{\lambda k (2p -1) - 2m}{2k\lambda p}$, 
the probability of $\yhat$ failing to reach 0 can thus be bounded by
\begin{align}
p_f = \Pr \left[ \text{$\yhat > 0$ after $\lambda$ prod-b steps} \right] 
&= 
\Pr \left[\; S_x \le \tfrac{\lambda k + 2m}{2k} \; \right] \\
&= 
\Pr\left[\; S_x \le (1 - \delta) \cdot \E[S_x] \; \right]  \\
&\le 
\exp\left(- \tfrac{\delta^2}{2}\cdot \E[S_x] \right)
= 
\exp \left(- \tfrac{(\lambda k(2p-1) - 2m)^2}{8k^2 \lambda p} \right) .
\label{pf-error}
\end{align}

The upper bound on $p_f$ can be slightly exaggerated by observing
\begin{align}
\exp \left(- \tfrac{(\lambda k(2p-1) - 2m)^2}{8k^2 \lambda p} \right)
&= 
\exp \left( - 
\left(
\tfrac{\lambda}{8p} \cdot (2p-1)^2 - 
\tfrac{(2p-1)m}{2pk} + 
\tfrac{m^2}{2\lambda k^2 p}
\right)
\right) \\
&\le 
\exp \left( - 
\left(
\tfrac{\lambda (2p-1)^2}{8p} - \tfrac{m}{k}
\right)
\right),
\label{pf-error-ub}
\end{align}
where the last inequality follows from $\tfrac{2p-1}{p} \le 1$ for 
$p \le 1$, and from $\tfrac{m^2}{2\lambda k^2 p} \ge 0$. 

To ensure that $p_f$ is no larger than $N^{-\alpha}$, it is then sufficient
to find $\lambda$ such that \eqref{pf-error-ub} is bounded from above by
$\exp(- \alpha \log N)$. This holds when 
\begin{align}
\frac{\alpha(2p-1)^2}{8p} - \frac{m}{k} \ge \alpha \log N 
\;\;\iff\;\;
\alpha \ge \frac{8p}{(2p-1)^2} \cdot \left( \alpha \log N + \frac{m}{k} \right).
\label{lambda-cond}
\end{align}
The term $\frac{8p}{(2p-1)^2}$ is decreasing in $p$ for $p > 1/2$, 
and since we have 
$
p \ge \tfrac{1}{2} + \tfrac{(k-4)}{10k}\cdot\tfrac{m}{N} =
\tfrac{5kN + (k-4)m}{10kN}
$,
we see that
\begin{align*}
\frac{8p}{(2p-1)^2} 
&\le 
\frac{40kN + 8(k-4)m}{10kN} \cdot \frac{25k^2 N^2}{(k-4)^2 m^2} \\
&= 
\frac{(100kN + 20(k-4)m) \cdot kN}{(k-4)^2 m^2} \\
&\le 
\frac{120 k^2 N^2}{(k-4)^2 m^2} \\
&\le 
\frac{120 k^2 N^2}{(\tfrac{k}{2})^2 m^2} \\
&= 
480 \cdot \frac{N^2}{m^2},
\end{align*}
where the last inequality holds given $\tfrac{k}{2} \le k - 4$
when $k \ge 8$, which is satisfied by our assumption that
$16 \le k \le \tfrac{m}{\alpha \log m}$. 

By this upper bound on $\tfrac{8p}{(2p-1)^2}$ and the
condition on $\lambda$ from \eqref{lambda-cond}, it 
follows that setting 
$$
\lambda = 480 \cdot \frac{N^2}{m^2}
\left( \alpha \log N + \frac{m}{k} \right)
$$
is a sufficient number of productive-b interactions to guarantee that 
$\yhat$ decreases to 0 with probability at least $1 - N^{-\alpha}$. 
\end{proof}

Lemma~\ref{lemma:dbamc23-corr-mk} relies on the assumption that
$\yhat$ starts at $m/k$ and never exceeds $(2m/k)$
(for $16 \le k \le \tfrac{m}{\alpha \log m}$). In the following lemma,
we prove this latter assumption holds with high probability, subject
to a mild set of constraints on the sizes of the population.

\begin{lemma}
\label{lemma:dbamc23-corr-mk-nodouble}
Assuming $m \ge n/10$ and $m \ge 8 \cdot \alpha \log_{12/11} 11m$
for $\alpha \ge 1$, 
if $\yhat$ reaches $m/k$ during the execution of $\dbamc$ for
$16 \le k \le \tfrac{m}{\alpha \log m}$, then $\yhat$ will 
never exceed $2m/k$ with probability at least $1 - N^{-\alpha}$. 
\end{lemma}

\begin{proof}
As before, let $p$ denote the probability of a successful prod-b step, 
conditioned on a prod-b step. Recall that
$p = \frac{i_x + x}{i_x + i_y + x + y}$, which, up until the moment that
$\yhat > 4m/k$ can be rewritten as
\begin{align}
p 
= 
\frac{i_x + x}{i_x + i_y + x + y}
&\ge 
\frac{N - (i_y + y + b)}{N} \\
&\ge
1 - \frac{i_y + 2\yhat}{N} \\
&\ge
1 - \frac{\tfrac{n}{2} + \tfrac{4m}{k}}{N} \\
&\ge
1 - \frac{2n + m}{4N} \\
\end{align}
where the final inequality holds since $16 \le k$,
and we can show that $p \ge \frac{12}{23}$ by observing that
\begin{align}
\frac{2n + m}{4N} 
= 
\frac{2n+m}{4n + 4m}
\le 
\frac{11}{23}
\;\; \iff \;\;
2n \le 21 m,
\end{align}
which is satisfied by the assumption that $m \ge 10n$. 

If initially $\yhat = 2m/k$, then $\yhat$ only exceeds $2m/k$
if the number of failing prod-b steps exceeds the number of succeeding
prod-b steps by $2m/k$, since each prod-b step changes $\yhat$ in 
magnitude by $1/2$. 

By Lemma~\ref{lemma:one-dim-rw}, this event
occurs with probability at most 
\begin{align}
\left(
\frac{1-p}{p}
\right)^{2m/k}
\le
\left(
\frac{11}{12}
\right) ^{2m/k}
\le
\left(
\frac{11}{12}
\right)^{m/8},
\end{align}
where the first inequality follows from the bound 
$p \ge 12/23$, and the second from assuming $k \ge 16$. 
We can then observe that
\begin{align}
\left(
\frac{11}{12}
\right)^{m/8}
\le 
N^{-\alpha}
\;\;\iff\;\;
m
\ge
8\alpha \cdot \log_{12/11} N,
\end{align}
where we note that $8 \alpha \cdot \log_{12/11} N$ is at 
most  $8 \alpha \cdot \log_{12/11}(11m)$ by the assumption 
that $m \ge n/10$.
Thus assuming we have $m \ge 8 \alpha \log_{12/11} (11m)$ and $m \ge n/10$, 
the probability  that $\yhat$ ever exceeds $2m/k$ is bounded from 
above by $N^{-\alpha}$.
\end{proof}

Using Lemmas~\ref{lemma:dbamc23-corr-mk} and 
\ref{lemma:dbamc23-corr-mk-nodouble} and taking a union bound, 
we have the following upper bounds on the number of productive-b steps 
needed to complete each stage of Phases 2, and Phase 3 of the 
protocol correctly: 
\begin{corollary}
\label{corr:dbamc23-prodb-steps}
During the $\dbamc$ protocol on a population with initial input 
margin $\epsilon \ge \alpha \sqrt{N \log N}$ for some $\alpha \ge 1$, 
assuming that $m \ge n/10$ and $m \ge 8 \alpha \cdot \log(11m)$, 
then with probability at least $1- 2N^{-\alpha}$:
\begin{itemize}
\item
At most 
$
480 \cdot
\frac{N^2}{m^2}
\left(
\alpha \log N + \frac{m}{2^s}
\right)
$
productive-b steps are required to complete 
stage $s$ of Phase 2 correctly.
\item
At most
$
480 \cdot 
\frac{N^2}{m^2} 
\left(
\alpha \log N + \alpha \log m
\right)
$
productive-b steps are required to complete 
Phase 3 correctly.
\end{itemize}
\end{corollary}

\subsubsection{Total Productive Step Upper Bounds}

Lemma \ref{lemma:doubling-difference}
and Corollary \ref{corr:dbamc23-prodb-steps} give upper bounds on the 
number of \textit{blank-consuming} (productive-b) steps needed to 
complete the stages of Phase 1, the stages of Phase 2, and Phase 3 
of the protocol correctly. 
The following simple lemma then gives corresponding upper bounds on the
\textit{total} number of productive steps (including $X + Y$ interactions) 
that are sufficient to ensure each stage/Phase completes correctly 
with high probability.

\begin{lemma}
\label{lemma:dbamc-prod-steps}
During the $\dbamc$ protocol on a population with input margin
$\epsilon \ge \alpha \sqrt{N \log N}$ for some $\alpha \ge 1$,
assuming $m \ge n/10$ and $m \ge 8\alpha \cdot \log(11m)$, then 
\begin{itemize}
\item
At most $201N$ total productive steps are required to complete
stage $t$ of Phase 1 correctly with probability at least 
$1 - N^{-(\alpha/6)} - N^{-(\alpha/2)}$
\item
At most $961\frac{N^2}{m^2}\left(\alpha \log N + \frac{m}{2^s}\right)$
total productive steps are required to complete stage $s$ of 
Phase 2 correctly with probability at least $1 - 2N^{-\alpha}$.
\item
At most $961\frac{N^2}{m^2} (\alpha \log N + \alpha \log m)$ 
total productive steps are required to complete Phase 3 correctly
with probability at least $1 - 2N^{-\alpha}$. 
\end{itemize}
\end{lemma}

\begin{proof}
Consider any sequence of $\lambda$ total productive steps
consisting of $S_{xy}$ $X+Y$ steps, and $S_b$ 
productive-b steps. 
At the very least, we require $S_{xy} \le y_0 + S_b$, 
(where $y_0$ denotes the number of $Y$ agents at the start of 
the sequence). Otherwise, the total number of $X+Y$ steps
would exceed the maximum number of $Y$ agents able to ``fuel"
these interactions.  
Since $\lambda = S_{xy} + S_b$, we require that
$\lambda \le y_0 + 2 \cdot S_b$. 

In each stage of Phase 1, we know $y_0 \le \yhat_0 \le m$, 
and from Lemma~\ref{lemma:doubling-difference} we require
$S_b \le 100N$ to complete the stage correctly with 
probability at least $1 - N^{-(\alpha/6)} - N^{-(\alpha / 2)}$. 
Therefore each stage of Phase 1 requires at most 
$\lambda \le m + 2\cdot100N \le 201N$ total productive
steps to complete correctly with this same probability. 

In stage $s$ of Phase 2, we know that $y_0 \le \yhat_0 \le m/16$, 
and from Corollary~\ref{corr:dbamc23-prodb-steps} we require
$S_b \le 480 \cdot (N^2/m^2) (\alpha \log N + m/2^s)$
to complete the stage correctly with probability at least $1- 2N^{-\alpha}$.
So stage $s$ of Phase 2 requires at most
$
\lambda \le m/16 + 960 \cdot (N^2/m^2) (\alpha \log N + m/2^s)
\le 961 \cdot (N^2/m^2) (\alpha \log N + m/2^s)
$ 
total productive steps to complete correctly with this same probability. 

Similarly in Phase 3, we have $y_0 \le \yhat_0 \le \alpha \log m$,
and again from Corollary~\ref{corr:dbamc23-prodb-steps} we require
$S_b \le 480 \cdot (N^2/m^2) (\alpha \log N + \alpha \log m)$
to complete the Phase correctly again with probability at least $1-2N^{-\alpha}$.
Thus 
$
\lambda \le \alpha \log m + 
960 \cdot (N^2/m^2) (\alpha \log N + \alpha \log m)
\le
961 \cdot (N^2/m^2) (\alpha \log N + \alpha \log m)
$
total productive steps is sufficient to complete the Phase 
correctly this same probability. 
\end{proof}

\subsubsection{Efficiency of the Protocol}

Lemma~\ref{lemma:dbamc-prod-steps} gives bounds on the number of 
\textit{productive} steps needed for the protocol to complete
correctly with high probability. Here, we give upper bounds on 
the \textit{total} number of interactions needed to ensure that the 
requisite number of productive steps are met in each stage and 
Phase with high probability. 

To bound the total number of steps $\phi$ needed to obtain at least 
$\lambda$ productive steps with high probability, 
we provide a lower bound on the probability of a productive step
during a given stage or Phase and then apply Chernoff bounds. 
Letting $\phat$ denote the probability that the next interaction
is a productive step, we can observe that
\begin{align}
\phat 
=
\frac{b (i_x + i_y + x + y) + xy}{\binom{N}{2}} 
\ge
\frac{bx + xy}{\binom{N}{2}}
\ge
\frac{2x\cdot(y + b/2)}{N^2} = \frac{2x\cdot\yhat}{N^2}
\end{align}
The correct starting and ending conditions of each stage in Phase 1
and Phase 2 correspond to lower bounds on $\yhat$, and thus a 
sufficient lower bound on $\phat$ can be obtained by providing a 
lower bound on $x$ throughout the protocol. Specifically, 
we wish to show that the number of $X$ workers eventually remains
above some constant fraction of $m$ with high probability, which 
we prove in the following lemma.

\begin{lemma}
\label{lemma:dbamc-blank-ub}
During the $\dbamc$ protocol on a population with input margin
$\epsilon \ge \alpha \sqrt{N \log N}$ for some $\alpha \ge 1$,
and with $m \ge n/10$ for sufficiently large $m$,
by the end of Phase 1, the number of 
$X$ worker agents will exceed and never drop below $3m/8$ for each
remaining stage and Phase of the protocol with probability at 
least $1 - N^{-\alpha}$.
\end{lemma}

\begin{proof}
Since $m = x + y + b$ throughout the protocol, observe that
$b \le m/4$ implies that $x + y \ge 3m/4$.
Therefore, once the protocol reaches a point where the 
difference $x - y$ remains positive with high probability, 
it follows that $x \ge 3m/8$ will also continue to hold
with high probability. 
So to show that $x \ge 3m/8$ by the end of Phase 1, it is 
sufficient to show that $b$ will reach and remain below 
$m/4$ by the end of Phase 1 with high probability. 

First, we recall that Phase 1 of the protocol ends when 
$P \ge \hat x - \hat y + \epsilon = 7m/8 + \epsilon$, 
where $\epsilon$ is the input margin of the population. 
Since Lemma~\ref{lemma:dbamc-prod-steps} tells us that
each stage of Phase 1 completes correctly with high probability, 
it follows that $x - y \ge 0$ must hold if 
the protocol has reached a configuration where Phase 1 is complete.


Moreover, because the correct ending conditions of Phase 1
imply that $\hat y = y + b/2 \le m/16$, it follows that if Phase 1 
completes correctly, we must have $b \le m/8$ --- any
larger value of $b$ would imply a value of $\yhat$ that exceeds $m/16$.
By Lemma~\ref{lemma:dbamc23-corr-mk-nodouble}, in  
Phase 2 of the protocol, which begins with $\yhat = m/16$,
the value of $\yhat$ will never exceed $m/8$ with probability 
at least $1-N^{-\alpha}$ assuming that $m \ge n/10$ and 
$m$ is sufficiently large. 
Furthermore, this means that each subsequent stage of Phase 2, 
and Phase 3 of the protocol will also ensure that $\hat y \le m/8$
with probability at least $1-N^{-\alpha}$ under the same assumptions.
In turn, this means that from the end of Phase 1 and onward, 
it will hold that $b \le m/4$ during each subsequent stages of Phase 2 
and Phase 3 of protocol with this same probability. 
\end{proof}

The fact that $x \ge 3m/8$ will hold with high probability
for the remainder of the protocol by the end of Phase 1 can now
be used to provide lower bounds on $\phat$ and subsequent 
high probability upper bounds on the total number of steps 
needed to complete each Phase and stage correctly. The following lemma gives 
these total step upper bounds for each part of the protocol:

\begin{lemma}
\label{lemma:dbamc-eff}
During the $\dbamc$ protocol on a population with 
input margin $\epsilon \ge \alpha \sqrt{N \log N}$ for $\alpha \ge 1$, 
assuming that $m \ge n/10$ for large enough $m$:
\begin{itemize}
\item
at most 
$
\max
\left\{
\tfrac{1608N^3}{mn}, \tfrac{8576N^3}{m^2}
\right\}
$
total steps are needed to complete each stage of Phase 1 correctly.
\item
at most
$
82,000 \cdot \frac{N^4}{m^3} + 
41,000 \cdot 2^{s+1} \frac{N^2}{m^2} \cdot \alpha \log N
$
total steps are needed to complete stage $s$ of Phase 2 correctly.
\item
at most 
$2563 \cdot \frac{N^4}{m^3}(\alpha \log m + \alpha \log N)$
total steps are needed to complete Phase 3 correctly,
\end{itemize}
each with probability at least
$1 - 2N^{-\alpha} - N^{-(\alpha/6)} - N^{-(\alpha/2)}$ when $N$ is 
sufficiently large.
\end{lemma}
   
\begin{proof}
Let $\phi_p$ denote the number of productive steps obtained 
in a sequence of $\phi$ total steps, and let $\lambda$ denote
the required number of productive steps for a Phase/stage to complete
correctly whp. 
If the probability of a productive step is $\phat \ge p'$, 
then the lower Chernoff bound from Lemma~\ref{lemma:chernoff-tail-bounds}
gives 
$
\Pr(\phi_p \le \tfrac{\phi\cdot p'}{2}) 
\le \exp(-\tfrac{\phi\cdot p'}{8})
$. 
So to provide upper bounds on $\phi$ for each Phase/stage,
we first give lower bounds on $\phat$ and then choose
$\phi$ such that $\phi_p \ge \lambda$ holds whp using the 
Chernoff bound. We do this for each Phase as follows:

\begin{enumerate}[i.]
\item
\textit{Stages of Phase 1:}
By Lemma~\ref{lemma:dbamc-blank-ub}, $b \le \frac{m}{4}$ will
eventually hold by the end of Phase 1 with probability at least 
$1 - N^{-\alpha}$. 
Thus, for the stages of Phase 1, we provide lower bounds on $\phat$ in 
two cases. 
For stages where $b \ge \frac{m}{4}$, we have
$$\phat \ge (b\cdot i_x)/N^2 \ge mn/4N^2,$$ since $i_x \ge n/2$ 
given an $I_X$ input majority. 
For stages where $b \le \frac{m}{4}$ holds, it follows that
$$\phat \ge (2x\cdot \yhat)/N^2 \ge 3m^2/64N^2$$ since
$\yhat \ge m/16$ for all of Phase 1. 

Now, by Lemma~\ref{lemma:dbamc-prod-steps}, each stage of Phase 1
completes correctly within $\lambda = 201N$ productive steps with 
probability at least $1- N^{-(\alpha/6)} - N^{-(\alpha/2)}$.
So taking $\phi$ total steps such that $\phi\cdot\phat/2 = \lambda = 201N$ 
means that $\phi_p \le \lambda$ with probability at most 
$\exp(-\lambda/4) \le \exp(-50N) \le N^{-\alpha}$.
In the case that $b \ge \frac{m}{4}$ and $\phat \ge \frac{mn}{4N^2}$, 
it is sufficient to set
$$
\phi = 2 \cdot \frac{4N^2}{mn} \cdot 201N = 1608 \cdot \frac{N^3}{mn}
$$
in order to ensure $\phi_p \ge 201N$ with probability at 
least $1 - N^{-\alpha}$.
In the case that $b \le \frac{m}{4}$ and $\phat \ge \frac{3m^2}{64N^2}$, 
it is sufficient to set
$$
\phi = 2 \cdot \frac{64N^2}{3m^2} \cdot 201 N = 8576 \cdot \frac{N^3}{m^2}
$$
in order to ensure the same result. 
Summing over all error probabilities and taking a union bound, 
it follows that each stage of Phase 1 requires at most
$$\phi \le \max\left\{1608\cdot\frac{N^3}{mn} \;,\; 8576\cdot\frac{N^3}{m^2}\right\}$$
total steps to complete correctly with probability at
least $1 - 2N^{-\alpha} - N^{-(\alpha/6)} - N^{-(\alpha/2)}$. 

\item
\textit{Stages of Phase 2}:
For each stage $s$ of Phase 2, we have $\yhat \ge \frac{m}{16\cdot 2^{s+1}}$, 
which means that
$$
\phat
\ge \frac{2x \cdot \yhat}{N^2} 
\ge \frac{3m}{4}\cdot\frac{m}{16\cdot2^{s+1}} \cdot \frac{1}{N^2}
= \frac{3}{64 \cdot 2^{s+1}}\cdot \frac{m^2}{N^2},
$$
where we use the fact that $x \ge \frac{3m}{8}$ holds with probability
at least $1 - N^{-\alpha}$ from Lemma~\ref{lemma:dbamc-blank-ub}.
By Lemma~\ref{lemma:dbamc-prod-steps}, we have that stage $s$ of
the Phase completes correctly within 
$
\lambda = 
961 (N^2/m^2)(\lambda \log N + m/2^2)
$
productive steps with probability at least $1 - 2N^{-\alpha}$,
and so taking $\phi$ total steps such that
$
(\phi\phat)/2 = \lambda 
$ 
ensures that $\phi_p \le \lambda$  with probability at most 
$\exp(-\lambda/4) \le \exp(-\alpha \log N)$ by the
Chernoff bound. 

Given the lower bound on $\phat$, it is thus sufficient to set
$$
\phi = 
2 \cdot \frac{64 \cdot 2^{s+1} N^2}{3m^2}
\cdot 961 \frac{N^2}{m^2}\left(\lambda \log N + \frac{m}{2^s}\right)
= 
82,000 \cdot \frac{N^4}{m^3} + 
41,000 \cdot 2^{s+1} \cdot \frac{N^2}{m^2}\cdot \alpha \log N
$$
to ensure $\phi_p \ge \lambda$ with probability 
at least $1 - N^{-\alpha}$ for each stage $s$ of the Phase.
Summing over all error probabilities and taking a union bound,
it follows that $\phi$ total steps are sufficient to complete
each stage of Phase 2 correctly with total probability 
at least $1 - 4N^{-\alpha}$

\item
\textit{Phase 3}:
Again because $\yhat \ge 1/2$ and $x \ge 3m/8$ with probability at 
least $1-N^{-\alpha}$ throughout Phase 3 by Lemma~\ref{lemma:dbamc-blank-ub},
we have
$$
\phat \ge \frac{2x\cdot \yhat}{N^2} 
\ge 
\frac{3}{4}\cdot\frac{m}{N^2}
$$
throughout the Phase with this same probability. 
By Lemma~\ref{lemma:dbamc-prod-steps},
$\lambda = 961 (N^2/m^2)(\alpha \log m + \alpha \log N)$
productive steps are needed to complete the Phase correctly
with probability at least $1 - 2N^{-\alpha}$. 
By a similar argument as in Phases 1 and 2, this means taking
$$
\phi 
= 
2 \cdot \frac{4N^2}{3m^2} \cdot 961 
\cdot\frac{N^2}{m^2}\left(\alpha \log m + \log N \right)
= 
2563 \cdot \frac{N^4}{m^3}
\left(\alpha \log m + \log N \right)
$$
total steps is sufficient to ensure that
$\phi_p \ge \lambda$ with probability at least
$1 - N^{-\alpha}$.
Now again by the union bound, it follows that these 
$\phi$ total steps are sufficient to ensure that
Phase 3 of the protocol completes correctly with 
probability at least $1 - 4N^{-\alpha}$. 
\end{enumerate}
We conclude the proof by observing that the success probability 
of stages of Phase 2 and Phase 3 are at least that of the stages of 
Phase 1, thus giving the stated result. 
\end{proof}

\subsubsection{Concluding the Proof}

Using the Lemmas and Corollaries established in the preceding
subsections, we now conclude to prove Theorem~\ref{thm:dbamc-overall}. 
For convenience, we restate the Theorem below:

\thmdbamcoverall*

\begin{proof}
By Lemma~\ref{lemma:dbamc-eff}, provided that
$m \ge n/10$ and $N$ is sufficiently large, each of the 
$O(\log N)$ stages of Phase 1 of the protocol 
require $\max\{O(N^3/(mn))\;,\; O(N^3/m^2)\} \le O(N^4/m^3)$
total steps to complete correctly with probability
at least $1 - 2N^{-\alpha} - N^{-(\alpha/6)} - N^{-(\alpha/2)}$;
each of the $O(\log N)$ stages of Phase 2 of the 
protocol require $O(N^4/m^3)$ total steps to complete
correctly with probability
at least $1 - 2N^{-\alpha} - N^{-(\alpha/6)} - N^{-(\alpha/2)}$;
and Phase 3 of the protocol requires $O((N^4/m^3)\log N)$ total
steps to complete correctly with probability at least 
$1 - 2N^{-\alpha} - N^{-(\alpha/6)} - N^{-(\alpha/2)}$.
Now for any $c \ge 1$, it follows by the union bound that 
for an appropriate choice of constant $\alpha \ge 1$, the entire
$\dbamc$ protocol will complete correctly within 
$O((N^4/m^3) \log N)$ total interactions with probability
at least $1 - N^{-c}$ given that $m \ge n/10$ and for
sufficiently large $N$, which gives the stated result.
\end{proof}

As a corollary of Theorem~\ref{thm:dbamc-overall},
we can simplify the
convergence guarantees of the $\dbamc$ protocol in the 
case when $m = \Theta(n)$. 

\begin{restatable}{corollary}{corrdbamcrelated}
\label{corr:dbamc-related-pop}
There exists some constant $\alpha \ge 1$ such that,
for a population of $n$ inputs, 
$m = cn$ workers where $c \ge 1$, and an initial input 
margin $\epsilon \ge \alpha\sqrt{N \log N}$, the
$\dbamc$ protocol correctly computes the majority
of the inputs within $O(N \log N)$ total 
interactions with probability at least $1 - N^{-a}$
for any $a \ge 1$ when $N$ is sufficiently large
\end{restatable}

We note that the result of Corollary~\ref{corr:dbamc-related-pop}
implies that the CI model input margin lower bound from 
Theorem~\ref{thm:input-margin-lb} is tight up to 
a multiplicative $O(\sqrt{\log N})$ factor. 


%
%

\section{Approximate Majority with Transient Leaks}
\label{sec:leaks}

We now consider the behavior of 
the $\dbam$ and $\dbamc$ protocols in the presence of
transient leak faults. 
Even in the presence of these adversarial events
(which occur up to some bounded rate $\beta$), 
both the $\dbam$ and $\dbamc$ 
protocols will, with high probability, reach configurations
where nearly all agents share the input majority opinion.
In the presence of leaks, we consider the approximate 
majority predicate to be computed correctly upon reaching
these low sample-error configurations. 

Recall that a transient leak is an event where
an agent spuriously changes its state according to some leak function $\ell$.
For example, we denote by $U \rightarrow V$ the event that an agent in 
state $U$ transitions to state $V$ due to a leak event, 
where the timing of such events are dictated by the random scheduler and
occur with probability $\beta$ at each subsequent interaction step. 
In both the \dbam\ and \dbamc\ protocols, the only state changes that could 
possibly take place due to leaks are $X \rightarrow B$, $Y \rightarrow B$,
$B \rightarrow X$, and $B \rightarrow Y$ because these describe 
all possible state changes that could take place in the presence of 
an interacting partner. However, our analysis considers an
\textit{adversarial} leak event $X \rightarrow Y$, which
maximally decreases our progress measures and can be considered 
the ``worst'' possible leak. Though this leak event is not chemically sound 
(because no normal interaction can cause an $X$ agent to transition 
to the $Y$ state), our results demonstrate that both the
$\dbam$ and $\dbamc$ protocols are robust to this strong
adversarial leak event. 
Thus in a more realistic chemically sound setting, 
our results will also hold, as the set of transitions working against our
progress measures are weaker.

\subsection{Leak Robustness of the $\dbam$ Protocol}
\label{sec:dbl}

We start by showing the leak-robustness of the $\dbam$ protocol 
for approximate majority in the original population protocol model. 
Recall that in the standard model, \textit{all} agents are susceptible to leaks. 
Our main result shows that when the leak rate $\beta$ 
is sufficiently small, the protocol still reaches a
configuration with bounded \textit{sample error} 
(the proportion of agents in the non-initial-majority state) 
within $O(n \log n)$ total interactions with high probability. 
Unlike the scenario without leak events, note that the 
protocol will never be 
able to fully converge to a configuration where all agents 
remain in the majority opinion. However, reaching a configuration
where despite leaks, \textit{nearly} all agents hold the 
input majority value state matches similar results of 
\cite{angluin2008simple, condon2019, alistarh2017robust, alistarh2020robust}. 
Formally, we have the following theorem, which 
characterizes the eventual sample error of the protocol with respect
to the magnitude of the leak rate $\beta$. 
\begin{restatable}{theorem}{thmdblmain}
  \label{thm:dbl-main}
  There exists some constant $\alpha \ge 1$ such that, for a 
  population with initial input margin
  $\epsilon \ge \alpha \sqrt{n \log n}$
  and adversarial leak rate $\beta \le (\alpha \sqrt{n \log n})/12672n$,
  an execution of the $\dbam$ protocol will reach a configuration with
  \begin{enumerate}
    \item 
    sample error $O(\log n/n)$ when
    $\beta \le O(\log n /n)$
    \item
    sample error $O(\beta)$ when 
    $\omega(\log n /n ) \le \beta \le (\alpha \sqrt{n \log n})/12672n$
  \end{enumerate}
  within $O(n \log n)$ total interactions with probability 
  at least $1 - n^{-c}$ for any $c \ge 1$ when $n$ is
  sufficiently large. 
\end{restatable}

\paragraph*{Analysis Overview}

To prove Theorem~\ref{thm:dbl-main}, we again make modified 
use of the random walk tools from \cite{condon2019}.
Using the progress measures $\hat y = y + b/2$ and 
$P = \hat x - \hat y$,
observe that an $X \to Y$ leak event incurs twice as much 
negative progress to both measures as opposed to 
$X+B$ events. 
Compared to the analysis from the non-leak setting,
the analysis \textit{with} adversarial leaks must account for the 
stagnation (or potentially the reversal) of the protocol's 
progress toward reaching a low-sample-error configuration.
Note that since the sample error of a configuration is
defined to be $(y+b)/n$ (since we assume an initial $x$ majority wlog), 
we will use the value $\hat y/n$ to approximate the sample error of 
a configuration.

%

%

Moreover, throughout the analysis, 
we make use of the following structure, which is partially
adapted from the non-leak $\dbam$ analysis in \cite{condon2019}:
\begin{enumerate}
\item 
First, we classify the leak rate $\beta$ into two categories:
when $\omega(\log n /n) \le \beta \le O(\sqrt{n\log n}/n)$
we say that $\beta$ is \textit{large}. 
When $\beta \le O(\log n/n)$, we say that $\beta$ is \textit{small}.
\item
\textbf{Phase 1} begins with the start of the protocol
and ends correctly when $P = \hat x - \hat y \ge 2n/3$.
Each \textbf{stage of phase 1} begins with $P = \Delta_0 \cdot 2^t$, 
and ends correctly once $P = \Delta_0 \cdot 2^{t+1}$ for 
$t \in \{0,1, \dots, O(\log n)\}$. 
\item
\textbf{Phase 2} of the protocol begins once $\hat y \le n/6$
(which is equivalent to $\hat x - \hat y \ge 2n/3$) and
ends correctly when $\hat y \le 25 \beta n$ when $\beta$ is large,
and when $\hat y \le 25 a \log n$ (where $a \ge 1$) when 
$\beta$ is small. 

For large $\beta$, each normal \textbf{stage of Phase 2} 
begins with $\hat y \le n/k$ where $6 \le k \le 1/(50\beta)$
and completes correctly when reaching 
$\hat y \le \max\{n/2k, 50\beta n\}$. 
Here, the final stage of Phase 2 begins with 
$\hat y = n/ k = 50\beta n$
and ends once $\hat  y \le 25 \beta n$. 

For small $\beta$, each normal stage of Phase 2 begins
with $\hat y= n/k$ where $6 \le k < n/(50 a \log n)$ and
completes correctly when reaching 
$\hat y \le \max\{n/2k, 50 a \log n\}$.
The final stage of Phase 2 in this case begins with $\hat y = n/k = 50 a \log n$ and ends once $\hat y \le 25 a \log n$.
\end{enumerate}

In this leak-prone setting, we refer to \textit{productive 
interactions} as any of the non-null transitions found
in Figure~\ref{fig:dbam} in addition to a leak event. 
The three non-null and non-leak transitions are referred
to as \textit{non-leak productive steps}.
For each phase and stage, we obtain high-probability 
estimates of the number of productive and total steps needed
to complete the phase/stage correctly in two steps: 
first, we bound
the number of leak events that can occur
during a fixed interval of productive events. Then, we 
show that a smaller sub-sequence of non-leak productive
steps is sufficient to ensure that enough positive progress is made
to offset the negative progress of the leak events.
We again rely on a combination of 
Chernoff concentration bounds and martingale inequalities in
order to show this progress at every phase and stage. 

Theorem~\ref{thm:dbl-main} also separates 
the behavior of the protocol into two classes: 
when $\beta \le O(\log n/n)$ (\textit{small} leak rate),
and when $\omega(\log n/n) \le \beta \le O(\sqrt{n \log n}/n)$ 
(\textit{large} leak rate). 
When the leak rate is large, 
the probability of a leak event conditioned on a productive step
becomes roughly equal to the conditional probability of a 
non-leak productive step when $\hat y = O(\beta n)$. 
Thus, we cannot expect the protocol to make further ``progress''
toward a lower sample-error configuration with high probability 
beyond $\hat y = O(\beta n)$. The same holds for small leak rate
when $\beta = O(\log n/n)$, and for even smaller values of $\beta$, 
our analysis tools only allow for the high-probability guarantee that
$\hat y$ eventually drops to $O(\log n)$. 

We also give additional arguments
showing that the protocol remains in a configuration 
with sample error $O(\log n/n)$ for small leak rate, and with
sample error $O(\sqrt{n \log n}/n)$ for large leak rate, for
at least a polynomial number of interactions with high probability
following the completion of Phase 2 of the protocol.


%
%

\subsubsection{Proof of Theorem~\ref{thm:dbl-main}}
We now develop and prove the lemmas which lead to 
a proof of our main result considering the $\dbam$ protocol 
with transient leaks. 


\paragraph*{Utility Lemmas}
To begin, we state and prove several lemmas that will be
used repeatedly throughout the analysis. The first lemma proves that,
with bounded leak rate $\beta$ and a bounded number of leak events during
a sequence of productive steps, the progress measure $\hat x - \hat y$ 
will never decrease by more than a factor of two during a stage of 
Phase 1. 

\begin{lemma}
  \label{lemma:dbl1-no-double}
  Throughout the $\dbam$ protocol on a population with
  initial input margin \ $\alpha \sqrt{n \log n}$ and adversarial
  leak rate $\beta \le \frac{\alpha\sqrt{n \log n}}{12672n}$,
  the probability that, during the stage of Phase 1 starting at
  $\hat x - \hat y = \Delta$, the protocol 
  reaches a configuration with $\hat x - \hat y = \Delta/2$ within
  a sequence of $66n$ productive steps before finishing correctly is at most
  $n^{-(\alpha/6)}$ when the number of leak events
  in the sequence is at most $(\alpha \sqrt{n \log n})/8$.
\end{lemma}

\begin{proof}
  For a population with initial input margin $\Delta_0 = \alpha \sqrt{n \log n}$,
  the number of leak events in the sequence of $66n$ productive steps
  is at most $(\alpha \sqrt{n \log n})/4 \le \Delta_0/8$. Because
  each leak event decreases the value $\hat x - \hat y$ by two, in total
  these leak events can decrease the $\hat x - \hat y$ by at most $\Delta_0/4$
  in magnitude. 

  For the stage of Phase 1 that starts with $\hat x - \hat y = \Delta$,
  up until the point that $\hat x - \hat y < \Delta/2$, the probability
  of a $B+X$ interaction conditioned on a blank-consuming productive step,
  denoted by $p(bx)$, is at least
  \begin{align}
    p(bx) \ge \frac{x}{n} \ge \frac{1}{2} + \frac{\hat x - \hat y}{4n}
    = \frac{2n + \Delta}{4n}.
  \end{align}

  Given the at most $\Delta/4$ decrease to $\hat x - \hat y$ contributed
  by the leak events, $\hat x  - \hat y$ can only decrease to $\Delta/2$
  if the number of $Y+B$ interactions exceeds the number of $X+B$ interactions
  by at least $\Delta/4$ before the stage completes correctly.

  Using Lemma~\ref{lemma:one-dim-rw}, observe that this bad event occurs
  with probability at most
  \begin{align}
    \left(
    \frac{2n - \Delta}{2n + \Delta}
    \right)^{\Delta/4}
    =
    \left(
    1- \frac{2\Delta}{2n+\Delta}
    \right)^{\Delta/4}
    &\le
    \exp\left(
    - \frac{\Delta^2}{4n + 2\Delta}
    \right)\\
    &\le
      \exp\left(
      - \frac{\alpha \log n}{6}
      \right)
      \le n^{-(\alpha/6)}.
  \end{align}
\end{proof}

The next lemma shows that during a fixed sequence of productive steps
during a stage of Phase 2 of the protocol with a bounded number of leaks,
the progress measure $\hat x - \hat y$ will never drop below $n/2$ with 
high probability. 
In turn, this implies that the probability of a $B+X$ interaction
(conditioned on a blank-consuming step) is bounded from below by
a constant fraction greater than $1/2$. 

\begin{lemma}
  \label{lemma:dbl2-no-double}
  During Phase 2 of the $\dbam$ protocol with adversarial leak
  rate $\beta \le O(\sqrt{n \log n}/n)$, the probability
  of a $X+B$ interaction conditioned on a blank-consuming productive
  step is at least $5/8$ throughout any sequence of \ $16n$ productive steps
  with probability at least $1 - n^{-a}$ for $a \ge 1$ and sufficiently large $n$
  when the number of leak events within this sequence is at most $n/24$.
\end{lemma}

\begin{proof}
  Recall that, conditioned on a blank-consuming productive step, the
  probability of a $X+B$ interaction, denoted by $p(bx)$, can be
  written as
  \begin{align}
    p(bx) \ge \frac{1}{2} + \frac{\hat x - \hat y}{4n}.
  \end{align}
  To show that $p(bx) \ge 5/8$ holds throughout a sequence of interactions
  during Phase 2, it is sufficient to show that $\hat x - \hat y \ge n/2$
  holds during the sequence.

  For this, recall that Phase 2 of the protocol begins when $\hat x - \hat y$
  reaches $2n/3$. Assuming that in a sequence of $16n$ productive steps we have
  at most $n/24$ leaks, it follows that in order for $\hat x - \hat y$ to drop
  further to $n/2$ requires the number of $Y+B$ steps to exceed the number
  of $X+B$ steps by at least $n/12$ throughout the sequence in the worst case. 

  Thus again using Lemma~\ref{lemma:one-dim-rw}, we have that this
  event occurs with probability at most
  \begin{align}
    \left(
    \frac{3/8}{5/8}
    \right)^{n/12}
    =
    \left(
    \frac{3}{5}
    \right)^{n/12}
    \le
    n^{-a}
  \end{align}
  for $a \ge 1$ when $n$ is sufficiently large. Thus with high probability,
  we have that $\hat x - \hat y$ will remain above $n/2$ throughout
  a sequence of at most $16n$ productive steps with high probability
  during Phase 2 of the protocol.
  It follows then that $p(bx) \ge 5/8$ will hold with probability
  at least $1 - n^{-a}$ throughout each such sequence in Phase 2. 
\end{proof}

\begin{lemma}
  \label{lemma:dbl-x-lb}
  Throughout Phase 1 of the $\dbam$ protocol on a population with
  initial input margin $\alpha \sqrt{n \log n}$ and adversarial
  leak rate $\beta \le O(\sqrt{n \log n}/n)$, 
  the number of $x$ agents in the population remains above
  $n/4$ with probability at least $1 - n^{-\alpha} - n^{-(\alpha/6)}$ when 
  $n$ is sufficiently large.
\end{lemma}

\begin{proof}
  When $b \le n/4$, we have $x + y \ge 3n/4$, which by
  Lemma~\ref{lemma:dbl1-no-double} implies that $x \ge 3n/8 > n/4$ with
  probability at least $1- n^{-(\alpha/6)}$.

  On the other hand, say $n/4 \le b \le n/2$ at some point during Phase 1
  of the protocol. As long as $b \ge n/2$,
  we have $x + y \ge n/2$, which by similar reasoning as the previous case
  implies $x \ge n/4$ with probability at least $1 - n^{-(\alpha/6)}$.
  Thus, to prove the claim it is sufficient to  show that $b \le n/2$ will
  hold with high probability throughout protocol.

  To do this, let $b_0$ denote the number of blank agents at a given point,
  and let $b_1$ denote the number of blank agents after the next
  productive step. Since we are considering adversarial leak behavior,
  note that the count of blank agents is only affected by $X+Y$ interactions
  (which increase $b$ by 2) and $X+B$ or $Y+B$ interactions (which decrease
  $b$ by 1). Thus given the history of the protocol up encoded by the
  value of $b_0$ we have
  \begin{align}
    \E[b_1 | b_0] = b_0 + 2 \cdot p(xy) - p(bx/by)
  \end{align}
  where $p(xy)$ denotes the conditional probability of an $X+Y$ interaction,
  and $p(bx/by)$ denotes the conditional probability of either a $X+B$ or $Y+B$
  interaction.
  Since we have $n/4 \le b \le n/2$, it follows up until the point
  $b > n/2$ that $x + y \ge n/2$, and thus $2\cdotp(xy) \le 2(n/4)(n/4) = n^2/8$.
  Additionally, we have $p(bx/by) = b(x+y) \ge (n/4)(n/2) = n^2/8$. This implies
  \begin{align}
    \E[b_1 | b_0] = b_0 + 2p(xy) - p(bx/by) \le n^2/8 - n^2/8 = b_0,
  \end{align}
  which means $b$ is a supermartingale with respect to the sequence of
  productive events.

  Letting $b_0 = n/4$ and letting $b_t$ denote the difference in the count of
  blank agents from $b_0$ after the next $t$ productive steps,
  Azuma's inequality \cite{gs-book} gives that
  \begin{align}
    \Pr[b_t > n/4] \le \exp\left(- \frac{(n/4)^2}{2t}\right)
    \le \exp\left(- \frac{n^2}{16t} \right).
  \end{align}
  Thus for any $t = O(n)$, it follows that
  $\Pr[b_t > n/4] \le \exp(- a \log n) \le n^{-a}$  for any $a \ge 1$
  when $n$ is sufficiently large. Then summing over all error probabilities
  and taking a union bound, we find that the result will hold
  with probability at least $1 - n^{-\alpha} - n^{-(\alpha/6)}$. 
\end{proof}

The next lemma shows a lower bound on the number of agents in 
state $X$ throughout a finite sequence of productive steps during
Phase 2 of the protocol when the number of leak events is bounded.

\begin{lemma}
  \label{lemma:dbl-x-lb2}
  Throughout Phase 2 of the $\dbam$ protocol on a population with
  adversarial leak rate $\beta \le O(\sqrt{n \log n}/n)$,
  during any sequence of at most $16n$ productive steps, 
  the number of $x$ agents in the population remains above
  $n/3$ with probability at least $1 - n^{-a}$  when the number of
  leak events in the sequence is at most $n/24$ for $a \ge 1$
  when $n$ is sufficiently large.
\end{lemma}

\begin{proof}
  The proof of Lemma~\ref{lemma:dbl2-no-double} shows that during
  Phase 2 of the protocol, $\hat x - \hat y$ will remain above
  $n/2$ throughout a sequence of at most $16n$ productive steps
  with probability at least $1 - n^{-a}$ when the number of
  leak events in the sequence is at most $n/24$. 
  Because $\hat x + \hat y = n$, it follows that
  $\hat y = (n - (\hat x - \hat y))/2 \le n/4 < n/3$, which will hold with
  probability at least $1 - n^{-a}$ throughout each stage of Phase 2. 
\end{proof}

Finally, the following lemma shows a relationship between the
number of blank-consuming productive steps guaranteed to occur
among a greater sequence of non-leak productive events.
\begin{lemma}
  \label{lemma:dbl-prob-prod}
  Throughout the $\dbam$ protocol, in a sequence of $\lambda$
  total non-leak productive steps, at least $(\lambda - y_{0})/2$
  will be blank-consuming steps, where $y_{0}$ is
  the initial number of $Y$ agents present in the population
  throughout the sequence. 
\end{lemma}

\begin{proof}
  Consider a sequence of $\lambda_b$ blank-consuming productive steps.
  This means at most $\lambda_b$ of the interactions are $Y+B$ steps,
  and so along with the $y_{0}$ initial count of $Y$ agents, there
  can be at most $\lambda_b + y_{0}$ interactions of type $X+Y$.
  Otherwise, the number of $Y$ agents needed to produce a
  $X+Y$ step would be insufficient. Thus a total of
  $\lambda_b$ blank-consuming steps will be  obtained in at most
  $\lambda \le 2\cdot\lambda_b + y_0$ total (non-leak) productive events.
  Solving for $\lambda_b$ completes the proof.
\end{proof}


\paragraph*{Phase 1 Behavior}
We now show the correctness of each stage of Phase 1 of the protocol. 
To begin, we bound the number of leak events with high probability
within a fixed sequence of productive steps. 

\begin{lemma}
  \label{lemma:dbl-a1.1}
  For a population with initial input margin
  $\epsilon \ge \alpha \sqrt{n \log n}$ for some $\alpha \ge 1$
  and adversarial leak rate $\beta \le \frac{\alpha\sqrt{n \log n}}{12672n}$,
  at any point in Phase 1 of the $\dbam$ protocol,
  the number of leak events among a sequence of
  $66n$ productive steps is no more than
  $(\alpha\sqrt{n \log n})/8$ with probability
  at least $1 - 2n^{-\alpha} - n^{-(\alpha/6)}$ for sufficiently large $n$. 
\end{lemma}

\begin{proof}
  Recall that $p(l)$, the probability of a leak event
  conditioned on a productive event can be written as
  \begin{align}
    p(l) =
    \frac{\beta \binom{n}{2}}{\beta\binom{n}{2} + (1 - \beta)\phi}
    \le
    \frac{\beta (n^2/2)}{\beta (n^2/3) + (1-\beta) \phi}
    \label{eq:pl-1}
  \end{align}
  where $\phi = b(x+y) + xy$ and the last inequality holds for
  sufficiently large $n$. To further bound $p(l)$ from above, we
  can observe that 
  \begin{align}
    \phi = b(x + y) + xy
    &= x(y + b) + by \\
    &\ge x\hat y 
  \end{align}
  Now, since $\hat y \ge n/6$ throughout Phase 1, and since by
  Lemma~\ref{lemma:dbl-x-lb}, $x \ge n/4$ holds throughout the entire
  protocol with high probability at least $1 - n^{-\alpha} - n^{-(\alpha/6)}$, 
  it follows that $\phi \ge x \hat y \ge n^2/24$ will
  hold with this same probability throughout all stages of the phase.
  Substituting into~\eqref{eq:pl-1} gives 
  \begin{align}
    p(l)
    \le \frac{\beta (n^2/2)}{\beta (n^2/3) + (1-\beta) \phi}
    &\le \frac{\beta (n^2/2)}{\beta (n^2/3) + (1-\beta) (n^2/24)}\\
    &\le \frac{\beta (n^2/2)}{\beta (7n^2/24) + n^2/24}\\
    &= \frac{1/2}{7/24 + 1/(24\beta)}
    \;\;\le 12\beta
      \label{eq:pl-2}
  \end{align}
  By the assumption that $\beta \le \frac{\alpha\sqrt{n \log n}}{12672n}$,
  this means
  \begin{align}
    p(l)
    \le \frac{12}{12672} \cdot \frac{\alpha\sqrt{n \log n}}{n}
    \le \frac{\alpha\sqrt{n \log n}}{1056n}
  \end{align}
  with high probability throughout all stages of Phase 1.

  In $66n$ productive steps then, the expected number of leaks
  $\E[l]$ is at most $\frac{\alpha\sqrt{n \log n}}{16}$. Using an upper
  Chernoff bound, and letting $l$ denote the number of leak events,
  it follows that
  \begin{align}
    \Pr\left[l > \frac{\alpha\sqrt{n \log n}}{8}\right]
    &=
    \Pr\left[l > 2 \cdot \E[l]\right]\\
    &\le
      \exp\left(- \frac{\alpha\sqrt{n \log n}}{54}\right)
    \le
    \exp\left( - \alpha\log n\right),
  \end{align}
  where the last inequality holds for sufficiently large $n$.
  So by the union bound, with probability at least $1 - 2n^{-\alpha} - n^{-(\alpha/6)}$, 
  the number of leak events within
  $66n$ productive steps is at most $(\alpha\sqrt{n \log n})/8$. \qedhere
\end{proof}

\begin{lemma}
  \label{lemma:dbl-a1.2}
  For a population with initial input margin
  $\epsilon \ge \alpha \sqrt{n \log n}$ for some $\alpha \ge 1$
  and adversarial leak rate $\beta \le \frac{\alpha\sqrt{n \log n}}{12672n}$,
  each stage of Phase 1 in the $\dbam$ protocol completes correctly
  within $2376n$ total interaction events with probability
  at least $1 - 6n^{-\alpha} - 3n^{-(\alpha/6)}$. 
\end{lemma}

\begin{proof}

  By Lemma~\ref{lemma:dbl-a1.1}, the number of leak events within
  $66n$ productive steps is at most $(\alpha\sqrt{n \log n})/8$ with 
  probability at least  $1 - 2n^{-\alpha} - n^{-(\alpha/6)}$.
  Since every leak event decrements the progress measure $\hat x - \hat y$
  by 2, this means leak events decrease $\hat x - \hat y$  
  in magnitude by at most $(\alpha\sqrt{n \log n})/4$ with this same 
  probability throughout the sequence.

  So within the sequence of $66n$ productive steps, at least $65n$ 
  are non-leak productive steps. Thus in order to ensure $\hat x - \hat y$ 
  doubles from its initial
  value $\Delta_0$, it is sufficient to show that among a smaller sequence
  of $\lambda \le 66n$ blank-consuming productive steps, the number of $X + B$
  interactions will exceed the number of $Y + B$ interactions by at least
  $\Delta_0 + (\alpha\sqrt{n \log n})/4$ with high probability.
  Note by Lemma~\ref{lemma:dbl-prob-prod} that in a sequence of
  at least $65n$ non-leak productive steps during Phase 1 of 
  the protocol (where $y \le n$),
  at least $32n$ will be blank-consuming productive steps. 
  
  Now, recall that when $\hat x - \hat y \ge \Delta_0/2$, the probability of
  a $X+B$ interaction conditioned on a blank-consuming interaction ---
  which we will denote by $p(bx)$ --- can be bounded below by
  \begin{align}
    p(bx) = \frac{bx}{bx + xy}
    \ge \frac{x}{x+y}
    \ge \frac{x + b/2}{x+y+b}
    \ge \frac{\hat x}{n}
    \ge \frac{1}{2} + \frac{\Delta_0}{4n}.
    \label{eq:pbx-1}
  \end{align}
  This also implies that the corresponding conditional probability 
  of a $Y+B$ interaction (denoted by $p(by)$) is at most
  \begin{align}
    p(by) = 1 - p(bx) \ge \frac{1}{2} - \frac{\Delta_0}{4n}.
    \label{eq:pby-1}
  \end{align}
  By Lemma~\ref{lemma:dbl1-no-double}, $\hat x - \hat y \ge \Delta_0/2$
  will hold with probability at least $1 - n^{-(\alpha/6)}$ throughout
  every stage of Phase 1, and so the bounds in 
  \eqref{eq:pbx-1} and \eqref{eq:pby-1} hold with the same probability. 
  
  Now we will show that over the sequence of $\lambda \le 66n$ 
  blank-consuming steps, the number of $X+B$ interactions 
  (denoted by $S_{bx}$) will exceed the number of 
  $Y+B$ interactions (denoted by $S_{by}$) by at least
  $\Delta_0 + (\alpha \sqrt{n \log n})/4$ with high probability. 
  To do this, first notice that using the upper bound on $p(by)$
  from \eqref{eq:pby-1}, we have in expectation 
  \begin{align}
    \E[S_{by}] = \lambda \cdot p(by)
    \le \frac{\lambda}{2} - \frac{\lambda \Delta_0}{4n}.
  \end{align}
  As long as $S_{by} < \frac{\lambda}{2} - \frac{5\Delta_0}{8}$,
  then $S_{bx} - S_{by} > \frac{5\Delta_0}{4} \ge \Delta_0 + 
  (\alpha \sqrt{n \log n})/4$ as required. 
  Using an upper Chernoff bound, we find that the probability of
  the event $S_{by} \ge \frac{\lambda}{2} - \frac{5\Delta_0}{8}$ is
  bounded by
  \begin{align}
      \Pr\left[S_{by} \ge \frac{\lambda}{2} - \frac{5\Delta_0}{8}\right]
      &=
      \Pr\left[S_{by} \ge \E[S_{by}] 
      \left(1 + \frac{\Delta_0}{\E[S_{by}]}
      \left(\frac{8\lambda - 20n}{32n}\right)
      \right)\right] \\
      &\le
      \exp\left(
        -\frac{1}{3}\cdot
        \frac{\Delta_0^2 (8\lambda-20n)^2}{(32n)^2}
        \cdot
        \frac{4}{\lambda}
      \right) \\
      &\le
      \exp\left(
      - 1.05 \cdot \alpha^2 \log n 
      \right) \\
      &\le n^{-\alpha}
  \end{align}
  where in the penultimate inequality we use the fact that
  $\Delta_0 \ge \alpha \sqrt{n \log n}$ and $32 n \lambda \le 66n$. 
  
  Thus within the $32 n \lambda \le 66n$ blank-consuming steps 
  that occur among the greater sequence of $66n$ productive steps,
  the progress measure $\hat x - \hat y$ will increase
  by at least $5\Delta_0/4$ with high probability. This is enough 
  to both offset the maximum decrease in $\hat x - \hat y$ 
  incurred by the leak events, while also allowing the 
  progress measure to double from its original value with high 
  probability. 

  To complete the proof, we will compute the number of
  total interaction events needed to obtain with high probability
  the requisite $66n$ productive steps that can be used to
  complete a stage of Phase 1 correctly. We denote by
  $p(prod)$ the probability that the next event is any productive step,
  and so
  \begin{align}
    p(prod)
    =
    \frac{(1-\beta) \phi}{\binom{n}{2}} + \beta
    \ge
    \frac{(1-\beta) \phi}{n^2/2},
    \label{eq:dbl-a1.2-prod}
  \end{align}
  where $\phi = b(x+y) + xy \ge x\cdot \hat y$. Since $\hat y \ge n/6$
  throughout all stage of Phase 1, and since by Lemma~\ref{lemma:dbl-x-lb}
  $x \ge n/4$ holds with probability at least $1 - n^{-\alpha}$ throughout
  Phase 1 of the protocol, it follows that $\phi \ge n^2/24$ with probability
  at least $1 - 2n^{-\alpha} - n^{-(\alpha/6)}$ throughout the first phase.
  Substituting this back into \eqref{eq:dbl-a1.2-prod} gives
  \begin{align}
    p(prod)
    \ge
    \frac{(1-\beta)\phi}{n^2/2}
    &\ge
    \frac{(1-\beta)(n^2/24)}{n^2/2}\\
    &\ge
      \frac{(2/3)(n^2/24)}{n^2/2}
      = \frac{1}{18}
  \end{align}
  where the final inequality is due to
  $\beta \le 1/3$ when $\beta \le O(\sqrt{n \log n}/n)$, which holds
  for sufficiently large $n$.

  Thus in a sequence of $2376n$ total steps, the expected number of
  productive events is at least
  \begin{align}
    \E[prod] = 
    2376 n \cdot p(prod)
    \ge 2376 n \cdot \frac{1}{18}
    = 132 n.
  \end{align}
  Then by applying an upper Chernoff bound, we can see that
  \begin{align}
    \Pr[\; prod \le 66n \;]
    =
    \Pr[\; prod \le (1/2) \cdot \E[prod]\;] 
    &\le
    \exp\left(
    -\frac{132n}{8}
    \right) \\
    &\le
    \exp\left(-n \right)
    \le n^{-\alpha}
  \end{align}
  for any $\alpha \ge 1$ for sufficiently large $n$.
  Thus assuming the lower bound on $p(prod)$, with probability 
  at least $1 - n^{-\alpha}$, the $66n$
  productive steps needed to complete a stage whp of Phase 1 are
  obtained within $2376n$ total steps.

  Now summing all error probabilities and taking a union bound,
  we conclude that $2376n$ total steps are sufficient to complete
  each stage of Phase 1 correctly with probability at least 
  $1 - 6n^{-\alpha} - 3n^{-(\alpha/6)}$.
\end{proof}


\paragraph*{Phase 2 Behavior}
The next set of lemmas give analogous leak event and total step bounds
to show the successful completion of stages of Phase 2 in the protocol,
despite the presence of adversarial leaks.

\begin{lemma}
  \label{lemma:dbl-a2.1-largeB}
  During Phase 2 of the $\dbam$ protocol on a population with adversarial
  leak rate $\omega(\log n/n) \le \beta \le \frac{\alpha \sqrt{n \log n}}{12672n}$,
  starting at any point during the stage that begins with $\hat y = n/k$
  for $6 \le k \le 1/(50\beta)$, the number of leak events among a sequence
  of $90n/k$ productive steps is no more than $6n/k$
  with probability at least $1 - 2n^{-a}$, for $a \ge 1 $ when
  $n$ is sufficiently large.
\end{lemma}

\begin{proof}
  Recall that when the leak rate $\beta$ is large, a stage of Phase 2
  begins with $\hat y = n/k$ for $6 \le k \le 1/(50\beta)$ and ends once
  $\hat y$ decreases by a factor of 2.
  
  These means that, throughout the stage that begins with $\hat y = n/k$,
  the quantity
  \begin{align}
    \phi = b(x+y) + xy \ge x\cdot \hat y
    \ge \frac{n}{3} \cdot \frac{n}{2k}
    = \frac{n^2}{6k}.
    \label{eq:a2.1-large-phi}
  \end{align}
  Here, the inequality holds from Lemma~\ref{lemma:dbl-x-lb2}, which says
  that $x \ge n/3$ with probability at least $1 - n^{-a}$ throughout
  Phase 2 of the protocol.

  As before, letting $p(l)$ denote the probability of a leak event
  conditioned on any productive step, we then have
  \begin{align}
    p(l) \le
    \frac{(n^2/2)\beta}{(n^2/3)\beta + (1-\beta)\phi}
    &\le
      \frac{(n^2/2)\beta}{(n^2/3)\beta + n^2/6k - (n^2/6k)\beta} \\
    &\le
      \frac{(n^2/2)\beta}{(11n^2/36)\beta + n^2/6k} \;\;,
  \end{align}
  where the final inequality holds given that we are assuming $k\ge 6$.
  Simplifying further then gives
  \begin{align}
    p(l) \le \frac{1/2}{11/36 + 1/(6k\beta)}
    \le
    \frac{1/2}{1/(6k\beta)}
    \le 3 k \beta.
    \label{eq:a2.1-pl}
  \end{align}
  
  Now, consider a sequence of $90n/k$ total productive events.
  Letting $l$ denote the number of leak events among
  this sequence, we then have
  \begin{align}
    \E[l] = p(l) \cdot \frac{90n}{k}
    \le (3k\beta) \cdot \frac{90n}{k}
    = 270 n \beta.
  \end{align}
  Using an upper Chernoff bound then shows that
  \begin{align}
    \Pr[\; l > 300 n \beta \;]
    &=
      \Pr[\; l > (1 + 1/9) \cdot \E[l] \;] \\
    &\le
      \exp\left(- \frac{1}{3\cdot 81}\cdot 270 n \beta\right) \\
    &\le
      \exp\left(-n\beta\right).
  \end{align}
  Given the assumption that $\beta \ge \omega(\log n /n)$, it follows
  that $n \cdot \beta \ge \omega(\log n)$, which means that
  we can further write
  \begin{align}
    \Pr[\; l > 300n\beta \;] \le \exp(- \omega(\log n)) \le n^{-a}
  \end{align}
  for $a \ge 1$ when $n$ is sufficiently large.
  Since we also only consider $6 \le k \le 1/(50\beta)$, it follows that
  $300n \beta \le 6n/k$. 
  
  Thus summing over all error probabilities (including the one
  used to derive the lower bound on $\phi$ in \eqref{eq:a2.1-large-phi}),
  we have that with probability at least $1 - 2n^{-a}$,
  the number of leak events within a sequence of $90n/k$ total productive
  steps is at most $6n/k$. 
\end{proof}

\begin{lemma}
  \label{lemma:dbl-a2.1-smallB}
  During Phase 2 of the $\dbam$ protocol with adversarial leak rate
  $\beta \le a \log n/n$, starting at any time during the
  stage that begins with $\hat y = n/k$ for $6 \le k \le n/(50 a \log n)$,
  the number of leak events among a sequence of \ $90n/k$ productive steps
  is no more than $6n/k$ with probability at least $1 - 2n^{-a}$ for
  $a \ge 1$ and sufficiently large $n$. 
\end{lemma}

\begin{proof}
  Recall from expression ~\eqref{eq:a2.1-pl} in the proof of
  Lemma~\ref{lemma:dbl-a2.1-largeB} that for the stage of Phase 2
  beginning at $\hat y = n/k$, the probability of a leak event
  conditioned on an productive step is
  \begin{align}
    p(l)  \le 3 \beta k \le \frac{3 a \log n}{n} \cdot k,
  \end{align}
  since we are assuming $\beta \le (a \log n) / n$, where
  the bound holds with probability at least $1 - n^{-a}$. 

  Consider now a sequence of $90n/k$ total productive steps. Again using
  $l$ to denote the number of leak events that occur among this sequence,
  we have
  \begin{align}
    \E[l] = \frac{90n}{k} \cdot p(l)
    \le \frac{90n}{k} \cdot \frac{3a \log n }{n} \cdot k
    \le 270 a \log n.
  \end{align}

  Using an upper Chernoff bound then shows that the probability of
  having more than $300 a \log n$ leak events in this sequence is at most
  \begin{align}
    \Pr[\; l > 300 a \log n \;]
    &=
      \Pr[\; l > (1 + 1/9) \E[l] \;] \\
    &\le
      \exp\left(
      - \frac{1}{3 \cdot 81} \cdot 270 a \log n
      \right) \\
    &\le
      \exp\left(
      - a \log n
      \right)
      \le n^{-a}.
  \end{align}
  
  Since we assume that $6 \le k \le n/(50 a \log n)$, it follows that
  $300 a \log n \le 6n/k$.
  Then summing over all error probabilities,
  we have that the number of leak events within the sequence of
  $90n/k$ productive steps is at most $6n/k$ with probability
  at least $1 - 2n^{-a}$. 
\end{proof}


\begin{lemma}
  \label{lemma:dbl-a2.2-largeB}
  During the $\dbam$ protocol on a population with adversarial leak rate
  $\omega(\log n /n) \le \beta \le (\alpha\sqrt{n \log n})/(12672n)$,
  each stage of Phase 2 with initial value $\hat y = n/k$ for
  $6 \le k \le 1/(50\beta)$ completes correctly within
  $1080n$ total steps with probability at least $1-6n^{-a}$ 
  for $a \ge 1$ when $n$ is sufficiently large. 
\end{lemma}

\begin{proof}

  Consider a sequence of $90n/k$ productive steps at the start of the stage
  when initially $\hat y = n/k$. By Lemma~\ref{lemma:dbl-a2.1-largeB},
  the number of leak events among such a sequence is at most $6n/k$
  with probability at least $1-2n^{-a}$.
  In order to ensure $\hat y$ decreases by a factor of two from its
  initial value, it is sufficient to show that among the remaining 
  at least $84n/k$ non-leak productive steps in the sequence, 
  the number of $X+B$ interactions exceeds the number of $Y+B$ by
  at least $2n/k + 6n/k = 8n/k$. 
  
  
  Consider the subsequence of $\lambda \le 90n/k$ blank-consuming 
  productive steps. Given that at least $84n/k$ of the productive steps 
  will be non-leak steps with high probability, Lemma~\ref{lemma:dbl-prob-prod}
  shows that $\lambda \ge 40n/k$, so $40n/k \le \lambda \le 90n/k$. 
  Now by Lemma~\ref{lemma:dbl2-no-double}, we have
  that the probability of a $X+B$ interaction conditioned on
  a blank-consuming step, which we denote by $p(bx)$ is at
  least $5/8$, which holds with probability at least $1 - n^{-a}$. 
  Then letting $S_{bx}$ and $S_{by}$ denote the the number of $X+B$
  and $Y+B$ interactions among the sequence of blank consuming steps
  respectively, we have in expectation that
  \begin{align}
    \E[S_{bx}] = p(bx) \cdot \lambda \ge \frac{5\lambda}{8}.
  \end{align}
  
  To ensure $S_{bx} - S_{by} \ge 8n/k$, it is sufficient to have
  $S_{bx} > \lambda/2 + 4n/k = (\lambda k + 4n)/2k$. 
  Using a lower Chernoff bound shows that the probability of the
  event $S_{bx} \le (\lambda k + 4n)/2k$ is bounded by
  \begin{align}
    \Pr\left[S_{bx} \le (\lambda k + 4n)/2k \right]
    &= 
    \Pr\left[
      S_{bx} \le \E[S_{bx}]
      \left(
      1 - \tfrac{1}{\E[S_{bx}]}\left(\tfrac{2\lambda k - 32n}{16k}\right)
      \right)
    \right] \\
    &\le
    \exp\left(
    - \frac{1}{2}\cdot
    \frac{(2\lambda k - 32n)^2}{16^2 k^2} \cdot
    \frac{8}{5\lambda}
    \right) \\
    &\le
    \exp\left(
    - 4 \cdot \frac{n}{k}
    \right) \\
    &\le
    \exp\left(
    - \omega(\log n)
    \right) \\
    &\le
    n^{-a}
  \end{align}
  for any $a \ge 1$ when $n$ is sufficiently large.
  Here, in the third inequality we use the fact that $40n/k \le \lambda \le 90n/k$, 
  and in the fourth inequality we use the assumption that 
  $k \le 1/(50\beta)$ and $\beta \ge \omega(w\log n/ n)$. 
  So within the subsequence of $\lambda$ blank-consuming productive
  steps, we have that $S_{bx} - S_{by} \ge 8n/k$ holds with high
  probability as needed.

  To complete the proof, we compute the number of \textit{total}
  steps required to obtain with high probability the $90n/k$
  productive events needed to complete the stage of Phase 2
  correctly.
  For this, note that the probability that the next
  interaction event is a productive step, which we denote
  by $p(prod)$ is
  \begin{align}
    p(prod) = \frac{(1-\beta)\phi}{\binom{n}{2}} + \beta
    \ge
    \frac{\phi - \beta \phi}{n^2/2}
    \label{eq:a2.2-pprod-large}
  \end{align}
  where $\phi = b(x+y) + xy \ge x\hat y$.
  Throughout the stage beginning with $\hat y = n/k$, we have
  $\hat y \ge n/(2k)$, and by Lemma~\ref{lemma:dbl-x-lb2}
  $x \ge n/3$ with probability at least $1-n^{-a}$ throughout
  the entirety of Phase 2. Together, this means that
  $\phi \ge n^2/6k$, and substituting back into ~\eqref{eq:a2.2-pprod-large}
  gives
  \begin{align}
    p(prod) \ge
    \frac{n^2/6k - \beta n^2/(6k)}{n^2/2}
    \ge
    \frac{n^2/12k}{n^2/2}
    =
    \frac{1}{6k},
  \end{align}
  where the final inequality holds from observing $\beta n^2/(6k) \le n^2/12k$
  for sufficiently large $n$ when $\beta \le O(\sqrt{n \log n}/n)$.
  So for the stage of phase 2 starting at $\hat y = n/k$, in a sequence
  of $1080n$ total interactions, the expected number of productive
  steps, denoted by $\E[prod]$ is at least
  \begin{align}
    \E[prod] \ge 1080n \cdot p(prod)
    \ge  \frac{1080}{6k} = \frac{180n}{k}.
  \end{align}
  By an upper Chernoff bound then, we have
  \begin{align}
    \Pr[\; prod < 90n/k \;]
    &=
      \Pr[\; prod < (1/2) \cdot \E[prod] \;]\\
    &\le
      \exp\left(
      - \frac{180}{8}\cdot \frac{n}{k}
      \right) \\
    &\le
      \exp\left(
      - \omega(\log n)
      \right)
      \le n^{-a}
  \end{align}
  for $a \ge 1$ when $n$ is sufficiently large $n$. So
  at least $90n/k$ in a sequence of $1080n$ total interactions will
  be productive events with high probability. 

  Now, by summing over all error probabilities and taking a union bound,
  we have that with large leak rate $\beta$, each stage of Phase 2 will
  complete correctly within $1080n$ total interactions with probability
  at least $1 - 6n^{-a}$ for $a \ge 1$ when $n$ is sufficiently large. \qedhere
\end{proof}

For the case when the leak rate $\beta$ is small (i.e. when
$\beta \le O(\log n/n)$, we can similarly show that each stage
of Phase 2 completes within $O(n)$ total steps with
high probability via the following lemma.

\begin{lemma}
  \label{lemma:dbl-a2.2-smallB}
  During the $\dbam$ protocol on a population with adversarial leak
  rate $\beta \le a \log n /n$, each stage of Phase 2 with initial
  value $\hat y = n/k$ for $6 \le k \le n/(50 a \log n)$ completes
  correctly within $1080n$ total steps with probability at least
  $1 - 6n^{-a}$ for $a \ge 1$ when $n$ is sufficiently large. 
\end{lemma}

The proof of this lemma uses the bound on the number of leak
events from Lemma~\ref{lemma:dbl-a2.1-smallB} for small $\beta$
and is nearly identical to that of Lemma~\ref{lemma:dbl-a2.2-largeB}
and is thus omitted. 

\paragraph*{Concluding the Proof}
Using the preceding lemmas, we have the following proof of
Theorem~\ref{thm:dbl-main}, which characterizes the behavior
of the $\dbam$ protocol in the presence of adversarial leaks.
We restate the theorem for convenience.

\thmdblmain*

\begin{proof}
  \noindent
  By Lemma~\ref{lemma:dbl-a1.2}, at most $O(n)$ total
  interactions are needed to complete each of the $O(\log n)$
  stages of Phase 1 correctly with probability at least
  $1 - 6n^{-\alpha} - 3n^{-(\alpha/6)}$.
  And by Lemmas~\ref{lemma:dbl-a2.2-largeB} and~\ref{lemma:dbl-a2.2-smallB}
  for both the small and large $\beta$ cases, at most
  $O(n)$ total interactions are needed to complete each
  stage of Phase 2 correctly with probability at least
  $1 - 6n^{-a}$ for any $a \ge 1$. 
  Following the correct completion of Phase 2, the
  protocol will have reached a configuration with
  sample error at most $O(\hat y /n) = O(\beta)$ when 
  $\beta$ is large, and at most $O(\log n/n)$ when
  $\beta$ is small.
  Thus setting the error parameters appropriately for
  each phase and taking a union bound, it follows that for  
  any $c \ge 1$ and sufficiently large $n$, there 
  exists some constant $\alpha \ge 1$ such that the
  protocol reaches a configuration with the 
  specified bounded sample error within $O(n \log n)$
  total steps with probability $1 - n^{-c}$ when
  the input margin is at least $\alpha \sqrt{n \log n}$. 
\end{proof}

%

\paragraph*{Long-term Behavior}

Theorem~\ref{thm:dbl-main} says that the protocol will reach 
a low-sample-error configuration within $O(n \log n)$ steps
with high probability in the presence of leaks. We further show
that the protocol remains in such a configuration for at least
a polynomial number of steps with high probability. To do this, we again 
analyze the protocol separately under the small and large leak rate regimes. 
In both cases, we show that following the completion of Phase 2, 
the value of $\hat y$ never fluctuates too far from its final value
with high probability within a short sequence of productive steps. 
Taking a union bound and tuning the degree of fluctuation lets
us show that this behavior persists for at least a polynomial number of 
rounds with all but polynomially-small error. 

We proceed by analyzing this behavior for the small leak rate case. 
Recall that with small leak rate $\beta \le \alpha \log n$, Phase 2 of
the protocol ends once $\hat y = y + b/2 \le 25 \alpha \log n$ for some $\alpha > 1$.
The next lemma shows that following the end of Phase 2, the value of $\hat y$
remains $O(\log n)$ throughout a short sequence of productive steps with 
high probability. 


\begin{lemma}
  \label{lemma:dbl-phase3-helper-sb}
  Suppose $\hat y = 26 \alpha_1 \log n$ for some $\alpha_1 \ge \alpha$
  following the end of Phase 2 of the $\dbam$ protocol with leak rate
  $\beta \le (\alpha \log n) / n$, and consider
  a sequence of $80 \alpha_1 \log n$ productive steps. Then
  \begin{enumerate}[i.]
  \item
    the maximum value of $\hat y$ throughout this sequence
    is at most $38 \alpha_1 \log n$
  \item
    $\hat y$ decreases to $25 \alpha_1 \log n$ by the
    end of the sequence
  \end{enumerate}
   both with probability at least $1 - 2n^{-0.66\alpha_1}$
   for sufficiently large $n$. 
\end{lemma}

\begin{proof}
  Let $S_l, S_{bx}, S_{by}$ and $S_{xy}$ be random variables denoting
  the number of leaks, $X+B$, $Y+B$, and $X+Y$ interactions respectively
  during the sequence of $80\alpha_1 \log n$ productive steps starting
  from $\hat y_0 = 26 \alpha_1 \log n$.
  We can observe the following:
  \begin{enumerate}
  \item
    
    If $S_{bx} - S_{by} - 2S_{l} \ge 2\alpha_1 \log n$, then $\hat y$
    must decrease to $25 \alpha_1 \log n$ by the end of the
    sequence of productive steps. This follows directly from the
    transition rules of the $\dbam$ protocol, as each $X+B$ interaction
    decreases $\hat y$ by $1/2$, each $Y+B$ interaction
    increases $\hat y$ by $1/2$, and each $X \rightarrow Y$ leak event
    increases $\hat y$ by $1$. 
  \item
    Letting $\hat y_{\max}$ denote the maximum value of $\hat y$ throughout
    this sequence, we have:
    $$\hat y_{\max} \le \hat y_0 + S_l + 0.5 S_{by}.$$
    Again, this follows directly from the transition rules of the protocol:
    each leak event increases $\hat y$ by 1, each $Y+B$ interaction
    increases $\hat y$ by $1/2$, and no other event causes $\hat y$ to
    increase. So no matter the order of the interactions within the
    sequence, the stated bound on $\hat y_{\max}$ holds. 
    
  \end{enumerate}
  Now suppose that throughout the sequence of $80 \alpha_1 \log n$ productive
  steps that
  \begin{align}
    \begin{split}\label{eq:dbl-sb:prod-steps}
      S_l &\le 4 \alpha_1 \log n \\
      S_{by} &\le 7 \alpha_1 \log n \\
      S_{bx} &\ge 17 \alpha_1 \log n
    \end{split}
  \end{align}
  Then by the first observation,
  $S_{bx} - S_{by} - 2S_{l}\ge (17 - 7 - 8) \alpha_1 \log n = 2 \alpha_1 \log n$,
  meaning that $\hat y$ must decrease to $25 \alpha \log n$ by the
  end of the sequence,
  and by the second observation,
  $$
  \hat y_{\max} \le (26 + 6 + 5.75) \alpha_1 \log n \le 38 \alpha_1 \log n.
  $$
  Thus to prove the lemma statement, it remains to show that
  the bounds on $S_l$, $S_{bx}$ and $S_{by}$ from \eqref{eq:dbl-sb:prod-steps}
  hold with high probability. \\

  \noindent
  \textit{Bound on $S_l$:}
  Let $p(l)$ denote the probability of a leak event conditioned
  on a productive step, and recall that
  \begin{align}
    p(l) = \frac{\beta \binom{n}{2}}{\beta \binom{n}{2} + (1-\beta)\phi}
    \label{eq:dbl-sl:pl}
  \end{align}
  where $\phi = b(x+y) + xy \ge x(b+y) \ge x\hat y$.
  Now observe the following inequalities which will be used to bound $p(l)$:
   \begin{enumerate}[i.]
  \item
    $\beta \binom{n}{2} \le 0.5 \alpha n \log n$, 
    which follows from $\beta \le \alpha \log n /n$ and
    $\binom{n}{2} \le 0.5 n^2$.
  \item
    $1-\beta \ge 0.95$, which holds since
    $\beta \le \alpha \log n/n \le 1/20$ for any $\alpha \ge 1$ when
    $n$ is sufficiently large.
  \item
    $x \ge 0.95n$, because while $\hat y \le 38\alpha_1 \log n$,
    it follows that $\hat x \ge n - 38 \alpha_1 \log n$.
    Since $\hat x = x + b/2$, this implies
    $x \ge n - 38 \alpha_1 \log n - b/2$. 
    Because $b \le 2 \hat y \le O(\alpha_1 \log n)$ while $\hat y \le 38 \alpha_1 \log n$,
    it follows that $38 \alpha_1 \log n + b/2 = O(\alpha_1 \log n) \le 0.05n$,
    which holds for any $\alpha_1 \ge 1$ when $n$ is sufficiently large. 
  \item
    $(1-\beta)\phi \ge 21 \alpha n \log n$, which follows from
    $\phi \ge x \hat y \ge x \cdot 25 \alpha_1 \log n \ge x \cdot 25 \alpha \log n$
    and by applying inequalities (ii) and (iii).
  \end{enumerate}
  Substituting these inequalities into \eqref{eq:dbl-sl:pl} shows that
  \begin{align}
    p(l) \le \frac{0.5 \alpha n \log n}{21 \alpha n \log n} < \frac{1}{40}
    \label{eq:dbl-sl:pl-2}
  \end{align}
  It follows that throughout the the sequence of $80 \alpha_1 \log n$
  productive steps, the expected number of leak events is
  $\E[S_l] < (1/40) \cdot (80 \alpha_1 \log n) = 2 \alpha_1 \log n$.
  Using a Chernoff bound then shows that
  \begin{align*}
    \Pr[S_l \ge 4 \alpha_1 \log n]
    &= \Pr[S_l \ge \E[S_l] \cdot 2] \\
    &\le \exp\left(
      - \frac{1}{3}\cdot 2 \alpha_1 \log n
      \right) \\
    &\le \exp\left(
      -  0.66 \alpha_1 \log n 
      \right) \\
    &\le n^{-0.66\alpha_1}.
  \end{align*}
  Thus $S_l \le 6\alpha_1 \log n$ with probability all but $n^{-0.66\alpha_1}$. \\

  \noindent
  \textit{Bound on $S_{by}$:}
  Let $p(y)$ denote the probability of a $B+Y$ interaction conditioned
  on a blank-consuming step, which we note is an upper bound on the
  probability of a $B+Y$ interaction conditioned on \textit{any}
  productive interaction. Recall that
  \begin{align}
    p(by)
    = \frac{by}{bx + by}
    &= \frac{y}{x + y} \nonumber \\
    &\le \frac{y + b/2}{x + y +b}
      = \frac{\hat y}{n}
    \label{eq:dbl-sl:pby}    
  \end{align}
  where the inequality holds when $y \le x$.
  Using inequality (iii) from the bound on $S_l$, it follows that
  $\hat x \ge x \ge 0.95n$ and thus $\hat y \le 0.05 n$ for sufficiently
  large $n$.
  Substituting this into \eqref{eq:dbl-sl:pby} gives
  $p(by) \le (0.05n)/n \le 0.05$.

  The expected value of $S_{by}$ over the sequence of productive steps
  can then be bounded from above by
  $\E[S_{by}] \le p(by) \cdot 80\alpha_1 \log n = 4 \alpha_1 \log n.$
  Again using a Chernoff bound, we then find that
  \begin{align*}
    \Pr[S_{by} \ge  7 \alpha_1 \log n]
    &\le
      \Pr[S_{by} \ge (4\alpha_1 \log n) \cdot (1 + 3/4)] \\
    &\le
      \exp\left(
      - \frac{1}{3} \cdot \frac{9}{16} \cdot 4 \alpha_1 \log n
      \right) \\
    &\le      
      \exp\left(
      - \frac{3}{4} \cdot \alpha_1 \log n
      \right) \\
    &\le
      n^{-0.75\alpha_1}
  \end{align*}
  Thus $S_{by}$ is at most $11 \alpha_1 \log n$ throughout the sequence
  of $80 \alpha \log n$ productive steps
  with probability all but $n^{-0.75\alpha_1}$. \\

  \noindent
  \textit{Bound on $S_{bx}$:}
  Assuming $S_l \le 6\alpha_1 \log n$ throughout the sequence of
  $80 \alpha_1 \log n$ productive steps, it follows that
  the number of non-leak interactions is
  \begin{align}
    \lambda = S_{bx} + S_{by} + S_{bx} &= 80\alpha_1 \log n - S_l \\
                                     &\ge 74 \alpha_1 \log n.
  \end{align}
  Moreover, Lemma~\ref{lemma:dbl-prob-prod} says that within a subsequence of
  $\lambda$ non-leak productive steps at least $(\lambda - y_0)/2$
  of the interactions must be blank-consuming. This means
  \begin{align*}
    S_{by} + S_{bx}
    &\ge \frac{1}{2} \cdot (74 \alpha_1 \log n - 26 \alpha_1 \log n) \\
      &= 24 \alpha_1 \log n.
  \end{align*}
  Now, if $S_{by} \le 7 \alpha_1 \log n$, this implies that
  \begin{align*}
    S_{bx} &\ge 24 \alpha_1 \log n - S_{by} \\
           &\ge 24 \alpha_1 \log n - 7 \alpha_1 \log n \\
           &= 17 \alpha_1 \log n.
  \end{align*}
  So throughout the seqeuence of $80 \alpha_1 \log n$ productive steps,
  $S_{bx} \ge 17 \alpha_1 \log n$ so long as
  the bounds on $S_{by}$ and $S_{l}$ also hold.
  Taking a union bound over both error probabilities gives
  the stated claim with probability at least $1 - 2n^{-0.66\alpha_1}$. 
\end{proof}

We can show the protocol follows similar behavior with large leak rate. 
Specifically, the next lemma states that when 
$\omega(\log n /n) \le \beta \le O(\sqrt{n \log n}/n)$,
the value of $\hat y$ remains $O(\beta n)$ throughout a short 
sequence of productive steps with high probability following the
completion of Phase 2. 
The proof of the lemma is nearly identical to
that of Lemma~\ref{lemma:dbl-phase3-helper-sb} and is thus omitted.

\begin{lemma}
\label{lemma:dbl-phase3-helper-lb}
  Suppose $\hat y = 26 \alpha_1 \beta n$ for some $\alpha_1 \ge \alpha$
  following the end of Phase 2 of the $\dbam$ protocol with leak rate
  $\omega(\log n /n) \le \beta \le (\alpha \sqrt{n\log n}) / n$,
  and consider a sequence of $80 \alpha_1 \beta n$ productive steps. Then
  \begin{enumerate}[i.]
  \item
    the maximum value of $\hat y$ throughout this sequence
    is at most $38 \alpha_1 \beta n$
  \item
    $\hat y$ decreases to $25 \alpha_1 \beta n$ by the
    end of the sequence
  \end{enumerate}
  both with probability at least $1 - 2n^{-0.66\alpha_1}$ 
  for sufficiently large $n$.
\end{lemma}

Using Lemmas~\ref{lemma:dbl-phase3-helper-sb} and \ref{lemma:dbl-phase3-helper-lb},
the following theorem states that following the completion 
of Phase 2, the protocol remains in a low-sample-error configuration 
for at least a polynomial number of steps with high probability.

%
%
\begin{theorem}
  \label{thm:dbl-phase3-convergence}
  Consider an execution of the $\dbam$ protocol 
  following the successful completion of Phase 2 of the protocol.

  \noindent
  Then given $\alpha_1 > \alpha$ and $0 < c < 0.6\alpha$,
  there is some $a > 0$ such that 
  \begin{enumerate}[i.]
    \item
    $\hat y \le 38 \alpha_1 \beta n$ when
    $\omega(\log n /n ) \le \beta \le (\alpha \sqrt{n \log n})/n$
    \item
    $\hat y \le 38 \alpha_1 \log n$ when $\beta \le (\alpha \log n)/n$
  \end{enumerate}   
  holds for at least the next $n^a$ steps
  with probability at least $1 - n^{-c}$ when $n$ is sufficiently large. 
\end{theorem}

\begin{proof}
  We prove the case for small leak rate (when $\beta \le \alpha \log n/n)$,
  which applies Lemma~\ref{lemma:dbl-phase3-helper-sb}. 
  The large leak rate case follows from an analogous application of 
  Lemma~\ref{lemma:dbl-phase3-helper-lb}.

  To this end, at the end of Phase 2 of the $\dbam$ protocol
  with leak rate $\beta \le (\alpha \log n)/n$, we have
  $\hat y \le 25 \alpha \log n$.
  Then consider the first time following the successful
  completion of Phase 2 that $\hat y$ increases to
  $26 \alpha_1 \log n$.
  By Lemma~\ref{lemma:dbl-phase3-helper-sb},
  $\hat y$ will decrease to $25 \alpha_1 \log n$
  within $O(\alpha_1 \log n)$ productive steps
  and will never exceed $38 \alpha_1 \log n$ throughout
  the sequence
  with probability at least $1 - 2n^{-0.66\alpha_1}$. 

  For every subsequent time that $\hat y$ increases
  to $26 \alpha_1 \log n$, the same behavior
  and error probability holds. Each such
  occurrence of this behavior takes at least 1 interaction,
  and thus by a union bound,
  the total error probability of $n^{0.6\alpha_1 - c}$
  such occurrences is at most
  $$
  \frac{2n^{0.6\alpha_1 - c}}{n^{0.66\alpha_1}}
  \le
  \frac{n^{0.6\alpha_1 - c}}{n^{0.6\alpha_1}} \le n^{-c}
  $$
  for sufficiently large $n$. 
  Setting $a = 0.6\alpha_1 - c$ then gives the
  stated result.   
\end{proof}



%
%


\subsection{Leak Robustness of the $\dbamc$ Protocol}

The analysis of the
previous subsection is adapted to show that
the $\dbamc$ protocol exhibits a similar form
of leak-robustness in the CI model,  
and in the following theorem we prove the case where $m = \Theta(n)$. 
Recall that in the CI model, only the non-catalytic worker agents
are susceptible to leak events. 

\begin{restatable}{theorem}{thmdbclmain}
  \label{thm:dbcl-main}
  There exist constants $\alpha, d \ge 1$ such that, for a population
  with $m = cn$ for $c \ge 4$ and input margin
  $\epsilon \ge \alpha \sqrt{N \log N}$, the $\dbamc$ protocol will reach
  a configuration with
  \begin{enumerate}
  \item
    sample error $O(\log N/N)$ when $\beta \le O(\log N / N)$
  \item
    sample error $O(\beta)$ when
    $\omega(\log N / N) \le \beta \le (\alpha \sqrt{N\log N})/dN$
  \end{enumerate}
  within $O(N \log N)$ total interactions with probability at
  least $1 - N^{-a}$ for $a \ge 1$ when $N$ is sufficiently large. 
\end{restatable}

The proof of the theorem uses the same 
progress measures and phase and stage structure 
introduced in the non-leak setting
in Section~\ref{sec:dbamc}, and the final sample-error 
guarantee of the protocol is again
defined with respect to the magnitude of the leak rate.  
The behavior of the protocol between the two classes of leak rate 
is similar as in the $\dbam$ analysis.
%

More formally, recall that 
we refer to \textit{productive steps} as any of the 
non-null transitions found in Figure~\ref{fig:dbamc},
and to \textit{non-leak productive steps} and 
\textit{blank-consuming productive steps} as in Section~\ref{sec:dbl}.
We say the leak rate $\beta$ is \textit{large} when
$\omega(\log N/N) \le O(\sqrt{N\log N}/N)$ and \textit{small} when
$\beta \le O(\log N/N)$. 
The adapted \textbf{phases} and \textbf{stages} used throughout 
the analysis are as follows:
\begin{enumerate}[i.]
\item
  \textbf{Phase 1} begins with the start of the protocol,
  and ends \textit{correctly} when $P = \epsilon + \hat x - \hat y \ge 7m/8$,
  where $\epsilon = i_x - i_y$ denotes the initial input margin of
  the population, and $\hat x = x + b/2$ and $\hat y = y + b/2$
  as used previously.

  Every stage of Phase 1 completes correctly with subsequent doublings
  of the progress measure $P$.
  Letting $\Delta_0$ denote the initial value of $P$, stage $t$ of
  Phase 1 begins when $P = \Delta_0 \cdot 2^t$ and ends once
  either $P \ge \Delta_0 \cdot 2^{t+1}$ or when $P \ge 7m/8 + \epsilon$.

\item
  \textbf{Phase 2} of the protocol begins once $P \ge 7m/8$, which
  is equivalent to when $\hat y \le m/16$.
  When $\beta$ is large, the phase completes correctly when
  $\hat y \le 300 \beta m$, and when $\beta$ is small, the phase
  completes correctly when $\hat y \le 300 a \log m$ for some $a \ge 1$.
  \begin{itemize}
  \item
    Every stage of Phase 2 completes correctly with a subsequent halving
    of $\hat y$. When $\beta$ is large, the stage of Phase 2 that begins
    with $\hat y = m/k$ where $16 \le k < 1/(600\beta)$  completes
    correctly when reaching $\hat y \le \max\{m/2k, 300\beta m\}$.
    The final stage of Phase 2 for large $\beta$ begins with
    $\hat y = 600\beta m$ and ends correctly once $\hat y \le 300 \beta m$.
  \item
    For small $\beta$, the stage of Phase 2 that begins
    with $\hat y = m/k$ where $16 \le k < m/(600 a \log m)$ and $a \ge 1$
    completes correctly when reaching $\hat y \le \max\{m/2k, 300 a\log m\}$.
    The final stage of Phase 2 for large $\beta$ begins with
    $\hat y = 600 a \log m$ and ends correctly once $\hat y \le 300 a \log m$.
  \end{itemize}
\end{enumerate}

\subsubsection{Proof of Theorem~\ref{thm:dbcl-main}}

We now proceed to build and prove the lemmas used to derive the 
main result in showing the robustness of the $\dbamc$ protocol
to transient leaks.

\paragraph*{Utility Lemmas}
Similar to the previous section, we start
by stating and proving several utility lemmas used 
throughout the analysis.

\begin{lemma}
  \label{lemma:dbcl-util-phase1-nodouble}
  Throughout the $\dbamc$ protocol on a population with $m = cn$ and
  initial input margin $\alpha \sqrt{N \log N}$, there exists some
  $d \ge 1$ such that when $\beta \le (\alpha \sqrt{N \log N})/dN$,
  the probability during the stage of Phase 1 starting at $P = \Delta_0$ that
  $P$ drops to $\Delta_0 / 2$ within a sequence of $200N$ productive steps
  before finishing the stage correctly is at most $N^{-\alpha/6}$, assuming
  that the number of leak events in the sequence is at
  most $(\alpha \sqrt{N \log N})/8$.
\end{lemma}

\begin{proof}
  If the number of leak events in the sequence of $200N$ productive
  steps is at most 
  $$(\alpha \sqrt{N \log N})/8 \le \Delta_0/8,$$
  then $P$ only further drops to $\Delta_0/2$
  if, in a subsequence of blank-consuming steps, the number
  of $B+I_y$ or $B+Y$ steps ever exceeds the number of
  $B+I_x$ or $B+X$ steps by an additional at $\Delta_0/4$.
  Recall from the non-leak analysis that up until the point $P$
  drops to $\Delta_0/2$, the conditional probability of a $P$-increasing
  blank-consuming step is at least $1/2 + \Delta_0/10N$.
  Thus using Lemma~\ref{lemma:one-dim-rw}, the probability of our
  bad event occurring within the sequence of steps is at most
  $$\exp\left(-\Delta_0^2/(5N+\Delta_0)\right) \le N^{-\alpha/6}.$$ \qedhere
\end{proof}

\begin{lemma}
  \label{lemma:dbcl-util-phase1-eventual}
  Throughout the $\dbamc$ protocol on a population with $m = cn$ and
  initial input margin $\alpha \sqrt{N \log N}$, there exists some
  $d \ge 1$ such that when $\beta \le (\alpha \sqrt{N \log N})/dN$,
  the number of $X$ worker agents in the population will, by the
  end of Phase 1, exceed $3m/8$, and it will remain above this value
  through the completion of Phase 2 with probability
  at least $1 - N^{-\alpha}$. Similarly, the number of $B$
  worker agents in the population will, by the end of Phase 1,
  be at most $m/4$, and it will remain below this value through
  the completion of Phase 2 with probability at least $1 - N^{-\alpha/6}$. 
\end{lemma}

\begin{proof}
  Recall that the protocol begins when $b = m$ and Phase 1
  ends once $\hat y \le m/16$, meaning that $b \le m/8$
  by the end of Phase 1. Thus at some point during Phase 1 of the
  protocol, the number of blank agents will drop and remain
  below $b \le m/4$ for the remainder of the Phase.

  Now if $b \le m/4$ then $\hat x + \hat y \ge 3m/4$, and
  Lemma~\ref{lemma:dbcl-util-phase1-nodouble} implies that
  throughout every stage of Phase 1, $\hat x \ge \hat y$ must
  hold with probability at least $1 - N^{-\alpha/6}$.
  Thus when $b \le m/4$ we must have $x \ge 3m/8$ with this same
  probability. 
\end{proof}

\begin{lemma}
  \label{lemma:dbcl-util-phase2-yhat}
  Throughout the $\dbamc$ protocol on a population with $m = cn$ and
  initial input margin $\alpha \sqrt{N \log N}$, there exists some
  $d \ge 1$ such that when $\beta \le (\alpha \sqrt{N \log N})/dN$,
  the value of $\hat y$ will remain below $m/6$ with probability
  at least $1 - N^{-\alpha}$ throughout a sequence of $200m/k$
  productive steps during the stage of Phase 2 beginning with $\hat y = m/k$
  when the number of leak events in the sequence is no more than $m/k$.
\end{lemma}

\begin{proof}
  By the end of Phase 1 of the protocol, we have $\hat x - \hat y \ge 7m/8$.
  So $\hat x - \hat y$ will remain above $2m/3$ throughout a stage of
  Phase 2 so long as $P$ doesn't decrease by at least an additional $m/12$
  (since the leak events decrease P by at most $m/8$ given that
  we consider $16 \le k$). 

  Using the fact that, until the point $\hat x - \hat y \ge 3m/4$
  we have the probability of an increase to $P$ conditioned on
  a blank-consuming step to be at least $1/2 + (3m/40N)$, it follows
  from an application of Lemma~\ref{lemma:one-dim-rw} that
  $\hat x - \hat y$ only drops to $2m/3$ during a sequence of
  $200m/k$ productive steps with probability at most $N^{-\alpha}$
  for $\alpha \ge 1$ when $N$ is sufficiently large.

  Because $\hat x + \hat y = m$, it follows that $\hat y \le m/6$
  when $\hat x - \hat y \ge 2m/3$, which finishes the proof. 
\end{proof}

\begin{lemma}
  \label{lemma:dbcl-util-prod-prodb}
  Throughout the $\dbamc$ protocol, in a sequence of $\lambda$ total
  non-leak productive steps, at least $(\lambda-y_0)/2$ will be
  blank-consuming steps, where $y_0$ is the initial number of
  $Y$ agents present in the population throughout the sequence. 
\end{lemma}

\begin{proof}
  The proof is identical to that of Lemma~\ref{lemma:dbl-prob-prod}
  from the leak analysis of the $\dbam$ protocol. \qedhere
\end{proof}


\paragraph*{Phase 1 Behavior}
The following two lemmas now characterize the number of 
total steps needed to complete each stage of Phase 1 of the protocol
correctly, despite the presence of leaks. 
Note that by Lemma~\ref{lemma:dbcl-util-phase1-eventual},
each stage of Phase 1 will either have $b \ge m/4$, or $b \le m/4$ and
$x \ge 3m/8$ with high probability. We will refer to these two scenarios
as case 1 and case 2 respectively, and we characterize the
following two lemmas in terms of both cases.

%
%
\begin{lemma}
  \label{lemma:dbcl-phase1-leaks}
  For a population where $m = cn$ and with input margin
  $\epsilon \ge \alpha\sqrt{N \log N}$ for $c, \alpha \ge 1$,
  there exists some $d \ge 1$ such that when $\beta \le (\alpha\sqrt{N\log N})/dN$,
  any point in Phase 1 of the $\dbamc$ protocol, the number of leak events
  within a sequence of $200N$ productive steps is at most
  $(\alpha \sqrt{N \log N})/8$ with probability at least $1 - N^{-\alpha} - N^{-\alpha/6}$.
  for sufficiently large $N$.
\end{lemma}

\begin{proof}
  Conditioned on a productive interaction, the probability of a leak
  event, denoted by $p(l)$ is given by
  \begin{align}
    p(l) \le \frac{\beta\binom{N}{2}}{\beta\binom{N}{2} + (1 - \beta)\phi},
    \label{eq:dbcl-pl1}
  \end{align}
  where $\phi = b(i_x + i_y) + b(x+y) + xy$.

  Recall that by Lemma~\ref{lemma:dbcl-util-phase1-eventual}, in the
  first case of a Phase 1 stage we have $b \ge m/4 = cn/4$,
  meaning that $\phi \ge bn \ge cn^2/4$.
  Thus in this first case, we can further simplify to find
  \begin{align}
    p(l) \le \frac{1/2}{1/4 + 1/(4\beta c)} = 2 \beta c.
  \end{align}

  Now in a sequence of $200N$ productive steps, the expected
  number of leak events in this first case is at most
  $400 N \beta c$. Applying an upper Chernoff bound further shows that
  the number of leak events within this sequence will only
  exceed $800 N \beta c$ with probability at most $N^{-\alpha}$.
  Then for any $\alpha, c \ge 1$, it follows that for some suitable choice
  of $d$ that
  \begin{align}
    800 \cdot N \cdot \frac{\alpha\sqrt{N \log N}}{dN}
    \le \frac{\alpha\sqrt{n \log N}}{8}.
  \end{align}

  In the second case given by Lemma~\ref{lemma:dbcl-util-phase1-eventual},
  we have $x \ge m/8$ with probability at least $1 - N^{-\alpha/6}$.
  Along with the fact that $\hat y \ge m/16$ throughout all stages
  of Phase 1 means that $\phi \ge x \hat y \ge (3/128)c^2 n^2$.
  Substituting this into \eqref{eq:dbcl-pl1} shows that in this
  second case,
  \begin{align}
    p(l) \le \frac{1/2}{1/4 + 3/(128\beta)} \le 143 \beta.
  \end{align}
    
  Now again in a sequence of $200N$ productive steps, the expected
  number of leak events in this case is at most
  $28600 \beta N$. When $\beta \le (\alpha \sqrt{N \log N})/dN$,
  setting $d$ sufficiently large and applying an upper Chernoff bound
  shows that the number of leaks among this sequence will only
  exceed $(\alpha \sqrt{N \log N})/8$ with probability at
  at least $N^{-\alpha}$ when $N$ is sufficiently large.

  Thus in either case of a Phase 1 stage, the number of leak events
  in a sequence of $200N$ productive steps will never exceed
  $(\alpha \sqrt{N \log N})/8$, and by summing over all error probabilities
  and taking a union bound, this behavior holds with probability at
  least $1 - N^{-\alpha} - N^{-\alpha/6}$. 
\end{proof}

%
%
\begin{lemma}
  \label{lemma:dbcl-phase1-steps}
  For a population where $m = cn$ and with input margin
  $\epsilon \ge \alpha \sqrt{N \log N}$ for $c, \alpha \ge 1$,
  there exists some $d \ge 1$ such that when
  $\beta \le (\alpha\sqrt{N\log N})/dN$,
  each stage of Phase 1 of the $\dbamc$ protocol completes correctly
  within $\max\{38000N, 3600Nc\}$ total interactions with probability
  at least $1 - N^{-\alpha} - N^{-\alpha/6}- N^{-\alpha/3}$ when
  $N$ is sufficiently large.
\end{lemma}

\begin{proof}
  Consider a stage of Phase 2 starting with 
  $P = \epsilon + (\hat x- \hat y) = \Delta_0$, and recall that
  the stage completes correctly once $P$ increases to
  $\max\{2\Delta_0, 7m/8 + \epsilon\}$.   
  By Lemma~\ref{lemma:dbcl-phase1-leaks},  when $\beta \le (\alpha \sqrt{N \log N}){dN}$,
  for suitable choice of $d \ge 1$, a sequence
  of $200N$ productive events will contain at most $(\alpha \sqrt{N \log N})/8$
  leak events with probability at least $1 - 2N^{-\alpha}$. This means
  that leak events decrease the value of $P$ by at most $(\alpha \sqrt{N \log N})/4$
  throughout the sequence with this same probability and choice of $d$. 
  
  We will show that despite these leaks, the stage still completes
  correctly in the remaining at least $199N$ non-leak productive events
  with high probability. 
  First, let $\lambda$ denote the number of blank-consuming
  productive steps within the sequence of $200N$ total productive steps. 
  Clearly $\lambda \le 200N$, and  by Lemma~\ref{lemma:dbcl-util-prod-prodb}
  we have that $\lambda \ge 98N$. 
  Letting $S_{bx}$ denote the number of $X+B$ or $I_X+B$ steps and
  letting $S_{by}$ denote the number of $Y+B$ or $I_Y + B$ steps
  throughout the subsequence of blank-consuming productive steps, we
  have that $\lambda = S_{bx} + S_{by}$.
  Our goal then is to show that $S_{bx} - S_{by} \ge 5\Delta_0/4$, 
  meaning that the progress measure $P$ increases enough to both offset the 
  effect of the leaks and to still double from its original value within the sequence
  of non-leak productive steps. Showing that $S_{bx} > \lambda/2 + 5\Delta_0/8$
  is sufficient to ensure that $S_{bx} - S_{by} \ge 5\Delta_0/4$, 
  and so we proceed to prove this bound holds with high probability. 
  
  First, recall from the non-leak analysis of the $\dbamc$ protocol 
  that, conditioned on a blank-consuming productive step, 
  the probability of a $X+B$ or $I_X+B$ step is at least
  $p(bx) \ge 1/2 + \Delta_0/(10N)$. Thus in expectation, we have
  $$\E[S_{bx}] \ge \frac{\lambda}{2} + \frac{\lambda\Delta_0}{10N} 
  = \frac{\lambda(10 N + \Delta_0)}{20N}.$$ 
  Using a lower Chernoff bound then shows that the probability
  of the event $S_{bx} \le \lambda/2 + 5\Delta_0/8$ is bounded by
  \begin{align}
      \Pr\left[
      S_{bx} \le \lambda/2 + 5\Delta_0/8
      \right]
      &=
      \Pr\left[
      S_{bx} \le \E[S_{bx}]
      \left(
      1 - \frac{\Delta_0}{\E[S_{bx}]}
      \left(
      \frac{8\lambda- 50N}{80N}
      \right)
      \right)
      \right] \\
      &\le
      \exp\left(
      - \frac{1}{2}
      \cdot 
      \frac{\Delta_0^2(8\lambda - 50n)^2}{(80N)^2}\cdot\frac{20N}{\lambda(10N + \Delta_0)}
      \right) \\
      &\le
      \exp\left(
      - 0.38 \alpha^2 \log N
      \right) \\
      &\le
      N^{-\alpha/3}
  \end{align}
  where we use the fact that $98N \le \lambda \le 200N$ and 
  $\Delta_0 \ge \alpha \sqrt{N \log N}$. 
  Thus with probability all but $N^{-\alpha/3}$, the $200N$ productive steps
  is sufficient to increase $P$ from $\Delta_0$ to at most $2\Delta_0$ 
  despite the leak events.
  




  To complete the proof, we compute the number of total interactions
  needed to obtain with high probability these $200N$ productive events.
  For this, let $p(prod)$ denote the probability that
  the next interaction is a productive event, where
  \begin{align}
    p(prod) \ge \frac{(1-\beta)\phi}{\binom{N}{2}} + \beta
    \ge
    \frac{(9/10)\phi}{N^2/2},
    \label{eq:dbl1-pprod}
  \end{align}
  where the inequality holds given that $\beta \le 1/10$ when $N$
  is sufficiently large. 

  Now recall that in the first case of a Phase 1 stage where
  $b \ge m/4$ that $\phi \ge cn^2/4$, and thus
  $p(prod) \ge (9/20)(c/(c+1)^2) \ge 1/(9c)$.
  Using an upper Chernoff bound, it follows that in a sequence
  of $3600Nc$ total steps, the probability of having fewer than
  $200N$ productive steps as at most $\exp(-50Nc) \le N^{-\alpha}$
  when $n$ is sufficiently large.

  In the second case of a Phase 2 stage where $x \ge m/4$,
  which by Lemma~\ref{lemma:dbcl-util-phase1-eventual} holds
  with probability at least $1 - N^{-\alpha/6}$, we have
  that $\phi \ge (3c^2n^2)/128$, and thus
  we find that $p(prod) \ge 1/95$.
  Again, using an Upper Chernoff bound shows that within
  a sequence of $38000N$ total steps, the probability of
  obtaining fewer than $200N$ productive steps is
  at most $\exp(-(38000/4)N) \le N^{-\alpha}$ for
  sufficiently large $N$.

  Summing all error probabilities and taking a union bound
  shows that for leak rate $\beta \le (\alpha \sqrt{N \log N})/dn$
  when $d \ge 1$ is chosen suitably, each stage of Phase 1
  will complete within at most
  $\max\{3600 N c \;,\; 38000N\}$
  with probability at least $1 - N^{-\alpha} -N^{-\alpha/6} - N^{-\alpha/3}$
  when $N$ is sufficiently large. 
\end{proof}

\noindent
\paragraph{Phase 2 Behavior} 
The next set of lemmas characterizes the number of productive 
and total steps needed to complete each 
stage of Phase 2 of the protocol correctly. Here, we make separate claims
based on whether the leak rate $\beta$ is large or small, however
the proofs involved for both cases are nearly identical in strategy. 
We first estimate the number of leak events that occur within
a sequence of steps during a stage of Phase 2 when the leak rate is large. 

%
%
\begin{lemma}
  \label{lemma:dbcl-phase2-leaks-largeB}
  During Phase 2 of the $\dbamc$ protocol on a population where $m = cn$
  for $c \ge 1$ and adversarial leak rate
  $\omega(\log N/N) \le \beta \le O(\sqrt{N \log N}/N)$,
  starting at any point during the stage that begins with $\hat y = m / k$
  for $16 \le k \le 1/(600\beta)$,
  the number of leak events in a sequence of $200m/k$ total productive
  steps is at most $m/k$ with probability at least $1-N^{-\alpha} - N^{-\alpha/6}$
  for $\alpha \ge 1$ when $N$ is sufficiently large. 
\end{lemma}

\begin{proof}
  The probability of leak event conditioned on any productive step,
  which we denote by $p(l)$, is given by
  \begin{align}
    p(l) = \frac{\beta \binom{N}{2}}{\beta \binom{N}{2}+ (1-\beta) \phi},
  \end{align}
  where again $\phi = bn + b(x+y) + xy \ge bn + x\hat y$.

  Recall that the stage of Phase 2 that begins  $\hat y = m/k$ and only finishes
  once $\hat y$ decrease to $m /2k$. Also, by Lemma~\ref{lemma:dbcl-util-phase1-eventual}
  we have that $x \ge 3m/8$ holds throughout Phase 2 of the protocol
  with probability at least $1-n{-\alpha/6}$.
  Thus throughout the stage it follows that $\phi \ge 3m^2/16k$.
  In turn, simplifying our expression for $p(l)$ leads to the bound
  $p(l) \le (16/16) k \beta$.

  Now consider a sequence of $200m/k$ productive steps. Letting $l$ denote
  the number of leak events within this sequence we have in expectation that
  $\E[k] < 540 m\beta$. Then applying an upper Chernoff bound shows that
  the probability that $l$ exceeds $600 m\beta$ is at most
  $\exp(-(1/3\cdot81)\cdot (540cn\beta)) \le n^{-\alpha}$ for $\alpha \ge 1$
  and sufficiently large $n$ given that $\beta = \omega(\log n /n)$.

  Futher observe that $600m\beta \le m/k$ when $k \le 1/(600\beta)$,
  which by our assumptions on $k$ is always satisfied.
  Summing over all error probabilities and taking a union bound then shows
  that the number of leaks within the sequence of productive steps
  as at most $m/k$ with probability at least $1 - N^{-\alpha}- N^{-\alpha/6}$.\qedhere
\end{proof}

%
%
\begin{lemma}
  \label{lemma:dbcl-phase2-leaks-smallB}
  During Phase 2 of the $\dbamc$ protocol on a population where $m = cn$
  for $c \ge 1$ and adversarial leak rate
  $\beta \le (a \log N) /N$ for $a \ge 1$, 
  starting at any point during the stage that begins with $\hat y = m / k$
  for $16 \le k \le m/(600 a \log m)$,
  the number of leak events in a sequence of $200m/k$ total productive
  steps is at most $m/k$ with probability at least $1-N^{-\alpha} - N^{-\alpha/6}$
  for $\alpha \ge 1$ when $N$ is sufficiently large. 
\end{lemma}

The proof of Lemma~\ref{lemma:dbcl-phase2-leaks-smallB}
is nearly identical to that of Lemma~\ref{lemma:dbcl-phase2-leaks-largeB},
and uses the distinct upper bounds on $\beta$ and $k$ to show the
result. For ease of readability, we omit the full proof. 
Now, in the next lemma, we estimate the number of total steps needed 
to complete each stage of Phase 2 correctly in the large leak rate setting. 

%
%
\begin{lemma}
  \label{lemma:dbcl-phase2-steps-largeB}
  During the $\dbamc$ protocol on a population where $m = cn$
  for $c \ge 4$ and adversarial leak rate
  $\omega(\log N / N) \le \beta \le O(\sqrt{N \log N} / N)$,
  each stage of Phase 2 with initial value $\hat y = m/k$
  for $16 \le k \le 1/(600\beta)$ completes within
  $1200cn$ total interactions with probability at least $1 - 5N^{-\alpha} - N^{-\alpha/6}$
  for $\alpha \ge 1$ when $N$ is sufficiently large. 
\end{lemma}

\begin{proof}

  Recall that by Lemma~\ref{lemma:dbcl-phase2-leaks-largeB} the
  number of leak events in a sequence of $200m/k$ productive steps
  is at most $m/k$ with probability at least $1-2N^{-\alpha}$.
  Our goal is to show that in the remaining at least $199m/k$
  non-leak productive steps, the stage of Phase 2 starting with
  $\hat y = m/k$ completes (decreases to $\hat y = m/2k$) with
  high probability. 
  
  Similar to the proof of Lemma~\ref{lemma:dbcl-phase1-steps}, let 
  $\lambda$ denote the number of blank-consuming
  productive steps within the sequence of $200m/k$ total productive steps. 
  Clearly $\lambda \le 200m/k$, and  by Lemma~\ref{lemma:dbcl-util-prod-prodb}
  we have that $\lambda \ge 98m/k$. 
  Now letting $S_{bx}$ denote the number of $X+B$ or $I_X+B$ steps and
  letting $S_{by}$ denote the number of $Y+B$ or $I_Y + B$ steps
  throughout the subsequence of blank-consuming productive steps, we
  have that $\lambda = S_{bx} + S_{by}$.
  To complete the stage correctly, our goal then is to show 
  that $S_{bx} - S_{by} \ge 2m/k + m/k = 3m/k$,
  since each blank-consuming interaction changes $\hat y$ 
  by a value of 1/2. To ensure this gap between 
  $S_{bx}$ and $S_{by}$ holds, it is sufficient to show that 
  $S_{bx} > \lambda/2 + 1.5m/k$ holds with high probability 
  throughout the sequence. 
  
  To prove this latter bound, recall that
  by Lemma~\ref{lemma:dbcl-util-phase2-yhat}, 
  throughout the stages of Phase 2 of the protocol, $\hat y \le m/6$
  holds with probability at least $1-N^{-\alpha}$ for $\alpha \ge 1$. 
  Letting $p(bx)$ denote the probability of a $X+B$ or $I_X+B$
  interaction, this means that $p(bx) \ge 1/2 + (2m/3)/10N$, 
  and when $c \ge 4$ it can be verified that
  $p(bx) \ge 1/2 + 7/320 = 167/320$.
  Thus among the subsequence of $\lambda$ blank-consuming productive
  steps, we have that 
  $\E[S_{bx}] \le \lambda 167/ 320$ with high probability. 

  Using a lower Chernoff bound then shows that
  the probability of the event $S_{bx} \le \lambda/2 + 1.5m/k$ 
  is bounded by
  \begin{align}
      \Pr\left[
      S_{bx} \le \lambda/2 + 1.5m/k
      \right]
      &= 
      \Pr\left[
      S_{bx} \le
      \E[S_{bx}]
      \left(
      1 - \frac{1}{\E[S_{bx}]}
      \left(
      \frac{14\lambda k - 960m}{640k}
      \right)
      \right)
      \right] \\
      &\le 
      \exp\left(
        -\frac{1}{2}
        \cdot
        \frac{(14\lambda k - 960m)^2}{(640k)^2}
        \cdot
        \frac{320}{167\lambda}
      \right) \\
      &\le
      \exp\left(
      - \frac{1}{505} \cdot \frac{m}{k}
      \right) \\
      &\le
      \exp\left(
      - 1.15 \cdot \beta m
      \right) \\
      &\le 
      \exp\left(
      - \omega(\log N)
      \right) \\
      &\le 
      N^{-\alpha}
  \end{align}
  for any $\alpha \ge 1$ when $N$ is sufficiently large. 
  Here, we use the fact that
  $98m/k \le \lambda \le 200m/k$, that $k \le 1/(600\beta)$
  and $\beta \ge \omega(\log N/N)$, and that $m = cn$. 
  So $S_{bx} > \lambda/2 + 1.5m/k$ will hold following the sequence
  of productive steps, and so $S_{bx} - S_{by} \ge 3m/k$ will hold 
  as required with probability all but $N^{-\alpha}$.
  
  To complete the proof, we must now compute the number of
  total interactions needed to obtain $200m/k$ productive
  steps with high probability.
  Recall from line~\eqref{eq:dbl1-pprod} in the Proof
  of Lemma~\ref{lemma:dbcl-phase1-steps} that
  the probability that the next event is a productive step,
  denoted by $p(prod)$ is at least $((9/10)\phi)/(N^2/2)$,
  where $\phi = bn + b(x+y) + xy \ge x\hat y$.

  Note that by Lemma~\ref{lemma:dbcl-util-phase1-eventual}
  we have $x \ge 3m/8$ with probability at least $1-N^{-\alpha/6}$
  throughout Phase 2, which means that throughout the stage
  beginning with $\hat y = m/k$ we have $\phi \ge 3m^2/(16k)$
  with high probability. This implies then that
  $p(prod) \ge (27/81)/k \ge 1/(3k)$ throughout the stage
  with high probability.

  Now, applying an upper Chernoff bound shows that in
  $1200m = 1200cn$ total steps, the probability of
  obtaining fewer than $200m/k$ productive steps is at most
  \begin{align}
    \exp\left(
    - 50 \frac{cn}{k}
    \right)
    \le
    \exp\left(
    - 30000 c \cdot \omega(\log N)
    \right)
    \le
    N^{-\alpha}
  \end{align}
  for sufficiently large $N$, where again we use the fact
  that $k \le 1/(600\beta)$, $\beta \ge \omega(N \log N)$,
  and $m = cn$ for $c \ge 4$. 

  Thus summing over all error probabilities and taking a union
  bound, we find that with probability at least
  $1-4N^{-\alpha} - N^{-\alpha/6}$,
  each stage of Phase 2 will complete correctly within
  $1200cn$ total steps when $N$ is sufficiently large. 
\end{proof}
  
\begin{lemma}
  \label{lemma:dbcl-phase2-steps-smallB}
  During the $\dbamc$ protocol on a population where $m = cn$
  for $c \ge 4$ and adversarial leak rate
  $\beta \le (a \log N)/N$ for $a \ge 1$,
  each stage of Phase 2 with initial value $\hat y = m/k$
  for $16 \le k \le m/(600a \log m)$ completes within
  $1200cn$ total interactions with probability at least
  $1 - 4N^{-\alpha} - N^{-\alpha/6}$ for $\alpha \ge 1$ when $N$ is
  sufficiently large. 
\end{lemma}

The proof of Lemma~\ref{lemma:dbcl-phase2-steps-smallB} is again
nearly identical to the proof of Lemma~\ref{lemma:dbcl-phase2-steps-largeB}
(save for the adjustments to the upper bound of $\beta$ and $k$),
and thus we omit the full proof for readability.


\noindent
\paragraph{Concluding the Proof} 
Finally, we use the preceding lemmas to prove the main
result of Theorem~\ref{thm:dbcl-main}, which characterizes 
the full behavior of the $\dbamc$ protocol in the presence of leaks. 
For convenience, we restate the theorem:

\thmdbclmain*

\begin{proof}
    For a population where $m = cn$, 
    by Lemma~\ref{lemma:dbcl-phase1-steps}, each of the 
    $O(\log N)$ stages of Phase 1 complete correctly
    within $O(cN)$ total steps with high probability
    when the constant $d$ is chosen appropriately. 
    By Lemmas~\ref{lemma:dbcl-phase2-steps-largeB} and
    ~\ref{lemma:dbcl-phase2-steps-smallB}, it follows
    that with high probability each of the $O(\log N)$
    stages of Phase 2 complete correctly within $O(N)$
    total steps for both small
    and large $\beta$ when $c\ge 4$.
    By setting the error parameters associated with 
    each stage of Phases 1 and 2 appropriately,
    it follows by a union bound that
    for sufficiently large $n$, there exist constants
    $\alpha, d \ge 1$ such that the protocol will reach 
    a configuration with the specified bounded sample error
    (according to the size of $\beta$)
    within $O(N \log N)$ total steps with probability
    at least $1 - N^{-a}$ for any $a \ge 1$. 
\end{proof}

%

\paragraph*{Long-term Behavior}

Similar to the $\dbam$ protocol for the original population model, 
we further show the long-term behavior of the $\dbamc$ protocol
in the CI model in the presence of leaks. 
In this setting, we analogously show that following the 
completion of Phase 2 of the protocol,
the value of $\hat y$ doesn't fluctuate too greatly for 
at least the next polynomially-many steps with all but
polynomially-small probability. 


We first prove the following lemma (analogous to 
Lemma~\ref{lemma:dbl-phase3-helper-sb}) which shows that 
with small leak rate $\beta$, the value of $\hat y$ remains 
$O(\log m)$ over a sequence of productive steps with high probability
after the completion of Phase 2 of the protocol. 

\begin{lemma}
 \label{lemma:dbcl-phase3-helper-sb}
 Suppose $\hat y = 301 a_1 \log m$ for some $a_1 \ge a$ following the 
 end of Phase 2 of the $\dbamc$ protocol with leak rate $\beta \le (a \log N)/N$
 on a population where $m = cn$ for $c \ge 4$. 
 Consider a sequence of $900 a_1 \log m $ productive steps. Then
 \begin{enumerate}
     \item 
     the maximum value of $\hat y$ throughout the sequence is
     at most $761 a_1 \log m$
     \item
     $\hat y$ decreases to $300a_1 \log m$ by the end of the sequence
 \end{enumerate}
 both with probability at least $1 - 2N^{-0.8a_1}$ when $N$ is sufficiently large. 
\end{lemma}

\begin{proof}
    The proof of the lemma is similar to that of 
    Lemma~\ref{lemma:dbl-phase3-helper-sb}. 
    Recall that when $\beta \le (a \log N)/N$ for some $a \ge 1$, 
    Phase 2 of the protocol ends when $\hat y \le 300 a \log m$.  
    Now, let $S_l, S_{bx},$ and $S_{by}$ be random variables denoting
    the number of leaks, $X+B$ or $I_X + B$, and $Y+B$ or $I_Y + B$
    interactions respectively during the sequence of $900 a_1 \log m$
    productive steps starting from $\hat y_0 = 301 a_1 \log m$. 
    We have the following two observations:
    \begin{enumerate}
    \item
    First, if 
    $$
    S_{bx} - S_{by} - 2S_l \ge 2 a_1 \log m,
    $$
    then $\hat y$ must decrease to $300 a_1 \log m$ by the end of 
    the sequence of productive steps. This follows directly 
    from the transition rules of the protocol, and the
    fact that a leak event increases the value of $\hat y$ by 1. 
   \item
    Now, letting $\hat y_{\max}$ denote the maximum value of $\hat y$ 
    throughout the sequence, and letting $\lambda$ denote the 
    number of blank-consuming steps among the sequence, observe that
    $$
    \hat y_{\max} \le \hat y_0 + S_l + 0.5 S_{by} 
    \le \hat y_0 + S_l + 0.5 \lambda,
    $$
    which again follows directly from the transition rules of
    the protocol. 
    Since at most $\lambda \le 900 a_1 \log m$, we have that
    $\hat y_{\max} \le 751 a_1 \log m + S_l$. 
    \end{enumerate}
    
    Using these two observations, we can see that if
    $S_{l} < 10a_1 \log m$ throughout the sequence of productive 
    steps, then claim (i) of the lemma is proven. 
    Moreover, if $S_l < 10a_1 \log m$, then it is sufficient to show
    that $S_{bx} - S_{by} \ge 22 a_1 \log m$ throughout the
    sequence to prove claim (ii) of the lemma. We will thus 
    proceed to proving these two bounds. 
    
    \noindent
    \textit{Bound on $S_l$}: 
    Let $p(l)$ denote the probability of a leak event conditioned
    on a productive step, and recall that 
    \begin{align}
        p(l) = \frac{\beta \binom{N}{2}}{\binom{N}{2} + (1-\beta) \phi}
    \end{align}
    where $\phi = bn + b(x+y) + xy \ge bn + x\hat y$. To bound $p(l)$ we
    observe the following inequalities:
    \begin{enumerate}[i.]
        \item 
        $\beta\binom{N}{2} \le 0.5 a N \log N$, which follows
        from $\beta \le a \log N/N$ and $\binom{N}{2} \le 0.5 N^2$. 
        \item
        $1-\beta \ge 0.95$, which holds since
        $\beta \le a \log N/N \le 1/20$ for any $a \ge 1$ and 
        sufficiently large $N$.
        \item
        $x \ge 0.95m$, since while $\hat y = y + b/2 \le O(\log m)$, then
        $x \ge m - O(\log m) \ge 0.95m$ when $m$ is sufficiently large. 
        \item
        $(1-\beta)\phi \ge 270 a m \log m$, which follows from
        $\phi \ge bn + x\hat y \ge x \cdot 300 a_1 \log m \ge x \cdot 300 a\log m$
        and from applying the inequalities in (ii) and (iii). 
    \end{enumerate}
    Using these inequalities, we can now bound $p(l)$ by
    \begin{align}
       p(l) \le \frac{0.5a N \log N}{270 a m \log m}
            \le \frac{1}{180},
    \end{align}
    which holds when $m = cn$  and $c \ge 4$. 
    It follows that throughout the sequence of $900 a_1 \log m$ 
    productive steps, 
    the expected number of leak events is 
    $\E[S_l] < (1/180) \cdot 900 a_1 \log m \le 5 a_1 \log m$. 
    Then using an upper Chernoff bound shows that 
    \begin{align}
        \Pr\left[
          S_l \ge 10 a_1 \log m
        \right]
        &\le
        \Pr\left[
          S_l \ge \E[S_l] \cdot 2
        \right] \\
        &\le
        \exp\left(
        - \frac{1}{3} \cdot
        5 a_1 \log m 
        \right) \\
        &\le 
        \exp\left(
        - 0.8 a_1 N
        \right) \\
        & \le N^{-0.8a_1},
    \end{align}
    where the penultimate inequality holds since $2\log m \ge \log N$
    when $m = cn$ for $c \ge 1$. 
    Thus $S_l \le 10 a_1 \log m$ with probability all but $N^{-0.8a_1}$. \\
    
    \noindent
    \textit{Bound on $S_{bx} - S_{by}$}:
    Recall that we wish to show $S_{bx} - S_{by} \ge 22 a_1 \log m$
    throughout the sequence of $900 a_1 \log m$ productive steps. 
    Letting $\lambda$ denote the number of blank-consuming productive
    steps throughout this sequence,
    by Lemma~\ref{lemma:dbcl-util-prod-prodb} it follows that
    $$
    290a_1 \log m \le \lambda \le 900 a_1 \log m.
    $$
    Given that $\lambda = S_{bx} + S_{by}$, showing that
    $S_{by} > \lambda/2 - 11 a_1 \log m$ is sufficient to ensure
    that $S_{bx} - S_{by} \ge 22 a_1 \log m$ throughout the sequence. 
    
    To show this bound on $S_{by}$ holds, we first bound the 
    probability of a $Y+B$ or $I_Y + B$ interaction conditioned on
    any blank-consuming productive step, which we denote by $p(by)$. 
    We have
    \begin{align}
        p(by) = \frac{by + bi_y}{bi_x + bi_y + bx + by}
        &= \frac{y + n/2}{x + y + n} \\
        &\le \frac{y + b/2 + n/2}{x + y + b + n}
        = \frac{\hat y + n/2}{N},
    \end{align}
    where the inequality holds when $x \ge y$ and given that
    $i_y < n/2$. 
    Assuming that $\hat y = O(\log m)$ throughout the sequence, 
    then $\hat y \le 0.05 m$ holds for sufficiently large $m$, and thus
    \begin{align}
        p(by) 
        &\le \frac{0.05m + 0.5 n}{m + n} \\
        &\le \frac{0.05 cn + 0.5 n}{(1+c)n} \\
        &< 0.14,
    \end{align}
    where the final inequality comes from the assumption that $c \ge 4$. 
    
    Thus throughout the subsequence of $\lambda$ blank-consuming productive
    steps, we have in expectation that
    $\E[S_{by}] < 0.14 \lambda$. Now using an upper Chernoff bound shows
    that 
    \begin{align}
        \Pr\left[S_{by} \ge (0.5\lambda - 11a_1\log m) \right] 
        &=
        \Pr\left[
        S_{by} \ge \E[S_{by}] 
        \left(1 + \frac{1}{\E[S_{by}]}\left(0.36\lambda - 11a_1\log m\right)\right)
        \right] \\
        &\le
        \exp\left(
        - \frac{1}{3} \cdot
         (0.36 \lambda - 11a_1 \log m)^2 \cdot \frac{7}{\lambda}
        \right) \\
        &\le
        \exp\left(
        - \frac{1}{3} \cdot 
        (93 a_1 \log m)^2 \cdot \frac{7}{900 a_1 \log m}
        \right) \\
        &\le 
        \exp\left(
        - 22 a_1 \log m 
        \right) \\
        &\le
        N^{-11a_1},
    \end{align}
    where we use the fact that $290 a_1 \log m \le \lambda \le 900 a_1 \log m$. 
    So $S_{by} > 0.5 \lambda - 11 a_1 \log m$ and thus 
    $S_{bx} - S_{by} \ge 22a_1 \log m$ throughout the sequence as required
    with probability all but $N^{-11a_1}$. 
    
    Taking a union bound over the error probabilities associated with
    the bounds on $S_{l}$ and $S_{by}$ then gives the
    stated claims with probability at least $1 - 2N^{-0.8a_1}$. 
\end{proof}

The protocol has a similar behavior in the case of \textit{large}
leak rate, which is given by the following lemma. Roughly,
this lemma says that for large $\beta$, the value of $\hat y$ 
remains at most $O(\beta m)$ over a sequence of $O(\beta m)$ productive
steps following the completion of Phase 2. The proof of
the lemma follows nearly identically to that of the previous lemma
and is thus omitted. 
\begin{lemma}
 \label{lemma:dbcl-phase3-helper-lb}
 Suppose $\hat y = 301 a_1 \beta m$ for some $a_1 \ge a$ following the 
 end of Phase 2 of the $\dbamc$ protocol with large leak rate 
 $\omega(log N/N) \le \beta \le (a \sqrt{N \log N})/N$
 on a population where $m = cn$ for $c \ge 4$. 
 Consider a sequence of $900 a_1 \beta m $ productive steps. Then
 \begin{enumerate}
     \item 
     the maximum value of $\hat y$ throughout the sequence is $761 a_1 \beta m$
     \item
     $\hat y$ decreases to $300a_1 \beta m$ by the end of the sequence
 \end{enumerate}
 both with probability at least $1 - 2N^{-0.8a_1}$ when $N$ is sufficiently large. 
\end{lemma}

Using the previous two lemmas, we can now state the following theorem
(analogous to Theorem~\ref{thm:dbl-phase3-convergence})
characterizing the long-term behavior of the protocol with leaks. 
The theorem says that
following the completion of Phase 2 with small or 
large leak rate, 
the protocol remains in a low-sample-error configuration for at
least a polynomial number of steps with high probability. 

\begin{theorem}
  \label{thm:dbcl-phase3-convergence}
  Consider an execution of the $\dbamc$ protocol 
  following the successful completion of Phase 2 of the protocol
  on a population with $m = cn$ for $c \ge 4$. 

  \noindent
  Then given $\alpha_1 > \alpha$ and $0 < c < 0.8\alpha$,
  there is some $a > 0$ such that 
  \begin{enumerate}[i.]
    \item
    $\hat y \le 761 \alpha_1 \beta m$ when
    $\omega(\log N / N) \le \beta \le (\alpha \sqrt{N \log N})/N$
    \item
    $\hat y \le 761 \alpha_1 \log m$ when $\beta \le \alpha \log N/N$
  \end{enumerate}   
  holds for at least the next $N^a$ steps
  with probability at least $1 - N^{-c}$ when $N$ is sufficiently large. 
\end{theorem}

The proof of the theorem follows similarly to that of 
Theorem~\ref{thm:dbl-phase3-convergence} and is omitted. 




\paragraph*{Simulation Results}
Note that as the upper bound on the leak rate $\beta$ is a decreasing 
function in $N$, the sample error guarantees of both protocols increase
with population size. 
This relationship is shown across various simulations of the $\dbamc$
protocol in Figure \ref{fig:simulations}. Moreover, Figure \ref{subfig:success} 
depicts aggregate sample data over many executions of the \dbamc{} protocol 
for varying values of $N$, and Figure \ref{subfig:runtime} illustrates the logarithmic parallel time needed to reach convergence in the non-leak setting.

\begin{figure}[htb!]
    \centering
    \includegraphics[width=6in]{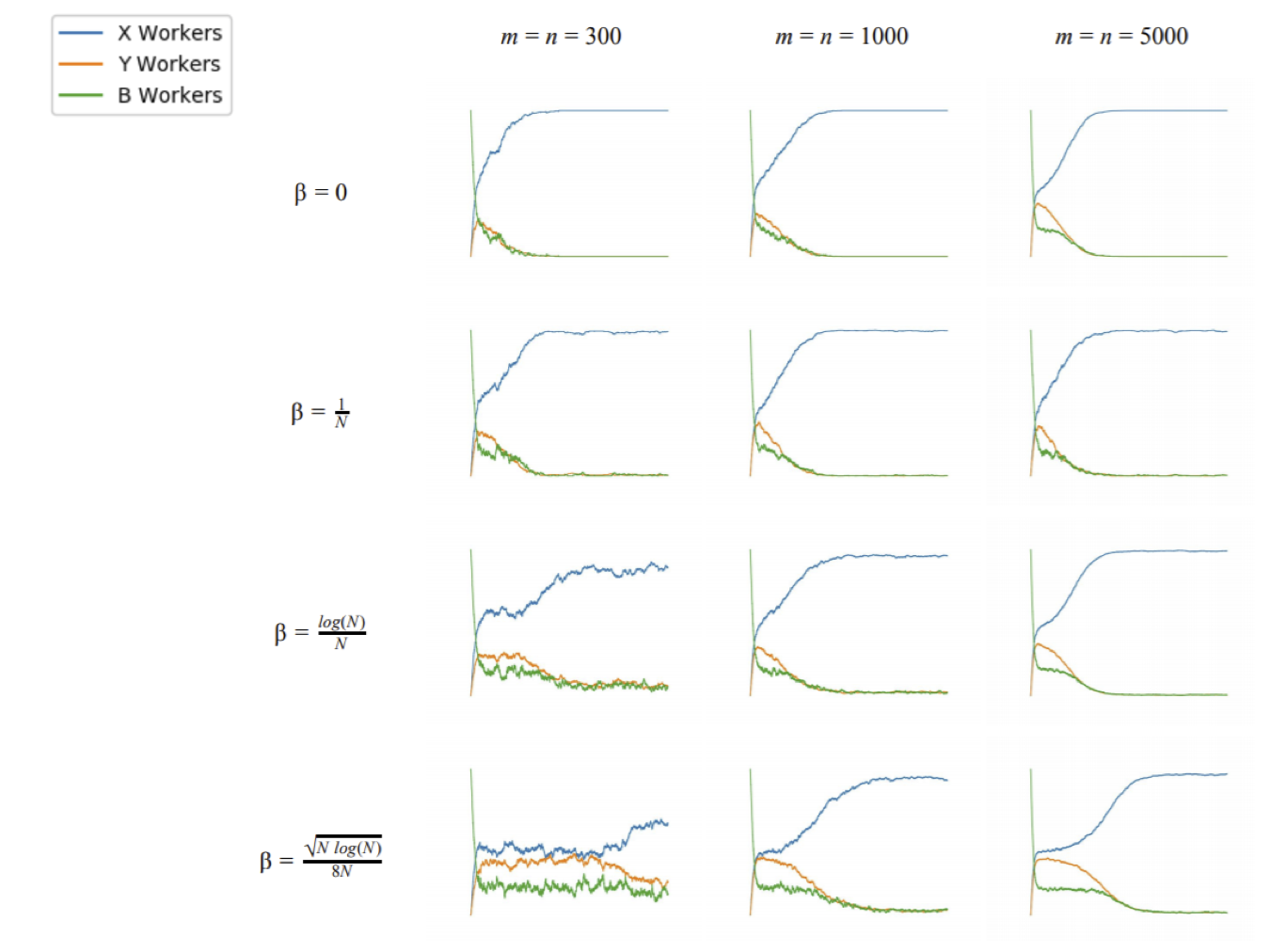}
    \caption{
    Simulations of the $\dbamc$ protocol, where each subplot
    shows how the count of $X$, $Y$ and $B$ worker agents
    evolve over the course of an execution.
    All simulations are for $m=n$, input
    margin $\epsilon = \sqrt{N \log N}$ (with an $I_X$ majority), 
    and varying values of
    leak rate $\beta$ over $4N\log N$ total interactions.
    Note that these plots are of single executions and thus provide a qualitative illustration of behavior, rather than statistically significant data. However, we can 
    see that for larger values of $m = n$ and smaller
    values of $\beta$, the number of $X$ worker agents 
    reaches a larger count more quickly.
    }
    \label{fig:simulations}
\end{figure}

\begin{figure*}[htb!]
    \centering
    \begin{subfigure}[b]{2.65in}
    \includegraphics[width=2.65in]{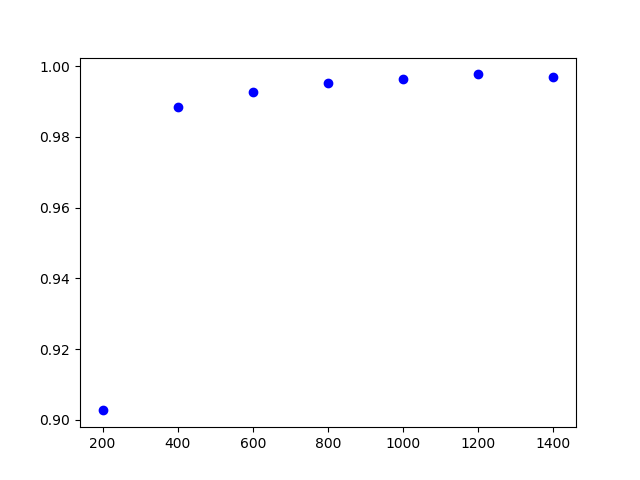}
    \caption{
    Sample success rate ($y$-axis) averaged over 3000 executions of \dbamc\ for $n=600$, $\Delta_0 = \sqrt{N \log N}$, and $\beta=1/N$, and varying values of $m$ ($x$-axis). Samples were drawn uniformly from the worker population after $4N \log N$ total interactions.}
    \label{subfig:success}
    \end{subfigure}
    ~
    \begin{subfigure}[b]{2.65in}
    \includegraphics[width=2.65in]{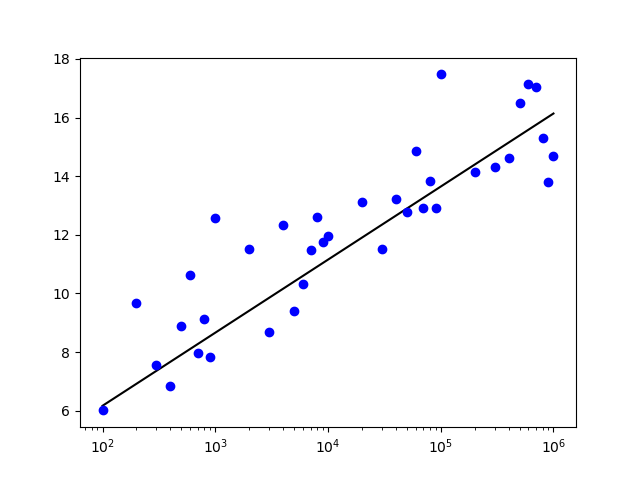}
    \caption{
    Parallel time ($y$-axis) for \dbamc\ without leaks to reach consensus for varying population sizes ($x$-axis) where $m=2n$. Data points represent a single execution. The solid black line is $\frac{3}{4} \log_2(N)$, showing that convergence takes $O(N \log N)$ interactions.}
    \label{subfig:runtime}
    \end{subfigure}
    \caption{Success rate and running time of \dbamc\ over various executions of the protocol.}
    \label{fig:success-and-runtime}
\end{figure*}



\section{Leaks Versus Byzantine agents}
\label{sec:byzantine-leaks-equiv}

The original third-state dynamics approximate majority protocol \cite{angluin2008simple} is robust to a bounded number of Byzantine agents, and as shown in the previous
sections, both the
$\dbam$ protocol and the $\dbamc$ protocol in the CI model are robust to a bounded leak rate.
In this section, we consider the connection between these two
types of faulty behavior. 
While leaks can occur at any agent with fixed probability throughout an execution, Byzantine agents are a fixed subset of the population, and while a leak event
does not change the subsequent behavior of an agent, Byzantine agents may continue to
misbehave forever.
However, there are parallels between these two models of adversarial behavior. A leak at one agent can cause additional agents to deviate from a convergent configuration; similarly, interactions among non-Byzantine agents, some of which have deviated from a convergent configuration by interacting with a Byzantine agent, can cause additional non-Byzantine agents to diverge. 

\paragraph*{Equivalence of Byzantine Agents and Leaks in \dbam\ and \dbamc}

We prove that for the \dbam{} and \dbamc{} protocols, introducing a leak rate of $\beta$ 
has the same asymptotic effect as introducing $O(\beta N)$ Byzantine agents to the population,
which demonstrates an equivalence between these two notions of adversarial behavior among the class of third-state dynamics
protocols.
Although the results of the previous section assumed leaks that do
not follow the laws of chemistry, the following result 
considers \concept{weak leaks}, which cause the selected agent to
decrease its confidence in the majority value by one degree (i.e. a leak causes an agent in state $X$ to transition to $B$ and an agent in state $B$ to transition to $Y$, matching the $X + Y$ and $Y + B$ transitions). 

For our purposes, we define two adversarial models $\mathcal{M}_1$ and $\mathcal{M}_2$ to be \concept{equivalent} for some protocol $\mathcal{P}$ if $\mathcal{P}$ converges to the same asymptotic sample error rate in the same asymptotic running time in both models. We
then have the following equivalence result:


\begin{restatable}{theorem}{thmbyz}
\label{thm:byz-leaks-equiv}
A population of $N$ agents running \dbam{} (or \dbamc) with weak leak rate $O(\beta)$ is equivalent to a population of $N+B$ agents, where $B=O(N\beta)$ agents are Byzantine, running \dbam{} (or \dbamc) without leaks, where in either setting the protocol converges in $O(N \log N)$ interactions with error probability $O(\beta)$.
\end{restatable}



\begin{proof}
Below we will refer to the running protocol as $\mathcal{P}$, referring to either \dbam{} or \dbamc. Though different from the notation used in previous sections, we will refer to the total number of agents in the population (catalytic or otherwise) as $N$. We will prove the equivalence of these adversarial behaviors in two parts.

First we must show that running $\mathcal{P}$ among $N$ agents with weak leak rate $\beta$ has the same asymptotic error rate and runtime as running $\mathcal{P}$ in a population of $N+B$ agents \textit{without leaks}, where $B=O(N\beta)$ is the number of Byzantine agents. We can do this by imagining that the $N$ honest agents are the entire population executing $\mathcal{P}$. At each step in time, the scheduler selects two agents from the population to interact. When any two of the $N$ honest agents interact, the interaction is indistinguishable from the case where there are $N$ total agents in the non-Byzantine setting. However, with some probability $p$, one Byzantine and one non-Byzantine agent are selected to interact with one another, potentially causing the non-Byzantine agent to diverge from a convergent state. In the \dbam{} and \dbamc{} protocols, this would mean that the Byzantine agent could just stay in state $Y$, causing catalysts and $Y$-agents to stay the same (as needed), $B$-agents to become $Y$-agents, and $X$-agents to become $B$-agents. The probability of this cross-interaction is
\begin{align}
     p = \frac{NB}{\binom{N+B}{2}}
     \leq \frac{N(cN\beta)}{\binom{N+cN\beta}{2}}
     &\leq \frac{2N(cN\beta)}{N^2} \\
     &= 2c\beta\\
     &= O(\beta)
\end{align}

We define an \concept{effective interaction} to be any interaction that produces a non-null state transition. Here, an effective interaction is any interaction between honest agents and any simulated leaks, i.e. interactions between honest and Byzantine agents.
The runtime of the simulation is the number of total steps it takes to have sufficiently many effective interactions to successfully complete the protocol among the $N$ honest agents. The probability of an effective interaction is
\begin{align}
    \frac{N(N-1)/2+NB}{\binom{N+B}{2}}
    &=\frac{N(N-1)+2NB}{(N+B)(N+B-1)} \\
    &=\frac{N(N+2B-1)}{(N+B)(N+B-1)}
\end{align} 

So if the runtime of $\mathcal{P}$ on $N$ honest agents with leaks is $O(N \log N)$ total interactions, then the simulation must run for $O((N+B)\log N)$ steps so that the expected fraction of ``productive'' interactions is \begin{align}
    O\left(\frac{N(N+2B-1)}{(N+B)(N+B-1)} (N+B)\log N\right) 
    &\leq O\left(\frac{N(2N+2B-2)}{N+B-1} \log N\right) \\
    &= O(N \log N)
\end{align}
For $B=O(N\beta)$ and $\beta \leq 1$, $(N+B)\log N = N \log N + O(N \beta \log N) = O(N \log N)$.

Next we must show that running $\mathcal{P}$ with $N+B$ agents, where $B=O(N\beta)$ agents are Byzantine, has the same asymptotic error rate and runtime as running $\mathcal{P}$ in a population of $N$ agents with leak rate $\beta$. In the Byzantine setting, we only care about the behavior of the honest nodes. Therefore, as before, we can imagine that the $N$ agents in the population are equivalent to the $N$ honest agents in the population we are trying to simulate and now suppose that there are an additional $O(N\beta)$ Byzantine agents. Assuming that the Byzantine agents are optimally adversarial (always in the $Y$ state), then weak leaks among the honest agents directly simulate interactions with Byzantine agents. Also as before, the total runtime of the simulation will be $O(N \log N) = O((N+B) \log (N+B))$ because we are omitting interactions among Byzantine agents entirely, yielding fewer total interactions to simulate the same behavior in the population of $N+B$.
\end{proof}

\paragraph*{Super-Adversarial Byzantine Agents}

In Theorem \ref{thm:byz-leaks-equiv}, we assumed that leaks are not fully adversarial (converting $X$ to $Y$), but rather only decrease the confidence in the majority value by one degree. However, our analysis in previous sections assumes fully adversarial leaks in order to demonstrate that in the worst possible case (i.e. for the maximum decrease in the progress measure), we still succeed in computing approximate majority up to leak rate $\beta$. The equivalent to this in the Byzantine model would be to add a transition to \dbam{} or \dbamc{} of the form $T + \_ \rightarrow T + Y$,
where $T$ is a special state held only by \concept{super-adversarial Byzantine agents}, and $\_$ is a wildcard representing any non-catalytic state. This new state transition indicates that interacting with a Byzantine agent causes any non-catalytic agent to shift into the $Y$ state, exactly modeling the fully adversarial leaks described in earlier sections.

We observe that the proof of Theorem \ref{thm:byz-leaks-equiv} can be repurposed to demonstrate an equivalence between the stronger notion of fully adversarial leaks used in earlier sections and the modified protocols defined above for super-adversarial Byzantine agents.

\section{Conclusion and Open Problems}

We have shown that third-state dynamics can be used to solve approximate
majority with high probability in $O(n \log n)$ steps up to leak rate 
$\beta = O(\sqrt{n \log n}/n)$, both in the standard population protocol model 
as well as the CI model when $m=\Theta(n)$.
While we showed a separation between
the CI and original population models, it remains an open question
what other problems (similar to approximate majority) can be
computed quickly in the CI model. 
Additionally, identifying which 
families of protocols are naturally robust to leak events in 
the original population model (similar to third-state dynamics) 
also remains an open question.

\paragraph*{Acknowledgements}
The authors would like to thank Anne Condon, Monir Hajiaghayi, 
David Kirkpatrick, and J{\'a}n Ma{\v{n}}uch for a helpful discussion 
regarding the analysis of the $\dbam$ protocol. 
The authors are also grateful for a discussion with 
Francesco d’Amore, Andrea Clementi, and 
Emanuele Natale, who pointed out the connection between 
catalytic agents in population protocols and stubborn agents in 
other types of multi-agent systems. We also thank the
anonymous reviewers for their helpful feedback.

\bibliographystyle{alpha}
\bibliography{references.bib}

\appendix

\end{document}